\documentclass[a4paper,UKenglish,cleveref, autoref, thm-restate,numberwithinsect]{lipics-v2021}

\usepackage[T1]{fontenc}
\usepackage{graphicx}

\usepackage{lmodern}
\usepackage[utf8]{inputenc}
\usepackage[T1]{fontenc}
\usepackage{graphicx}
\usepackage{longtable}
\usepackage{rotating}
\usepackage[normalem]{ulem}
\usepackage{textcomp}
\usepackage{version}
\usepackage{hyperref}
\usepackage{amsmath}
\usepackage{amsfonts}
\usepackage{amssymb,dsfont}
\usepackage{amsthm}
\usepackage{xcolor}
\usepackage[inline, shortlabels]{enumitem}
\usepackage{appendix}
\usepackage{stackengine}
\usepackage{tikz}
\usetikzlibrary{automata,positioning}
\usetikzlibrary{arrows,automata}
\usetikzlibrary{shapes,snakes}
\usetikzlibrary{fit}
\usepackage{caption}
\usepackage{subcaption}
\usepackage{wasysym}
\usepackage{ stmaryrd }
\usepackage{float}
\usepackage{environ}
\usepackage{float}
\usepackage{booktabs}%commande toprule, midrule...

\usepackage[linesnumbered,lined,boxed,commentsnumbered]{algorithm2e}
\usepackage{algpseudocode}
\usepackage{mathtools}
\usepackage{multicol}

\usepackage{ifpdf}
\ifpdf
\usepackage[framemethod=TikZ]{mdframed}
\usepackage[textsize=small]{todonotes}
\newcommand{\lug}[1]{}
\newcommand{\nas}[1]{}
\newcommand{\ars}[1]{}
\newcommand{\lugtext}[1]{}
\newcommand{\nastext}[1]{}
\newcommand{\arstext}[1]{}
\fi

\usetikzlibrary{calc,decorations.pathmorphing,shapes}

\newcounter{sarrow}

\newcommand\xrsquigarrow[1]{%
	\stepcounter{sarrow}%
	\mathrel{\begin{tikzpicture}[baseline= {( $ (current bounding box.south) + (0,-0.5ex) $ )}]
			\node[inner sep=.5ex] (\thesarrow) {$\scriptstyle #1$};
			\path[draw,<-,decorate,
			decoration={zigzag,amplitude=0.7pt,segment length=1.2mm,pre=lineto,pre length=4pt}] 
			(\thesarrow.south east) -- (\thesarrow.south west);
	\end{tikzpicture}}%
}

\DeclareFontFamily{U}{MnSymbolC}{}
\DeclareFontShape{U}{MnSymbolC}{m}{n}{
	<-6> 	MnSymbolC5 <6-7>	MnSymbolC6	<7-8>		MnSymbolC7
	<8-9>  MnSymbolC8 <9-10>  MnSymbolC9  <10-12> 	MnSymbolC10
	<12->	MnSymbolC12}{}
\DeclareFontShape{U}{MnSymbolC}{b}{n}{
	<-6> 	MnSymbolC-Bold5 <6-7>	MnSymbolC-Bold6	<7-8>		MnSymbolC-Bold7
	<8-9>  MnSymbolC-Bold8 <9-10>  MnSymbolC-Bold9  <10-12> 	MnSymbolC-Bold10
	<12->	MnSymbolC-Bold12}{}
\DeclareSymbolFont{MnSyC}{U}{MnSymbolC}{m}{n}

\DeclareFontFamily{U}{MnSymbolD}{}
\DeclareFontShape{U}{MnSymbolD}{m}{n}{
	<-6> 	MnSymbolD5 <6-7>	MnSymbolD6	<7-8>		MnSymbolD7
	<8-9>  MnSymbolD8 <9-10>  MnSymbolD9  <10-12> 	MnSymbolD10
	<12->	MnSymbolD12}{}
\DeclareFontShape{U}{MnSymbolD}{b}{n}{
<-6> 	MnSymbolD-Bold5 <6-7>	MnSymbolD-Bold6	<7-8>		MnSymbolD-Bold7
<8-9>  MnSymbolD-Bold8 <9-10>  MnSymbolD-Bold9  <10-12> 	MnSymbolD-Bold10
<12->	MnSymbolD-Bold12}{}
\DeclareSymbolFont{MnSyD}{U}{MnSymbolD}{m}{n}

\DeclareMathSymbol{\otherPlus}{\mathrel}{MnSyC}{20}
\DeclareMathSymbol{\otherMinus}{\mathrel}{MnSyC}{16}
\DeclareMathSymbol{\otherEquality}{\mathrel}{MnSyD}{0}

\protected\def\verythinspace{%
\ifmmode \mskip0.5\thinmuskip \else \ifhmode \kern0.08334em \fi \fi}
\newcommand {\inc}[1]{\ensuremath{#1{\otherPlus}}}
\newcommand {\dec}[1]{\ensuremath{#1{\otherMinus}}}
\newcommand {\testz}[1]{\ensuremath{#1\verythinspace{\otherEquality}\verythinspace0}}
\newcommand {\nbdec}[1]{\ensuremath{nb(#1{\otherMinus})}}

\newenvironment{sketch-proof} {\textit{Sketch of proof.}}{\hfill$\square$\\ }

\newcommand{\PP}{\ensuremath{\mathcal{P}}}
\newcommand{\CC}{\ensuremath{\mathcal{C}}}

\newcommand{\Init}{\ensuremath{\mathcal{I}}}
\newcommand{\FinalE}{\ensuremath{\mathcal{F}_{\exists}}}
\newcommand{\FinalA}{\ensuremath{\mathcal{F}_{\forall}}}

\newcommand{\Exec}{\ensuremath{\textsc{Exec}}}

\newcommand{\Read}[1]{\ensuremath{R(#1)}}

\newcommand{\nop}{\ensuremath{\bot}}
\newcommand{\Counters}{X}
\newcommand{\nat}{\mathbb{N}}
\newcommand{\set}[1]{\{#1\}}
\newcommand{\mset}[1]{\Lbag #1 \Rbag}

\newcommand{\nb}[1]{\ensuremath{\mathbf{nb}(#1)}}
\makeatletter
\newcommand*{\da@rightarrow}{\mathchar"0\hexnumber@\symAMSa 4B }
\newcommand*{\da@leftarrow}{\mathchar"0\hexnumber@\symAMSa 4C }
\newcommand*{\xdashrightarrow}[2][]{%
  \mathrel{%
    \mathpalette{\da@xarrow{#1}{#2}{}\da@rightarrow{\,}{}}{}%
  }%
}
\newcommand*{\da@xarrow}[7]{%
  \sbox0{$\ifx#7\scriptstyle\scriptscriptstyle\else\scriptstyle\fi#5#1#6\m@th$}%
  \sbox2{$\ifx#7\scriptstyle\scriptscriptstyle\else\scriptstyle\fi#5#2#6\m@th$}%
  \sbox4{$#7\dabar@\m@th$}%
  \dimen@=\wd0 %
  \ifdim\wd2 >\dimen@
    \dimen@=\wd2 %   
  \fi
  \count@=2 %
  \def\da@bars{\dabar@\dabar@}%
  \@whiledim\count@\wd4<\dimen@\do{%
    \advance\count@\@ne
    \expandafter\def\expandafter\da@bars\expandafter{%
      \da@bars
      \dabar@ 
    }%
  }%  
  \mathrel{#3}%
  \mathrel{%   
    \mathop{\da@bars}\limits
    \ifx\\#1\\%
    \else
      _{\copy0}%
    \fi
    \ifx\\#2\\%
    \else
      ^{\copy2}%
    \fi
  }%   
  \mathrel{#4}%
}
\makeatother

\newcommand{\qinit}{\ensuremath{q_\textit{in}}}
\newcommand{\ellinit}{\ensuremath{\ell_\textit{in}}}
\newcommand\stackrqarrow[1]{\overset{#1}{\rightsquigarrow}}

\newcommand{\arrowa}[1]{\ensuremath{\Rightarrow}}
\newcommand{\arrowP}[1]{\ensuremath{\xrightarrow{}_{#1}}}
\newcommand{\arrowPlab}[2]{\ensuremath{\xrightarrow{#2}_{#1}}}

\newcommand{\arrowPa}[1]{\ensuremath{\Rightarrow_\PP}}

\newcommand{\Ptime}{\textsc{Ptime}}
\newcommand{\Expspace}{\textsc{Expspace}}
\newcommand{\Loc}{\ensuremath{\text{Loc}}}
\newcommand{\cpt}{\mathtt{x}}
\newcommand{\CAct}[1]{\mathsf{CAct}(#1)}
\newcommand{\Cover}{\ensuremath{\textsc{SCover}}}
\newcommand{\CCover}{\ensuremath{\textsc{CCover}}}
\newcommand{\CMCover}{\ensuremath{\textsc{Cover}}}
\newcommand{\Target}{\ensuremath{\textsc{Synchro}}}

\newcommand{\Term}{\ensuremath{\textsc{Term}}}

\newcommand{\NRCM}{\text{\normalfont{NB-R-CM}}}
\newcommand{\NBCM}{\text{\normalfont{NB-CM}}}

\newcommand{\NBVAS}{\text{\normalfont{NB-VAS}}}

\newcommand{\transVas}{\ensuremath{\rightsquigarrow}}
\newcommand{\transVasZ}{\ensuremath{\rightharpoondown}}
\newcommand{\xtransVasZ}[1]{\ensuremath{\xrightharpoondown[]{#1}}}
\newcommand{\xtransVas}[1]{\ensuremath{\stackrqarrow{#1}}}
\newcommand{\transCM}{\ensuremath{\rightsquigarrow}}
\newcommand{\transNbCM}{\ensuremath{\rightsquigarrow}}
\newcommand{\xtransNbCM}[1]{\ensuremath{\ensuremath{\stackrqarrow{#1}}}}
\newcommand{\xtransCM}[1]{\ensuremath{\stackrqarrow{#1}}}

\newcommand{\testfreeCM}{test-free CM}

\newcommand{\iflong}[1]{#1}

\newcommand{\constantM}{\ensuremath{N}}

\newcommand{\Toks}{\ensuremath{\textit{Toks}}}

\newcommand{\Rec}[1]{\text{Rec}(#1)}

\newcommand{\mst}{\mathsf{st}}

\newcommand{\qone}{\ensuremath{q_1}}
\newcommand{\qtwo}{\ensuremath{q_2}}
\newcommand{\qthree}{\ensuremath{q_3}}
\newcommand{\qfour}{\ensuremath{q_4}}
\newcommand{\qfive}{\ensuremath{q_5}}
\newcommand{\qsix}{\ensuremath{q_6}}

\newcommand{\cpty}{\ensuremath{\mathtt{y}}}
\newcommand{\cptz}{\ensuremath{\mathtt{z}}}
\newcommand{\cpts}{\ensuremath{\mathtt{s}}}

\newcommand{\Interp}[1]{\llbracket #1 \rrbracket}
\newcommand{\ifccover}[1]{#1}
\newcommand{\ifterm}[1]{}

\newcommand{\starg}[1]{\mst(\mathit{{\kern-1pt}#1})}

\makeatletter
\newcommand{\problemtitle}[1]{\gdef\@problemtitle{#1}}% Store problem title
\newcommand{\probleminput}[1]{\gdef\@probleminput{#1}}% Store problem input
\newcommand{\problemquestion}[1]{\gdef\@problemquestion{#1}}% Store problem question
\NewEnviron{decproblem}{
  \problemtitle{}\probleminput{}\problemquestion{}% Default input is empty
  \BODY% Parse input
  \par\addvspace{0.8\baselineskip}
  \noindent
  \fbox{\centering
  \begin{tabular}{rl}
    \multicolumn{2}{l}{\normalsize \@problemtitle} \\% Title
    \hline
    \rule{0pt}{3ex}
    \normalsize \textbf{Input:} & \normalsize \@probleminput \\[0.5ex]% Input
    \normalsize \textbf{Question:} & \normalsize \@problemquestion% Question
  \end{tabular}}
  \par\addvspace{0.8\baselineskip}
}

\title{Safety Analysis of Parameterised Networks with Non-Blocking Rendez-Vous}

\author{Lucie Guillou}{IRIF, CNRS, Universit\'e Paris Cit\'e,
  France}{}{}{}

\author{Arnaud Sangnier}{IRIF, CNRS, Universit\'e Paris Cit\'e,
  France}{}{}{}

\author{Nathalie Sznajder}{LIP6, CNRS, Sorbonne Universit\'e,
  France}{}{}{}

\authorrunning{L. Guillou and A. Sangnier and N. Sznajder}

\Copyright{L. Guillou and A. Sangnier and N. Sznajder}



\relatedversion{}

\EventEditors{}
\EventNoEds{0}
\EventLongTitle{}
\EventShortTitle{}
\EventAcronym{}
\EventYear{2023}
\EventDate{}
\EventLocation{}
\EventLogo{}
\SeriesVolume{}
\ArticleNo{}
\newif\iflong\longtrue\longfalse
\newif\iftable\tabletrue%\tablefalse

\nolinenumbers
\begin{document}

\maketitle            
 \begin{abstract}
We consider  networks of processes that all execute the same
finite-state protocol and communicate via a rendez-vous
mechanism. When a process requests a rendez-vous, another process can
    respond to it and they both change their control states
    accordingly. We focus here on a specific semantics, called
    non-blocking, where the process requesting a rendez-vous can
    change its state even if no process can respond to it. In this context, we study the
    parameterised coverability problem of a configuration, which consists in determining whether there is an initial
    number of processes and an execution allowing to reach a
    configuration bigger than a given one. We show
    that this problem is EXPSPACE-complete and can be solved in polynomial time if the protocol is 
    partitioned into two sets of states,
    the states from which a process can request a rendez-vous and the
    ones from which it can answer one. We also prove that the
    problem of the existence of an execution bringing all the
    processes in a final state is undecidable in our context. These
    two problems can be solved in polynomial time with the classical rendez-vous semantics.
\end{abstract}

 \section{Introduction}

\emph{Verification of distributed/concurrent systems.} Because of their ubiquitous use in applications we rely on constantly, the development of formal methods to guarantee the correct behaviour of distributed/concurrent systems has become one of the most important research directions in the field of computer systems verification in the last two decades. Unfortunately, such systems are difficult to analyse for several reasons. Among others, we can highlight two aspects that make the verification process tedious. First, these systems often generate a large number of different executions due to the various interleavings generated by
 the concurrent behaviours of the entities involved. Understanding how these interleavings interact is a complex task which can often lead to errors at the design-level or make the model of these systems very complex. Second, in some cases, the number of participants in a distributed system may be unbounded and not known a priori. To fully guarantee the correctness of such systems, the analysis would have to be performed for all possible instances of the system, i.e., an infinite number
  of times. 
 % As a consequence, one may have to verify the behavior of such
%systems for different instantiations of this number and
%in order to fully guarantee the correctness of the systems, it would
%be required to perform such an analysis for all the possible
%instantiations, i.e. an infinite number of times, which is in practice
%not feasible. 
%
%phrase suivante pas indispensable à mon avis
%Of course, other difficulties raise if one assumes that the
%processes manipulate data belonging to an  infinite domain or if they
%use complex communication mechanisms. 
As a consequence, classical
techniques to verify finite state systems, like testing or
model-checking, cannot be easily adapted to distributed systems and
it is often necessary to develop new techniques.

\noindent\emph{Parameterised verification.} When designing systems with an
unbounded number of participants, one often provides a schematic program (or protocol)
intended to be implemented by multiple identical processes, parameterised by
the number of participants. In general, even if the verification problem is
decidable for a given instance of the parameter, verifying all possible instances
is undecidable~(\cite{apt86limits}).
% in a distributed systems, one method
%consists in considering models where each participant executes the
%same program or protocol and where the number of participants is  a parameter to
%be instantiated.
% The goal being then to ensure the correctness of the
%system for all possible instantiations. In general, even if given an
%instantiation the verification problem is decidable, establishing this
%property for all the cases is an undecidable problem as shown in
%\cite{apt86limits}. 
However, several settings come into play that can be 
 adjusted to allow automatic verification.
%adjusting  a bit the context allows in
%some cases to design automatic verification techniques. 
One key aspect
to obtain decidability  is to assume that the processes do not manipulate
identities \iflong in the protocols\fi and use simple communication mechanisms like
pairwise synchronisation (or rendez-vous) \cite{german92}, broadcast of
a message to all the entities \cite{esparza-verif-lics99} (which can as well be lossy in order to
simulate mobility \cite{delzanno-complexity-fsttcs12}), shared register containing values of a finite
set \cite{esparza-param-cav13}, and so on (see~\cite{esparza14} for a survey).  
In every aforementioned case, all the entities execute
the same protocol given by a finite state automaton.
% whose edges are
%labeled with some communication actions. 
Note that parameterised verification, when decidable like in the above models, is
also sometimes surprisingly easy, compared to the same problem with a fixed number of participants.
%with a complexity lower than the same problem with a fixed
%number of participants. 
%these simple
%contexts not only allow to obtain decidability results for the parameterised
%verification, but surprisingly in some cases the
%complexity of verification process is relatively low and in some
%cases lower than the problem with a fixed number of participants. 
For
instance, liveness verification  of parameterised systems with
shared memory is \textsc{Pspace}-complete for  a fixed number of
processes and in \textsc{NP} when parameterised~%this number is a parameter
\cite{durand-model-fmsd17}.

\noindent\emph{Considering rendez-vous communication.} In one
of the seminal papers for
the verification of parameterised networks~\cite{german92}, German and Sistla 
(and since then \cite{bala21finding,horn20deciding})
assume that the entities communicate by ``rendez-vous'', 
a synchronisation mechanism in which two processes (the \emph{sender} and the \emph{receiver}) 
agree on a common action by which they jointly change their local state. 
This mechanism is synchronous and symmetric, meaning that if no process is ready to receive a message, the sender cannot send it. 
However, in some applications, such as Java Thread programming, this is not exactly the primitive that is implemented. When a Thread is suspended in a waiting state, it is woken up by the reception of a message \texttt{notify} sent by another Thread. However, the sender is not blocked if there is no suspended Thread waiting for 
its message; in this case, the sender sends the \texttt{notify} anyway and the message is simply lost.
%
%communicate thanks to rendez-vous: the
%communication is pairwise and at any instant if one process requests
%a rendez-vous (the set of possible rendez-vous being given by a finite
%alphabet) and another process is ready to accept the same
%rendez-vous, then they both change simultaneously their control state
%according to the transition relation of the protocol. 
%One of the
%justification to consider this rendez-vous mechanism is that it
%corresponds quite well to some synchronisation mechanisms used in
%concurrent program languages. Indeed, in \textsc{Java} for instance, a thread can
%wait on a object, using the instruction \texttt{wait} and another
%thread will be able to wake up such a thread using the instruction
%\texttt{notify} on the same object. 
%However as noticed in
%\cite{delzanno-towards-tacas02}, the rendez-vous  does not
%faithfully correspond to the semantics of such  mechanism, as a thread
%performing a \texttt{notify} will wake up a waited process if there is
%one, but if it is not the case, it will anyway perform the instruction
%\texttt{wait} and not be blocked. 
This is the reason why Delzanno et. al.
have introduced \emph{non-blocking}
rendez-vous in~\cite{delzanno-towards-tacas02}
 a communication primitive in which the sender of a message is not blocked
 if no process receives it. One of the problems of interest in
parameterised verification is the coverability problem: is it
possible that, starting from an initial configuration, (at least)
one process reaches a bad state?
%, where a process requesting a rendez-vous will either move
%together with another process accepting the rendez-vous or will move
%alone in there is no such process.
In \cite{delzanno-towards-tacas02}, and later in~\cite{raskin03petrinets},
the authors introduce variants of Petri nets to handle this type
of communication. 
In particular, the authors investigate in~\cite{raskin03petrinets}
the coverability problem for an extended class of Petri nets
with
non-blocking arcs, and show that for this model the coverability
problem is decidable using the techniques of Well-Structured
Transitions Systems
\cite{abdulla-general-lics96,abdulla-algorithmic-ic00,finkel-well-tcs01}.
However, since
their model is an extension of Petri nets, the latter problem is
\Expspace-hard \cite{lipton76reachability} (no upper bound
is given). Relying on Petri nets to obtain algorithms for
parameterised networks is not always a good option.
%
% Also, Petri nets are a generalization
% of rendez-vous protocols, and
% this translation is not always desirable in terms of complexity.
In fact, the
coverability problem for parameterised networks with rendez-vous \iflong can
be solved in polynomial time\else is in \textsc{P}\fi \cite{german92}, while it is \Expspace-complete for Petri 
nets~\cite{rackoff78covering,lipton76reachability}. Hence, no upper
bound or lower bound can be directly deduced for the verification of
networks with non-blocking rendez-vous from \cite{raskin03petrinets}.

\noindent\emph{Our contributions.} We show that the coverability
problem for parameterised networks  with \emph{non-blocking
rendez-vous communication} over a finite alphabet is
\Expspace-complete. To obtain this result, we consider an
extension of counter machines (without zero test) where we add non-blocking decrement
actions and \iflong some restore mechanism, i.e.\fi edges that can bring back
the machine to its initial location at any moment. We show that
the coverability problem for these extended counter machines is
\Expspace-complete~(\cref{sec:cover-nb-machines}) and that it is equivalent to our problem
over parameterised networks~(\cref{sec:cover-rdv-protocols}). 
%je dirais que les détails techniques ne sont pas indispensables ici
%To get the upped bound we adapt, in a
%quite straightforward way, the
%\Expspace~proof of Rackoff \cite{rackoff78covering} for VASS. For the lower
%bound, we rely  on the \Expspace-hardness proof for VASS
%provided by Lipton \cite{lipton76reachability}, but here the changes in
%the construction are much more tricky. 
We consider then a subclass of parameterised networks -- 
\emph{wait-only protocols} --  in which
no state can allow to both request a rendez-vous and wait for one. This
restriction is very natural to model concurrent programs since when a thread
is waiting, it cannot perform any other action. We show that coverability problem 
can then be solved in polynomial time~(\cref{sec:wo}).
%show that if in the
%considered networks, a process cannot be in state where it can both
%request a rendez-vous and wait for one, then the coverability
%problem can be solved in polynomial time. Note that this last
%restriction makes senses to model concurrent
%programs where when a thread is waiting, it cannot perform any
%other action. 
Finally, we show that the synchronization problem, where we look for a reachable configuration with all the processes in a given state,
%which consists in seeking for an initial configuration leading to a configuration in which all processes are in a given state, 
is undecidable in our framework, even for wait-only protocols~(\cref{sec:target}).
Due to lack of space, some proofs are only given in the appendix.

\section{Rendez-vous Networks with Non-Blocking Semantics}
\label{section:definition-rdv}
%\subsection{Preliminaries}
%For a finite alphabet $\Sigma$, we let $\Sigma^*$ and $\Sigma^\omega$ to denote
%respectively the sets of finite and infinite sequences over $\Sigma$ (or words). 
For a finite alphabet $\Sigma$, we let $\Sigma^*$ denote the set of finite sequences over $\Sigma$ (or words). 
%Given $w\in\Sigma^*\cup\Sigma^\omega$, 
%we let $|w|$ denote its length: if $w=w_0\dots w_{n-1}\in \Sigma^*$, then $|w|=n$.
%and if $w\in \Sigma^\omega$, $|w|=\omega$. \nas{je pense que $\Sigma^\omega$ inutile maintenant}
Given $w\in\Sigma^*$, we let $|w|$ denote its length: if $w=w_0\dots w_{n-1}\in \Sigma^*$, then $|w|=n$.
We write $\nat$ to denote the set of natural numbers and $[i,j]$ to
represent the set $\set{k\in \nat \mid i\leq k \mbox{ and } k \leq
  j}$ for $i,j \in \nat$. For a finite set $E$, the set $\nat^E$ represents
the multisets over $E$. For two elements $m,m' \in  \nat^E$, we denote
$m+m'$ the multiset such that $(m+m')(e) = m(e) +m'(e)$ for
all $e \in E$. We say that $m \leq m'$ if and only if $m(e) \leq
m'(e)$ for all $e \in E$. If $m \leq m'$, then $m'-m$ is the multiset
such that  $(m'-m)(e) = m'(e)-m(e)$ for
all $e \in E$. Given a subset $E' \subseteq E$ and $m \in \nat^E$, we denote by $||m||_{E'}$ the sum $\Sigma_{e\in E'}m(e)$ of elements of $E'$ present in $m$. The size of a multiset $m$ is given by
$||m|| =||m||_E$. For $e \in E$, we use sometimes the
notation $\mset{e}$ for the multiset $m$ verifying $m(e)=1$ and
$m(e')=0$ for all $e' \in E\setminus\set{e}$ and, to represent for instance the multiset
with four elements $a, b,b$ and $c$, we will also use 
the notations $\mset{a, b, b, c}$ or $\mset{a, 2\cdot b, c}$.

%%%%%%%%%%%%%%%%%%%%%
% Arn 9/2/2023
%%%%%%%%%%%%%%%%%%%%%
%% Let $\Sigma$ be a finite alphabet, we note $!\Sigma$ (resp. $?\Sigma$) for the set 
%% $\{ !a~|~a\in \Sigma\}$ (resp. $\{?a~|a\in \Sigma\}$). For $n \in \mathbb{N}$, we note $[n]$ for the set $\{1, \dots, n \}$.
%
%% \paragraph{Multisets.} Let $E$ be a finite set. A multiset over $E$ is a function $m : E \rightarrow \mathbb{N}$, the size of a multiset $m$ is the number $||m|| = \sum_{e \in E} m(e)$. We denote multisets using a set-like notation, e.g $m = \mset a, 2\cdot b, c\Rbag$ denotes the multiset given by $m(a) = 1$, $m(b) = 2$, $m(c) = 1$ and for all $e \not \in \{a,b,c\}$, $m(e) = 0$. 

%% Given two multisets $v,v'$, we let $v + v'$ be the multiset given by $(v+v')(e) = v(e) + v'(e)$
%% for all $e \in E$ and $v-v'$ given by $(v-v')(e) = v(e) -v'(e)$ for all $e \in E$ if 
%% for all $e \in E$, $v(e) \geq v'(e)$. 

\subsection{Rendez-Vous Protocols}

We can now define our model of networks. We assume that all processes in the network follow the same protocol. Communication in the network is pairwise and is performed by \emph{rendez-vous} through a finite communication alphabet $\Sigma$. Each process can either perform an internal action using the primitive $\tau$, or request a rendez-vous by sending the message $m$ using the primitive $!m$ or answer to a rendez-vous by receiving the message $m$ using the primitive $?m$ (for $m \in \Sigma$). Thus, the set of primitives used by our protocols is $RV(\Sigma)=\set{\tau} \cup \set{?m,!m \mid m \in \Sigma}$. %\nas{Je ne sais pas si c'est indispensable, ici. J'enleverais cette (longue) phrase, on le dira avant, et on le redit apres.}We shall see in the sequel that we consider  a \emph{non-blocking} semantics for a process requesting a rendez-vous, i.e. when a process requests a rendez-vous with $!m$ if there is no process to answer to it, the requesting process will anyway be able to take the corresponding transition of the protocol.

\begin{definition}[Rendez-vous protocol]
	A \emph{rendez-vous protocol} (shortly protocol) is a tuple $\PP = (Q, \Sigma, \qinit, q_f, T)$ where $Q$ is a finite set of states, $\Sigma$ is a finite alphabet, $\qinit \in Q$ is the initial state, $q_f \in Q$ is the final state and $T \subseteq Q \times RV(\Sigma) \times Q$ is the finite set of transitions.
\end{definition}

%% A transition can be a sending transition (labeled with $!m$ for some $m \in \Sigma$), 
%% a receiving transition (labeled with $?m$ for some $m \in \Sigma$), or an internal 
%% transition (labeled with the symbol $\tau$). We note $R(a)$ the set of states $q$ such that there exists $q'$ with $(q, ?a, q') \in T$.
%%  \emph{configuration} of the protocol $\PP$ is a multiset $C \in \nat^Q$ for which  $c(q)$ 
%% denotes the number of processes in the state $q$ and $||c|| = \sum_{q \in Q} c(q)$ is the number of processes in the configuration $c$. For $Q' \subseteq Q$ (resp. $q_1, \dots, q_n \in Q$), we note $c(Q')$  (resp. $c(q_1, \dots, q_n)$) for the number $\sum_{q \in Q'} c(q)$  (resp. $\sum_{1 \leq i \leq n} c(q_i)$). 

For a message $m \in \Sigma$, we denote by $\Read{m}$ the set of states $q$ from which the message $m$ can be received, i.e.\ states $q$ such that there is a transition $(q, ?m, q') \in T$ for some $q' \in Q$. 
%For a transiton $t = (q, a, q')\in T$, we use $\textsf{src}(t)$ to represent the source state $q$ of the transition and $\textsf{tgt}(t)$ to represent its target state $q'$.\arstext{Is this used somewhere ?}
%\lugtext{We use it in the part wait-only once, and otherwise it only serves for proofs so not in this document}

A \emph{configuration} associated to the protocol $\PP$ is a non-empty multiset $C$ over $Q$ for which $C(q)$ 
denotes the number of processes in the state $q$ and $||C||$ denotes the total number of processes in the configuration $C$.  
%For a subset of states $S$, we write $\processin{C}{S}$ for the number of processes in the set $S$.
A configuration $C$ is said to be \emph{initial} if and only if $C(q)=0$ for all $q \in Q\setminus\set{\qinit}$. We denote by $\CC(\PP)$ the set of configurations and by $\Init(\PP)$ the set of initial configurations. Finally for  $n \in \nat\setminus\set{0}$, we use the notation $\CC_n(\PP)$ to represent the set of configurations of size $n$, i.e.\ $\CC_n(\PP)=\set{C \in \CC(\PP) \mid ||C||=n}$. When the protocol is made clear from the context, we shall write $\CC$, $\Init$ and $\CC_n$.

%% Let $n \in \mathbb{N}$, we note $\CC_n(\PP)$ the set of configurations with $n$ 
%% processes and $\CC(\PP)$ the set $\bigcup_{n \in \mathbb{N}} \CC_n(\PP)$.
%% We define the set of \emph{initial configurations} as the set 
%% $\Init(\PP) := \{c \in \CC(\PP) ~|~c(Q \setminus \{q_0\}) = 0\}$. 
%% When the context is clear, we will write $\CC$, $\CC_n$, and $\Init$.

%The size of a protocol \PP~is noted $|\PP|$ and is equal to $|T| \times 2\log (|Q|) (\log(|\Sigma|) + 2)$.

%\paragraph{Semantics and executions.}
We explain now the semantics associated with a protocol. For this
matter we define the relation ${\arrowP{\PP}} \subseteq {\bigcup_{n\geq
  1} \CC_n \times \big(\set{\tau} \cup \Sigma \cup \set{\mathbf{nb}(m)
  \mid m \in \Sigma}\big) \times \CC_n}$ as follows (here $\nb{\cdot}$ is a special symbol)\color{black}. Given $n \in
\nat\setminus \set{0}$ and $C,C' \in \CC_n$ and $m \in \Sigma$, we have:
\vspace*{-0.2cm}
\begin{enumerate}
	\item $C \arrowPlab{\PP}{\tau} C'$ iff there exists $(q, \tau, q') \in T$ 
	such that $C(q) > 0$ and  
	%if $q \not = q'$, $c'(q) = c(q) - 1$, $c'(q') = c(q') + 1$, 
	%and for all $p \not = q, q'$, $c'(p) = c(p)$, otherwise if $q =q'$, $c' = c$;
	$C' = C - \Lbag q \Rbag + \Lbag q' \Rbag$ \textbf{(internal)};
	\item $C \arrowPlab{\PP}{m} C'$ iff there exists $(q_1, !m, q_1') \in T$ and
	$(q_2, ?m, q_2')\in T$ such that $C(q_1)>0$ and $C(q_2)>0$ and $C(q_1)+C(q_2)\geq 2$ (needed when $q_1 = q_2$) \color{black}and
	  $C' = C - \mset{q_1, q_2} + \mset{q_1', q_2'}$  \textbf{(rendez-vous)};
	\item $C \arrowPlab{\PP}{\mathbf{nb}(m)} C'$ iff  there exists $(q_1, !m, q_1') \in T$, such 
	that $C(q_1)>0$ and $(C-\mset{q_1})(q_2)=0$ for all $(q_2, ?m, q_2') \in T$ and 
	$C' = C - \mset{q_1} + \mset{q'_1}$ \textbf{(non-blocking request)}.
	\vspace*{-0.1cm}
\end{enumerate}

Intuitively, from a configuration $C$, we allow the following
behaviours: either a process takes an internal transition
  (labeled by $\tau$), or two processes synchronize over a
rendez-vous $m$, or a process requests a rendez-vous  to which no
process can answer (non-blocking sending).

This allows us to define $S_\PP$   the transition system $ (\CC(\PP),
\arrowP{\PP})$ associated to $\PP$. We will write $C \arrowP{\PP} C'$ when there exists $a \in \set{\tau} \cup \Sigma \cup \set{\mathbf{nb}(m)
  \mid m \in \Sigma}$ such that $C \arrowPlab{\PP}{a} C'$ and denote by
$\arrowP{\PP}^\ast$ the reflexive and transitive closure of $
\arrowP{\PP}$. Furthermore, when made clear from the context, we might
simply write $\arrowP{}$ instead of $\arrowP{\PP}$.
An \emph{execution} is a finite sequence of configurations
$\rho = C_0C_1\dots$ such that, for all $0\leq i< |\rho|$, 
$C_i\arrowP{\PP} C_{i+1}$. The execution is said to be initial if $C_0\in \Init(\PP)$. 

	\begin{figure}[t]
%		\begin{minipage}[c]{0.5\textwidth}
			
%		\end{minipage}
%\begin{minipage}[c]{0.5\textwidth}

		\begin{center}

		\resizebox*{6cm}{!}{
			\begin{tikzpicture}[->, >=stealth', shorten >=1pt,node distance=2cm,on grid,auto, initial text = {}] 
				\node[state, initial] (q0) {$\qinit$};
				\node[state] (q1) [ right = of q0, accepting] {$q_1$};
				\node[state] (q5) [below   =  1.5 of q0] {$q_5$};
				\node[state] (q3) [above left = 1.5 of q5, xshift = -27] {$q_3$};
				\node[state] (q4) [ left =  of q5] {$q_4$}; 
				\node[state] (q6) [ right = of q5] {$q_6$};

				\node[state] (q2) [right = of q1] {$q_2$};
				
				\path[->] 
				(q0) edge node {$!a$} (q5)
				edge node {$?b$} (q1)
				(q1) edge node {$!c$} (q2)
				(q5) edge [bend right = 15] node [above] {$?a$} (q3)
				edge   node [above] {$?b$} (q4)
				edge   node  {$!b$} (q6)
				(q6) edge [bend right] node {$?c$} (q2)
				%					 edge  [bend right] node [below left] {$?d$} (q6)
				;

			\end{tikzpicture}

	}
		\end{center}
		\caption{Example of a rendez-vous protocol $\PP$}
		\label{fig-rdv2}
%		
%
%\end{minipage}
\end{figure}

\begin{example}\label{example-exec}
	Figure \ref{fig-rdv2} provides an example of a rendez-vous protocol
	where $\qinit$ is the initial state and $\qone$ the final state. A
	configuration associated to this protocol is for instance the multiset
	$\Lbag 2 \cdot q_1, 1\cdot q_4, 1 \cdot q_5\Rbag$ and the following
	sequence represents an initial execution: $\Lbag 2 \cdot \qinit \Rbag \arrowPlab{}{\mathbf{nb}(a)} \Lbag \qinit, \qfive \Rbag \arrowPlab{}{b} \Lbag \qone, \qsix \Rbag \arrowPlab{}{c} \Lbag 2 \cdot \qtwo \Rbag$.
\end{example}

% We give an example of  in \cref{fig-rdv}\iflong{, the set of states is $\{ q_0, q_1, q_2, q_3, p_1, p_2\}$, the initial state is $q_0$, the final state is $q_3$}. A configuration of the rendez-vous protocol \PP~is for example the multiset $\Lbag 2 \cdot q_1, 1\cdot p_1, 1 \cdot p_3\Rbag$.

\begin{remark} When we only allow behaviours of type \textbf{(internal)}
	and \textbf{(rendez-vous)}, this semantics corresponds to the classical rendez-vous semantics (\cite{german92,bala21finding,horn20deciding}). 
	In opposition, we will refer to the semantics defined here as the
	\emph{non-blocking semantics} where a process is not
	\emph{blocked} if it requests a rendez-vous and no process can
	answer to it. 
	Note that all behaviours possible in the classical rendez-vous semantics are as well possible in the non-blocking semantics but the converse is false.
\end{remark}

\subsection{Verification Problems}

We now present the problems studied in this work. For this matter,
given a protocol $\PP = (Q, \Sigma, \qinit, q_f, T)$, we define two
sets of final configurations. The first one  $\FinalE(\PP) = \{ C \in
\CC(\PP) ~\mid~C(q_f)> 0\}$
characterises the configurations where one
of the processes is in the final state. The second one  $\FinalA(\PP) = 
\{ C \in \CC(\PP) ~\mid~C(Q \setminus \{q_f\})= 0\}$ represents the
configurations where all the  processes are in the final state. Here
again, when the protocol is clear from the context, we might use the
notations $\FinalE$ and $\FinalA$. 
\ifterm{We study three problems: the
\emph{coverability problem} (\Cover), the synchronization problem
(\Target) and the termination problem (\Term) which all takes as input
a protocol $\PP$ and which can be stated as follows:}
\ifccover{We study three problems: the
	\emph{state coverability problem} (\Cover), the \emph{configuration coverability}  problem (\CCover) and 
	the \emph{synchronization problem}
	(\Target), which all take as input
	a protocol $\PP$~%, and for \CCover, a configuration as well, 
	and can be stated as follows:}
\begin{center}
  \begin{tabular}{|c|l|}
    \hline
    {\bf Problem name} & {\bf Question}\\
    \hline
    \hline
    \Cover & Are there $C_0 \in \Init$ and $C_f \in \FinalE$, such
      that $C_0 \arrowP{}^\ast C_f$?\\
    \hline
    \ifccover{
    	\CCover & Given $C \in \CC$, are there $C_0 \in \Init$ and $C' \geq C$, such that $C_0 \arrowP{}^\ast C'$?\\
    \hline
}
    \Target & Are there $C_0 \in \Init$ and $C_f \in \FinalA$, such
      that $C_0 \arrowP{}^\ast C_f$?\\
    \hline
    \ifterm{
    \Term & Does $\Exec_\infty (S_\PP) = \emptyset $?\\
    \hline
	}
    
\end{tabular}
\end{center}
\Cover~expresses a safety property: if $q_f$ is an error state and the answer is negative, then for any number of processes, no process will ever be in that error state. \Term, in another hand, is a liveness property: if $q_f$ is a deadlock state (a state in which no action is possible), and the answer is negative, then for any number of processes, all processes together are never blocked at the same time.
\color{black}
\begin{remark}
The difficulty in solving these problems lies in the fact that we
are seeking for an initial configuration allowing a specific execution
but the set of initial configurations is infinite. The difference
between \Cover~and \Target~ is that in the first one we ask for at least one
process to end up in the final state whereas the second one requires all the processes to end in this state.
\ifccover{Note that \Cover~is an instance of \CCover~but \Target~is not.}
\end{remark}

\begin{example}\label{example-verif-pbs}
	The rendez-vous protocol of Figure \ref{fig-rdv2} is a positive instance of \Cover, as shown in~\cref{example-exec}.
	%following execution putting one process in $\qone$: $\mset{2\cdot \qinit} \arrowPlab{}{\nb{a}} \mset{\qinit, \qfive} \arrowPlab{}{b} \mset{\qone, \qsix}$. 
	However, this is not the case for \Target: if an execution
    brings a process in $\qtwo$, this process cannot be brought afterwards to $\qone$. If $\qtwo$ is the final state, \PP~is now a positive instance of \Target~(see Example \ref{example-exec}).
	Note that if the final state is $\qfour$, $\PP$ is not a positive
    instance of \Cover~anymore. In fact, the only way to reach a
    configuration with a process in $\qfour$ is to put (at least) two
    processes in state $\qfive$ as this is the only state from which
    one process can send the message $b$. %This can not be done: let us see why. From any initial configuration, a process can only reach $\qfive$ by doing a non-blocking sending of message $a$. For example, from the initial configuration with 2 processes: $\mset{2 \cdot \qinit} \arrowPlab{}{\nb{a}} \mset{\qinit, \qfive}$. If another process sends the message $a$ from $\qinit$ to reach $\qfive$, the message is necessary received by the process on $\qfive$, and the first process leaves state $\qfive$. For example with two processes: $\mset{\qinit, \qfive} \arrowPlab{}{a} \mset{\qfive, \qthree}$. As a consequence, it is not possible to reach a configuration with two or more processes on state $\qfive$, and so the rendez-vous $\mset{2\cdot \qfive} \arrowPlab{}{b} \mset{\qsix, \qfour}$ will never occur (note however that state $\qsix$ can be covered).
	However, this cannot happen, since from an initial configuration,
    the only available action consists in sending the message $a$ as a non-blocking request. Once
	there is one process in state $q_5$, any other attempt to put another process in this state will induce a reception of message $a$ by the process already in $q_5$, which will hence leave $q_5$. 
	Finally, note that for any $n \in \mathbb{N}$, the configuration $\mset{n \cdot \qthree}$ is coverable, even if $\PP$ with $\qthree$ as final state is not a positive instance of \Target.
	%: for example for $n =2$: $\mset{3\cdot \qinit} \arrowPlab{}{\nb{a}} \mset{2\cdot \qinit, \qfive} \arrowPlab{}{a} \mset{ \qinit, \qfive, \qthree} \arrowPlab{}{a} \mset{ \qfive, 2 \cdot \qthree}$. 
%	However, if $\qthree$ is the final state, $\PP$ is not a positive instance of \Target.
	
\end{example}

\label{sec:def}

 	%\section{Coverability for specific counter machines}
	\section{Coverability for Non-Blocking Counter Machines}\label{sec:cover-nb-machines}
	
%In the sequel, we shall see that some of the verification problems we define for protocols are \emph{undecidable}, in order to show that, we will use a classical technique: a reduction from the reachabilty problem for Minksy machine. We define those machines in this section, along with another type of machine that we will use for the \Cover~problem. We call those \emph{non-blocking} counter machines. It has a set of counters, a set of locations, and to go from a location to another, it can either decrement one counter (if its value is greater than 0), increment it, or do a \emph{non-blocking} decrements: i.e decrement it if its value is greater than 0, or leave it to 0 otherwise. We will also define a particular class of non-blocking counter machines, called non-blocking counter machine with restores. \nas{on verra si on garde ce paragraphe tel quel...}\ars{Oui il faut réécrire ce texte}
%
We first detour into new classes of counter machines, which we call \emph{non-blocking counter machines} and \emph{non-blocking counter machines with restore}, in which a new way of decrementing the counters is added to the classical one: a non-blocking decrement, which is an action that can always
be performed. If the counter is strictly positive, it is decremented; otherwise it is let to 0. We show that the coverability of a control 
state in this model is \Expspace-complete, and use this result to solve coverability problems in rendez-vous protocols.

To define counter machines, given a set of integer variables (also called counters) $\Counters$, we use the notation $ \CAct{\Counters}$ to represent the set of associated actions given by $\set{\inc{\cpt},\dec{\cpt},\testz{\cpt} \mid \cpt \in \Counters} \cup \set{\nop}$. Intuitively, $\inc{\cpt}$ increments the value of the counter $\cpt$, while $\dec{\cpt}$ decrements it and $\testz{\cpt} $ checks if it is equal to $0$. We are now ready to state the syntax of this model.

\begin{definition}
	A \emph{counter machine} (shortly CM) is a tuple $M = (\Loc, \Counters, \Delta, \ellinit)$ such that $\Loc$ is a finite set of locations, $\ellinit \in \Loc$ is an initial location, $\Counters$ is a finite set of counters, and $\Delta \subseteq \Loc \times \CAct{\Counters} \times \Loc$ is finite set of transitions. 
  \end{definition}

 We will say that a CM is test-free (shortly \testfreeCM) whenever $\Delta \cap \Loc \times \{\testz{\cpt} \mid \cpt \in \Counters\} \times \Loc = \emptyset$.  A configuration of a CM $M = (\Loc, \Counters, \Delta, \ellinit)$ is a pair $(\ell, v)$ where $\ell \in \Loc$  specifies the current location of the CM and $v\in \nat^\Counters$ associates to each counter a natural value. The size of a CM $M$ is given by $|M|= |\Loc| + |\Counters| + |\Delta|$. %The initial configuration is the pair $(\ell_{in},\mathbf{0})$ where $\mathbf{0}(\cpt)=0$ for all $\cpt \in \Counters$.  
 Given two configurations $(\ell, v)$ and $(\ell',v')$ and a transition $\delta \in \Delta$, we define $(\ell, v) \xtransCM{\delta}_M (\ell', v')$ if and only if $\delta = (\ell, op, \ell')$ and one of the following holds:\\
\begin{minipage}[t]{7cm}
\begin{itemize}[itemindent=-0.2cm]
	\item $op = \nop$ and $v =v'$;
	\item $op = \inc{\cpt}$ and $v'(\cpt) = v(\cpt) + 1$ and \\ $v'(\cpt') = v(\cpt')$ for all $\cpt' \in \Counters \setminus \set{\cpt}$;
\end{itemize}
\end{minipage}
\begin{minipage}[t]{7cm}
\begin{itemize}
	 \item $op = \dec{\cpt}$ and $v'(\cpt) = v(\cpt) - 1$ and $v'(\cpt') = v(\cpt')$ for all $\cpt' \in \Counters \setminus \set{\cpt}$;
	 \item $op = \testz{\cpt}$ and $v(\cpt) = 0$ and  $v'= v$.
  \end{itemize}
\end{minipage}
\\

%We note $\transCM$ for $\bigcup_{\delta \in \Delta} \xtransCM{\delta}$.

 In order to simulate the non-blocking semantics of our rendez-vous protocols with counter machines, we extend the class of test-free CM with non-blocking decrement actions.
 
\begin{definition}
	A \emph{non-blocking test-free counter machine} (shortly \NBCM)~is a tuple $M=(\Loc, \Counters, \Delta_b, \Delta_{nb}, \ellinit)$ such that $(\Loc, \Counters, \Delta_b, \ellinit)$ is a \testfreeCM~and $\Delta_{nb} \subseteq \Loc \times \{\nbdec{\cpt} \mid \cpt \in \Counters\} \times \Loc$ is a finite set of \emph{non-blocking} transitions. 
%	When $\Delta_{nb}=\emptyset$, we say that it is a \emph{counter machine with zero test (0-CM)}, and we might describe $\Delta$ as $\Delta_b$ only. If, in addition, $C = \{c_1, c_2\}$, we say that it is a \MinskyMachine, this is actually a Minsky Machine.
%	
%	
%	If there is no zero test transitions (i.e $\Delta \cap \{\cpt \testZ \mid \cpt \in \Counters \}  = \emptyset$), we say that it is a \emph{non-blocking counter machine (\NBCM)}.
  \end{definition}
  
  Observe that in a \NBCM, both blocking and non-blocking decrements are possible, according to the definition of the transition relation. \color{black}
Again, a configuration is given by a pair $(\ell,v)\in \Loc\times\nat^\Counters$.
%, and the initial configuration is $(\ellinit, \mathbf{0})$. 
Given two configurations $(\ell, v)$ and $(\ell, v')$ and $\delta\in\Delta_b\cup\Delta_{nb}$, we extend the transition relation $(\ell,v)\xtransNbCM{\delta}_M (\ell,v')$
over the set $\Delta_{nb}$ in the following way: for $\delta= (\ell, \nbdec{\cpt}, \ell') \in \Delta_{nb}$, we have $(\ell,v) \xtransNbCM{\delta}_M (\ell',v')$ if and only if  $v'(\cpt) = \max(0, v(\cpt) - 1)$, and  $v'(\cpt') = v(\cpt')$ for all $\cpt' \in \Counters \setminus \set{\cpt}$. 

We say that $M$ is an \NBCM~\emph{with restore} (shortly \NRCM) when $(\ell, \nop, \ellinit) \in \Delta$ for all $\ell\in \Loc$, i.e. from each location, there is a transition leading to the initial location with no effect on the counters values.

 For a CM $M$ with set of transitions $\Delta$ (resp. an \NBCM~ with sets of transitions $\Delta_b$ and $\Delta_{\textit{nb}}$),
 we will write $(\ell, v) \xtransCM{}_M (\ell', v')$ whenever there exists $\delta \in \Delta$ (resp. $\delta \in \Delta_b\cup\Delta_{\textit{nb}}$) such that $(\ell, v) \xtransCM{\delta}_M (\ell', v')$ and use  $\xtransCM{}^\ast_M$ to represent the reflexive and transitive closure of $\xtransCM{}_M$. When the context is clear we shall write $\transCM$ instead of $\transCM_M$. We let $\mathbf{0}_\Counters$ be the valuation such that $\mathbf{0}_\Counters(\cpt)=0$ for all $\cpt\in \Counters$. 
An execution is a finite sequence of configurations  $(\ell_0, v_0) \transNbCM (\ell_1, v_1) \transNbCM \ldots \transNbCM(\ell_k, v_k)$.
 It is said to be initial if $(\ell_0,v_0)=(\ellinit, \mathbf{0}_\Counters)$. A configuration $(\ell,v)$ is called reachable if  $(\ellinit, \mathbf{0}_\Counters) \transCM^\ast (\ell,v)$.
%We say that an execution is initial if it starts with the configuration $(\ell_0, \mathbf{0}_{\Counters})$.

We shall now define the coverability problem for (non-blocking test-free) counter machines, which asks whether a given location can be reached from the initial configuration. We denote this problem \CMCover[$\mathcal{M}$], for $\mathcal{M}\in \{\textrm{CM}, \textrm{\testfreeCM}, \textrm{\NBCM}, \textrm{\NRCM}\}$. It takes as input a machine $M$ in $\mathcal{M}$  (with initial location $\ellinit$ and working over a set $\Counters$ of counters) and a location $\ell_f$ and it checks whether there is a valuation $v \in \mathbb{N}^\Counters$ such that $(\ellinit, \mathbf{0}_\Counters) \transNbCM^*(\ell_f, v)$.
In the rest of this section, we will prove that \CMCover[\NRCM]~is \Expspace-complete. To this end, we first establish that \CMCover[\NBCM]~is in \Expspace, by an adaptation of Rackoff's proof which shows that coverability in Vector Addition 
Systems is in \textsc{Expspace} \cite{rackoff78covering}. This gives also the upper bound for $\NRCM$, since any \NRCM~is a \NBCM. This result is established by the following theorem, whose proof is omitted due to lack of space.\color{black}
%We will in fact introduce a model in which \Cover~problem is \emph{as hard as} the analogous problem in \NBCM~called Non-blocking Vector Addition Systems (\NBVAS). We use this model in order to have a simple adaptation of the proof of \cite{rackoff78covering}, then, we will prove the link between the two models for the \Cover~problem.

%\nas{du coup, mettre ici la definition de \NBVAS.}\lug{fait}

\begin{theorem}\label{thm:cover-nbcm-in-expspace}
	$\CMCover[\NBCM]$ and $\CMCover[\NRCM]$~are in \Expspace.
\end{theorem}

%As any \NRCM~is a \NBCM, we get the following corollary.
%\begin{corollary}\label{coro:cover-nrcm-in-expspace}
%	
%\end{corollary}
%

%\subsection{Lower bound for \CMCover[\NRCM]}\label{subsec:cover-lower-bound}
%\nastext{dans toutes les figures de cette section, il faudrait peut-être rajouter le symbole $\nop$ sur les transitions vides, pour coller avec nos définitions}
%\lugtext{fait}
%\lugtext{Fait le 21 avril :
%
%- macros pour les compteurs y/z et s : \cptz, \cpty, \cpts
%
%- échange des roles de $\overline \cpt$ et $\cpt$ comme discuté sur zulip
%
%- $\overline Y_n$: emptyset, plus de machine $M'$, plus de $\mathtt{Inc}_n$.}
To obtain the lower bound, inspired by Lipton's proof showing that coverability in Vector Addition Systems is \Expspace-hard \cite{ esparza98decidability,lipton76reachability}, we rely on 2\textsc{Exp}-bounded \testfreeCM. We say that a CM $M = (\Loc,\Counters, \Delta,\ellinit)$ is 2\textsc{Exp}\emph{-bounded} if there exists $n \in O(|M|)$ such that any reachable configuration $(\ell, v)$ satisfies $v(\cpt) \leq 2^{2^n}$ for all $\cpt \in \Counters$.
We use then the following result.
\begin{theorem}[\cite{ esparza98decidability,lipton76reachability}]\label{th:expspace-hard-lipton}
\CMCover[2\textsc{Exp}-bounded \testfreeCM] is \Expspace-hard.
\end{theorem}

We now  show how to simulate a 2\textsc{Exp}-bounded \testfreeCM~by a \NRCM, by carefully handling restore transitions that may occur at any point in the
execution. We will ensure that each restore transition is followed by a reset of the counters, so that we can always extract from an execution of the
\NRCM~a correct initial execution of the original \testfreeCM. The way we enforce resetting of the counters is inspired by the way Lipton simulates 0-tests of a CM in a \testfreeCM. 
%In order to show that, we will show how to simulate the \Cover~problem in 2EXP-bounded Counter Machine.  This problem is shown to be \Expspace-hard in \cite{esparza98decidability,lipton76reachability}. The proof of this reduction is actually inspired by the proof of \Expspace-hardness for 2EXP-bounded CM.\nastext{je n'ai pas compris ce paragraphe. Pour moi la preuve de EXPSPACE-hardness est que : Turing Machine bornée exponentiellement -> 2EXP-bounded CM -> VASS. La preuve de Lipton parle de la deuxième flèche. Esparza dit que la première flèche est "bien connue". Est-ce que tu es d'accord Lucie, ou j'ai mal compris?}
%
%
As in \cite{lipton76reachability,esparza98decidability}, we will describe the final \NRCM~by means of several submachines. To this end, we define \emph{procedural non-blocking counter machines} that are \NBCM~ with several identified \emph{output states}: formally, a procedural-\NBCM~is a tuple $N = (\Loc, \Counters, \Delta_b, \Delta_{nb}, \ell_{in}, L_{\textit{out}})$ such that $(\Loc, \Counters, \Delta_b, \Delta_{nb}, \ell_{in})$ is a \NBCM, $L_{\textit{out}} \subseteq \Loc$, and there is no outgoing transitions from states in $L_{\textit{out}}$.\color{black}

%In \cite{esparza98decidability,lipton76reachability}, the authors present a proof for the \Expspace-hardness of the Cover problem for CM. In order to do so, they define several \emph{Procedural} Non-blocking Counter Machines. A Procedural \NBCM~is a tuple $N = (\Loc, \Counters, \Delta_b, \Delta_{nb}, \ell_{in}, L_{out})$ such that $(\Loc, \Counters, \Delta_b, \Delta_{nb}, \ell_{in})$ is a \NBCM~and $L_{out} \subseteq \Loc$. (Actually, the authors define procedural counters machines and not procedural \NBCM~but we can adapt the definitions to \NBCM~with $\Delta_{nb} = \emptyset$).

\begin{figure}[t]
		\resizebox{!}{1.5cm}{%
	\begin{tikzpicture}[->, >=stealth', shorten >=0.5pt,node distance=2cm,on grid,auto, initial text = {}] 
	\tikzstyle{initial}= [initial by arrow,initial text=]
	\node[state,initial, minimum width=0.1pt] (0) at (0,0) {$\ellinit'$};
	\node[state] [right of=0] (1) {$\ell_a$};
	%\node[state] [right of=0, yshift=-20] (y) {};
	\node[state] [right = 2.5 of 1] (2) {$\ell_{b}$};
	\node[state] [right of=2] (3) {$\ellinit$};
	\node[state] [right =5 of 3] (4) {$\ell_{f}$};
	
	\node[draw, fill = yellow, fill opacity = 0.2, text opacity = 1, fit=(1) (2), text height=0.04 \columnwidth] (RstInc) {$\mathtt{RstInc}$};
	\node[draw, fill = pink, fill opacity = 0.2, text opacity = 1, fit=(3) (4), text height=0.04 \columnwidth] (CM) {Counter Machine $M$};

%	\path (6) -- node[auto=false]{\ldots} (7);

	\path[->] 	
	(0) edge node {$\nop$} (1)
	
	(2) edge node {$\nop$} (3)
	
	;
	
	\draw[->, dashed] (RstInc.south) |- ++(0,-0.2) -|(0.south); 
	\draw[->, dashed] (CM.south) |- ++(0,-0.2) -| (0.south);
	\node[anchor=east, xshift = -40, yshift=-12] at (CM.south) {\small Restore transitions};
%	\node[anchor=north, yshift=-5] at (CM.north) {\small Counter machine $M$};
%	\node[anchor=north, yshift=-15] at (CM.north) {\small machine $M$};
\end{tikzpicture}
	}
	\caption{The \NRCM~$N$} \label{fig:nrvass}
\end{figure}

Now fix a 2\textsc{Exp}-bounded \testfreeCM~$M = (\Loc,\Counters, \Delta,\ellinit)$, $\ell_f\in\Loc$ the location to be covered. There is some $c$, such that, any reachable configuration $(\ell, v)$ satisfies $v(\cpt) < 2^{2^{c |M|}}$ for all $\cpt \in \Counters$, fix $n = c|M|$.
We build a \NRCM~$N$ as pictured in~\cref{fig:nrvass}. The goal
of the procedural \NBCM~$\mathtt{RstInc}$ is to ensure that all counters in $\Counters$ are reset. Hence, after each restore transition, we are sure that we start over a fresh execution of the \testfreeCM~$M$.
%In the following, we explain how we build the procedural-\NBCM~$\mathtt{RstInc}$.
%The non-blocking transitions of a \NBCM~allow us to reset the counters as long as we know that they are bounded, with as many non-blocking decrements as needed. However, in order to keep a number of states polynomial in the size of $M$, we will implement these non-blocking decrement transitions within loops, making use of other counters to make sure that the loops are taken
%the right number of times. Hence 
We will need the mechanism designed by Lipton to test whether a counter is equal to 0. \iflong For a counter $\cpt$ bounded by some value $K$, this is done by duplicating $\cpt$ into $\overline{\cpt}$ and ensure along any execution that the sum of $\cpt$ and $\overline\cpt$ is equal to $K$.\fi
% To be able to reset the counters of the \NRCM~ we build, we will duplicate each counter
%$\cpt$ in a copy $\overline{\cpt}$. 
So, we define two families of sets of counters $(Y_i)_{0\leq i \leq n}$ and $(\overline{Y_i})_{0\leq i\leq n}$ as follows. Let $Y_i = \{\cpty_i, \cptz_i, \cpts_i \}$ and 
$\overline Y_i = \{\overline \cpty_i, \overline \cptz_i, \overline \cpts_i\}$ for all $0\leq i < n$ and $Y_n = \Counters$ and $\overline{Y}_n = \emptyset$ and $\Counters'=\bigcup_{0\leq i\leq n} Y_i\cup \overline Y_i$. All the machines we will describe from now on will work over the set of counters $\Counters'$.

\smallskip

\noindent\emph{Procedural-\NBCM~$\mathtt{TestSwap}_i(\cpt)$.}
We use a family of procedural-\NBCM~defined in~\cite{lipton76reachability, esparza98decidability}: for all $0\leq i <n$, for all $\overline\cpt\in \overline Y_i$, $\mathtt{TestSwap}_i(\overline\cpt)$ is a procedural-\NBCM~with an initial location $\ellinit^{\mathtt{TS},i,\cpt}$, and two output locations $\ell^{\mathtt{TS},i,\cpt}_{z}$ and 
$\ell^{\mathtt{TS},i,\cpt}_{\textit{nz}}$.  It tests if the value of $\overline{ \cpt}$ is equal to 0, using the fact that the sum of the values of $\cpt$ and $\overline\cpt$ is equal to $2^{2^i}$. If $\overline\cpt=0$, it swaps the values of $\cpt$ and $\overline\cpt$, and the execution ends in the output location $\ell^{\mathtt{TS},i,\cpt}_z$. Otherwise, counters values are left unchanged and the execution ends in $\ell^{\mathtt{TS},i,\cpt}_{nz}$. In any case, other counters are not modified by the execution. Note that
$\mathtt{TestSwap}_i(\cpt)$ makes use of variables in $\bigcup_{1\leq j< i} Y_i\cup\overline Y_i$.

\iflong
Formally, these machines have the following property:
We use this proposition. 
\begin{proposition}[\cite{lipton76reachability,esparza98decidability}]\label{prop:test-swap}
Let $0\leq i < n$, and $\overline\cpt\in \overline Y_i$.
For all $v,v'\in\mathbb{N}^{X'}$, for $\ell\in\{\ell^{\mathtt{TS},i,\cpt}_{\textit{z}},\ell^{\mathtt{TS},i,\cpt}_{\textit{nz}}\}$, 
we have $(\ellinit^{\mathtt{TS},i},v)\transCM^*(\ell,v')$ in ${\mathtt{TestSwap}_i(\overline\cpt)}$
if and only if :
\begin{itemize}
\item (PreTest1): for all $0 \leq j < i$, for all $\overline\cpt_j \in \overline Y_j$, $v(\overline\cpt_j) = 2^{2^j}$ and for all $\cpt_j \in  Y_j$, $v(\cpt_j) = 0$;
\item (PreTest2): $v(\overline \cpts_i) = 2^{2^i}$ and $ v( \cpts_i) = 0$;
	\item (PreTest3): $v(\cpt) + v(\overline\cpt) = 2^{2^i}$;
\item (PostTest1): For all $\cpty\notin\{\cpt,\overline\cpt\}$, $v'(\cpty) = v(\cpty)$; 
	\item (PostTest2): either $(i)$ $v(\overline\cpt) = v'(\cpt) = 0$, $v(\cpt) = v'(\overline\cpt)$ and $\ell = \ell^i_z$, or $(ii)$ $v'(\overline \cpt) = v(\overline \cpt) >0$, $v'(\cpt) = v(\cpt)$ and $\ell = \ell^{\mathtt{TS},i,\cpt}_{nz}$.
\end{itemize}
Moreover, if for all $0 \leq j \leq n$, and any counter $\cpt \in Y_j \cup \overline Y_j$, $v(\cpt)\leq 2^{2^j}$, then for all $0 \leq j \leq n$, and any counter $\cpt\in Y_j \cup \overline Y_j$, the value of $\cpt$ will never go above $2^{2^j}$ during the execution. 
\end{proposition}

Note that for a valuation $v\in\mathbb{N}^{X'}$ that meets the requirements (PreTest1), (PreTest2) and (PreTest3), there is only one configuration
 $(\ell,v')$ with $\ell \in \{\ell^{\mathtt{TS},i,\cpt}_{\textit{z}},\ell^{\mathtt{TS},i,\cpt}_{\textit{nz}}\}$ such that $(\ell_{in},v) \transNbCM^* (\ell,v')$.
 \fi
 
 \smallskip
 
\noindent\emph{Procedural \NBCM~$\mathtt{Rst}_i$.} 
 We use these machines to define a family of procedural-\NBCM~($\mathtt{Rst}_i)_{0\leq i\leq n}$ that reset the counters in $Y_i\cup\overline{Y_i}$, assuming that 
 their values are less than or equal to $2^{2^i}$.
 %\nas{cas spécial pour $X$?}. %lucie : géré avec le Y_i 
 Let $0\leq i\leq n$, we let $\mathtt{Rst}_i=(\Loc^{\mathtt{R},i}, \Counters',\Delta_b^{\mathtt{R},i},\Delta^{\mathtt{R},i}_{nb}, \ell^{\mathtt{R},i}_{in}, \{\ell_{out}^{\mathtt{R},i}\})$. The machine $\mathtt{Rst}_0$ is pictured Figure~\ref{fig:rst}.
%\begin{figure}[H]
%	\includegraphics[width=14cm]{pic/rst.png}
%	\centering
%	\caption{Description of the procedural-\NRCM~family} \label{png:rst}
%\end{figure}
For all $0\leq i< n$, the machine $\mathtt{Rst}_{i+1}$ uses counters from $Y_i\cup\overline{Y_i}$ and procedural-\NBCM~$\mathtt{Testswap}_i(\overline \cptz_i)$ and 
$\mathtt{Testswap}_i(\overline \cpty_i)$ to control the number of times variables from $Y_{i+1}$ and $\overline{Y}_{i+1}$ are decremented. It is pictured Figure~\ref{fig:rst2}. Observe that since $Y_n=\Counters$, and $\overline{Y_n}=\emptyset$, the machine $\mathtt{Rst}_n$ will be a bit different from the picture: there will only be non-blocking decrements over counters from $Y_n$, that is over counters $\Counters$ from the initial \testfreeCM~$M$. If $\overline \cpty_i$, $\overline \cptz_i$ (and $\overline{\mathtt{s}}_i$) are set to $2^{2^i}$ and ${\cpty_i}$, ${\cptz_i}$ (and $\mathtt{s}_i$) are set to 0, then each time 
this procedural-\NBCM~ takes an outer loop, the variables of $Y_{i+1}\cup\overline Y_{i+1}$ are decremented (in a non-blocking fashion) $ 2^{2^i}$ times.
This is ensured by \iflong Proposition~\ref{prop:test-swap}\else the properties of $\mathtt{TestSwap}_i(\cpt)$\fi. Moreover, the location $\ell^{\mathtt{TS}, i, \cpty}_z$ will only be reached when the counter $\overline \cpty_i$ 
is set to 0, and this will happen after $2^{2^i}$ iterations of the outer loop, again thanks to \iflong Proposition~\ref{prop:test-swap}\else the properties of $\mathtt{TestSwap}_i(\cpt)$\fi. So, all in all, variables from $Y_i$ 
and $\overline{Y}_{i+1}$ will take a non-blocking decrement $2^{2^i}.2^{2^i}$ times, that is $2^{2^{i+1}}$. 

\begin{figure}[t]
	\resizebox{14cm}{!}{%
	\begin{tikzpicture}[->, >=stealth', shorten >=1pt,node distance=2cm,on grid,auto, initial text = {}] 
	\tikzstyle{initial}= [initial by arrow,initial text=,initial
	distance=.7cm]
	\node[state,initial, minimum width=0.1pt] (0) at (0,0) {$\ellinit^{\mathtt{R},0}$};
	\node[state] [right of=0, xshift=10] (1) {};
	%\node[state] [right of=0, yshift=-20] (y) {};
	\node[state] [right of=1, xshift=10] (2) {};
	\node[state] [right of=2, xshift=10] (3) {};
	\node[state] [right of=3, xshift=10] (4) {};
	\node[state] [right of=4, xshift=10] (5) {};
	\node[state] [right of=5, xshift=10] (6) {$\ell_{out}^{\mathtt{R},0}$};
%	\node [right of=4, xshift = 2] (do) {...};
%	\node[state] [right of=y, xshift=10] (2) {$q_2$};
	
%	\path[->, bend left=10] 	
%	(0) edge node[above, xshift=-10] {${init}{1}$} (x)
%	;
%	\path[->, bend right=10] 	
%	(0) edge node[below, xshift=-10] {${init}{2}{}$} (y)
%	;
	\path (4) -- node[auto=false]{\ldots} (5);
	%\draw (4) -- (5);
	\path[->] 	
	(0) edge node {$\nbdec{\mathtt{y}_0}$} (1)
	(1) edge node {$\nbdec{\mathtt{y}_0}$} (2)
	(2) edge node {$\nbdec{ \mathtt{\bar y}_0}$} (3)
	(3) edge node {$\nbdec{ \mathtt{\bar y}_0}$} (4)
%	(4) edge node {} (5)
	(5) edge node {$\nbdec{ \mathtt{\bar s}_0}$} (6)
	;

\end{tikzpicture}
}
	\caption{Description of $\mathtt{Rst_0}$} \label{fig:rst}
\end{figure}

\begin{figure}[t]
	\resizebox{14cm}{!}{%
	\begin{tikzpicture}[->, >=stealth', shorten >=1pt,node distance=2cm,on grid,auto, initial text = {}] 
	\tikzstyle{initial}= [initial by arrow,initial text=]
	\node[state,initial, minimum width=0.1pt] (0) at (0,0) {$\ellinit^{\mathtt{R},i+1}$};
	\node[state] [right of=0] (1) {$\ell^{\mathtt{R},i+1}_1$};
	%\node[state] [right of=0, yshift=-20] (y) {};
	\node[state] [right of=1] (2) {$\ell_2^{\mathtt{R},i+1}$};
	\node[state] [right of=2] (3) {$\ell_3^{\mathtt{R},i+1}$};
	\node[state] [right of=3] (4) {$\ell_4^{\mathtt{R},i+1}$};
	\node[state] [below = 2.2 of 1, xshift=40] (5) {$\ell_5^{\mathtt{R},i+1}$};
	\node[state] [right of=5, xshift=20] (6) {$\ell_6^{\mathtt{R},i+1}$};
	\node[state] [right of=6, xshift=14] (7) {$\ell_r^{\mathtt{R},i+1}$};
	\node[state] [right of=7] (linz) {$\ellinit^{\mathtt{TS},i, \cptz}$};
	\node[state] [above =   of linz] (lnzz) {$\ell_{nz}^{\mathtt{TS},i,\cptz}$};
	\node[state] [right = 2  of linz] (lzz) {$\ell_{z}^{\mathtt{TS},i,\cptz}$};\\
	\node[state] [right of=lzz] (liny) {$\ellinit^{\mathtt{TS},i,\cpty}$};
	\node[state] [above =   of liny] (lnzy) {$\ell_{nz}^{\mathtt{TS},i,\cpty}$};
	\node[state] [right = 2 of liny] (lzy) {$\ell_{z}^{\mathtt{TS},i,\cpty}$};
	\node[state] [right = 1.8  of lzy] (lout) {$\ell_{out}^{\mathtt{R},i+1}$};
	\node[draw, fill = teal, fill opacity = 0.2, text opacity = 1, fit=(lnzz) (linz) (lzz)] (Testz) {$\mathtt{TestSwap_i}(\overline \cptz_i)$};
	\node[draw, fill = teal, fill opacity = 0.2, text opacity = 1, fit=(lnzy) (liny) (lzy)] (Testy) {$\mathtt{TestSwap_i}(\overline \cpty_i)$};
	%	\node [right of=4, xshift = 2] (do) {...};
	%	\node[state] [right of=y, xshift=10] (2) {$q_2$};
	
	%	\path[->, bend left=10] 	
	%	(0) edge node[above, xshift=-10] {${init}{1}$} (x)
	%	;
	%	\path[->, bend right=10] 	
	%	(0) edge node[below, xshift=-10] {${init}{2}{}$} (y)
	%	;
	\path (6) -- node[auto=false]{\ldots} (7);
	%\draw (4) -- (5);
	\path[->] 	
	(0) edge node {$\dec{\overline\cpty_i}$} (1)
	(1) edge node {$\inc{\cpty_i}$} (2)
	(2) edge node {$\dec{\overline\cptz_i}$} (3)
	(3) edge node {$\inc{\cptz_i}$} (4)
	(4) edge [out=-95, in=50] node  [above]{$\nbdec{\mathtt{y}_{i+1}}$} (5)
	(5) edge node {$\nbdec{\mathtt{\bar y}_{i+1}}$} (6)
	(7) edge node {$\nop$} (linz)
	(lzz) edge node {$\nop$} (liny)
	(lzy) edge node {$\nop$} (lout)
	(lnzz.north) edge [out=160, in=15] node[above] {$\nop$} (2.north)
	(lnzy.north) edge [out=170, in=15] node[above] {$\nop$} (0.north)
	;
	
%	\path[->, bend right=25] 	
%	(lnzz.north) edge node[below] {$\nop$} (2.north)
%	;
%	
%	\path[->, bend right=25] 	
%	(lnzy.north) edge node[above] {$\nop$} (0.north)
%	;

	%\draw[color=orange] (A11) ..controls +(0.5,0) and +(0,0.5).. (3,4.5); % arc 1

\end{tikzpicture}
	}
	\caption{Description of $\mathtt{Rst_{i+1}}$} \label{fig:rst2}	
\end{figure}
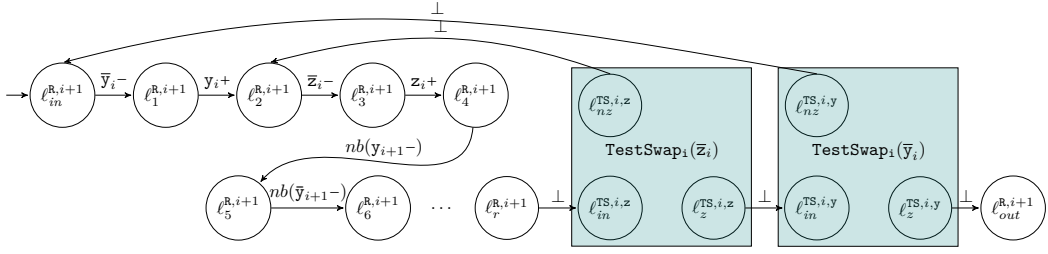

\iflong
 These properties are formalized in the following proposition.
\begin{proposition}\label{lemma:cover:expspace-hard:nbcm:rst-spec}
	%Let  $0\leq i \leq n$ and $\ellinit^{\mathtt{R},i}$ (resp. $\ell^{\mathtt{R},i}_{out}$) be the entering (resp. exiting) state of the procedural-\NBCM~$\mathtt{Rst_i}$.
	 For all $0\leq i\leq n$, for all $v\in \mathbb{N}^{\Counters'}$ such that 
	 \begin{itemize}
		\item (PreRst1): for all $0 \leq j < i$, for all $\overline \cpt \in \overline Y_j$, $v(\overline \cpt) = 2^{2^j}$ and for all $\cpt \in Y_j$, $v(\cpt) = 0$,
	\end{itemize}
	for all $v' \in \mathbb{N}^{\Counters'}$, if $(\ellinit^{\mathtt{R},i}, v) \transNbCM^* (\ell^{\mathtt{R},i}_{out},v')$ in $\mathtt{Rst_i}$ then
	%then, there for all paths $(\ellinit^{\mathtt{R},i}, v) \transNbCM^+ (\ell^{\mathtt{R},i}_{out},v')$ in $\mathtt{Rst_i}$, $v'$ satisfies the following:
	\begin{itemize}
		\item (PostRst1): for all $\cpt \in Y_i \cup \overline Y_i$, $v'(\cpt) = \max(0, v(\cpt) - 2^{2^i})$,
		\item (PostRst2): for all $\cpt \not \in Y_i \cup \overline Y_i$, $v'(\cpt) = v(\cpt)$.
		\end{itemize}
 \end{proposition}
\fi

For all $\cpt\in\Counters'$, we say that
$\cpt$ is \emph{initialized} in a valuation $v$ if $\cpt\in Y_i$ for some $0\leq i\leq n$ and $v(\cpt)=0$, or $\cpt\in \overline Y_i$ for some
$0\leq i\leq n$ and $v({\cpt})=2^{2^i}$. 
For $0\leq i\leq n$, we say that a valuation $v\in\mathbb{N}^{\Counters'}$ is \emph{$i$-bounded} if for all $\cpt \in Y_i \cup \overline Y_i$, $v(\cpt) \leq 2^{2^i}$.
%We get the immediate corollary:
%\begin{lemma}\label{lemma:cover:expspace-hard:nbcm:rst-spec-2}
%Let $0\leq i\leq n$, and $v\in\mathbb{N}^{\Counters'}$ satisfying (PreRst1) for $\mathtt{Rst}_i$. If $v$ is $i$-bounded, then 
%	the unique configuration such that $(\ell^{\mathtt{R},i}_{in},v) \transNbCM^+ (\ell^{\mathtt{R},i}_{out}, v')$ in $\mathtt{Rst}_i$ is defined $v'(\cpt) = 0$ for all 
%	$\cpt \in Y_i \cup \overline Y_i$ and $v'(\cpt) = v(\cpt)$ for all $\cpt \notin Y_i \cup \overline Y_i$.
%	
%%	Moreover, for all $0 \leq j \leq n$, and any counter $\cpt_j \in Y_j \cup \overline Y_j$, its value never goes above $2^{2^j}$ during the execution. \lug{rajouté ici pour ne pas avoir une propriété juste pour ça}
%\end{lemma}	

\iflong
The procedural-\NBCM~$\mathtt{Rst}_i$ is taking care of resetting counters in $Y_i\cup\overline{Y}_i$.

The following lemma states that no counter in $Y_j\cup\overline{Y_j}$,
for $1\leq j\leq n$,
will be increased over $2^{2^j}$ during this process, and that it reset properly counters in $Y_i \cup \overline{Y_i}$. 

%\begin{proposition}\label{lem:rst-bounded}
%	Let $0\leq i \leq n$, and let $v\in\mathbb{N}^{\Counters'}$ satisfying (PreRst1) %such that (PreRst1) holds 
%	%and (PreRst2) holds 
%	for $\mathtt{Rst_i}$. If for all $0\leq j \leq n$, $v$ is $j$-bounded, then  for all $(\ell,v')\in\Loc^{\mathtt{R},i}\times\mathbb{N}^{\Counters'}$ such that 
%	$(\ell^{\mathtt{R},i}_{in},v) \transNbCM^* (\ell, v')$ in $\mathtt{Rst}_i$, $v'$ is $j$-bounded for all $0\leq j \leq n$.
%\end{proposition}

\begin{lemma}\label{lem:rst-spec-bounded}
	Let $0\leq i \leq n$, and let $v\in\mathbb{N}^{\Counters'}$ satisfying (PreRst1) for $\mathtt{Rst_i}$. If for all $0\leq j \leq n$, $v$ is $j$-bounded, then  for all $(\ell,v')\in\Loc^{\mathtt{R},i}\times\mathbb{N}^{\Counters'}$ such that 
	$(\ell^{\mathtt{R},i}_{in},v) \transNbCM^* (\ell, v')$ in $\mathtt{Rst}_i$, $v'$ is $j$-bounded for all $0\leq j \leq n$. Furthermore, the unique configuration such that $(\ell^{\mathtt{R},i}_{in},v) \transNbCM^* (\ell^{\mathtt{R},i}_{out}, v')$ in $\mathtt{Rst}_i$ is defined by $v'(\cpt) = 0$ for all 
	$\cpt \in Y_i \cup \overline Y_i$ and $v'(\cpt) = v(\cpt)$ for all $\cpt \notin Y_i \cup \overline Y_i$.
\end{lemma}
\else
The construction ensures that when one enters $\mathtt{Rst}_i$ with a valuation $v$ that is $i$-bounded, and in which all variables in $\bigcup_{0\leq j<i} Y_j\cup\overline Y_j$ are initialized, the location $\ell^{\mathtt{R},i}_{out}$ is reached with a valuation $v'$ such that: $v'(\cpt) = 0$ for all 
	$\cpt \in Y_i \cup \overline Y_i$ and $v'(\cpt) = v(\cpt)$ for all $\cpt \notin Y_i \cup \overline Y_i$. Moreover, if $v$ is $j$-bounded for all $0\leq j\leq n$, then 
	any valuation reached during the execution remains $j$-bounded for all $0\leq j\leq n$. 
%\begin{lemma}\label{lem:rst-spec-bounded}
%	Let $0\leq i \leq n$, and let $v\in\mathbb{N}^{\Counters'}$ such that all the variables in $\bigcup_{0\leq j <i} Y_j\cup\overline Y_j$ are initialized. 
%	If $v$ is $i$-bounded, then 
%	the unique configuration such that $(\ell^{\mathtt{R},i}_{in},v) \transNbCM^+ (\ell^{\mathtt{R},i}_{out}, v')$ in $\mathtt{Rst}_i$ is defined by $v'(\cpt) = 0$ for all 
%	$\cpt \in Y_i \cup \overline Y_i$ and $v'(\cpt) = v(\cpt)$ for all $\cpt \notin Y_i \cup \overline Y_i$.
%	Moreover, if, for all $0\leq j \leq n$, $v$ is $j$-bounded, then  for all $(\ell,v')\in\Loc^{\mathtt{R},i}\times\mathbb{N}^{\Counters'}$ such that 
%	$(\ell^{\mathtt{R},i}_{in},v) \transNbCM^* (\ell, v')$ in $\mathtt{Rst}_i$, $v'$ is $j$-bounded for all $0\leq j \leq n$. \end{lemma}
%
\fi

\smallskip

\noindent{\emph{Procedural \NBCM~$\mathtt{Inc}_i$.}
The properties we seek for $\mathtt{Rst}_i$ are ensured whenever the variables in $\bigcup_{0\leq j<i}Y_j\cup \overline Y_j$ are initialized. This is taken care of by
a family of procedural-\NBCM~introduced in~\cite{lipton76reachability,esparza98decidability}. For all $0\leq i< n$, ${\mathtt{Inc}_i}$ is a procedural-\NBCM~with initial location $\ellinit^{\mathtt{Inc}, i}$, and unique output location $\ell^{\mathtt{Inc}, i}_\textit{out}$. %Fix $n \in \mathbb{N}$, and define two families of counters $Y_i = \{y_i, z_i, s_i \}$ and $\bar Y_i = \{\bar y_i, \bar z_i, \bar s_i\}$ for any $0 \leq i < n$. We also define $Y_n$ and $\bar Y_n$ as two arbitrary set of counters. All the procedural-\NBCM~presented below have as counters' set $\Counters' := \bigcup_{0 \leq j \leq n} (Y_i \cup \bar{Y_i})$.
%In \cite{esparza98decidability,lipton76reachability}, they present a family of procedural-\NBCM~$\mathtt{(Inc_i  )_{0\leq i\leq n}}$, which does the following: given a set of counters, it increases those counters'values of $2^{2^i}$ for some $i$. 
They enjoy the following property: \iflong\else for $0\leq i<n$, when one enters $\mathtt{Inc}_i$ with a valuation $v$ in which all the variables in $\bigcup_{0\leq j<i} Y_j\cup\overline Y_j$
are initialized and $v(\cpt)=0$ for all $\cpt\in \overline Y_i$, then the location $\ell^{\mathtt{Inc}_i}_{out}$ is reached with a valuation $v'$ such that $v'(\cpt)=2^{2^i}$
for all $\cpt\in\overline{Y}_i$, and $v'(\cpt)=v(\cpt)$ for all other $\cpt\in\Counters'$. Moreover, if $v$ is $j$-bounded for all $0\leq j\leq n$, then 
	any valuation reached during the execution remains $j$-bounded for all $0\leq j\leq n$. 
\fi
\iflong
\begin{proposition}[\cite{lipton76reachability,esparza98decidability}]\label{proposition:inc}
For all $0\leq i< n$, for all $v,v'\in\mathbb{N}^{\Counters'}$, $(\ellinit^{\mathtt{Inc}, i},v) \transNbCM^* (\ell_{out}^{\mathtt{Inc}, i}, v')$ in $\mathtt{Inc}_i$ if and only if:
\begin{itemize}
	\item (PreInc1) for all $0 \leq j < i$, for all $\cpt \in \overline Y_j$, $v(\cpt) = 2^{2^j}$ and for all $\cpt \in  Y_j$, $v(\cpt) = 0$;
	\item (PreInc2) for all $\cpt \in \overline Y_i$, $v( \cpt) = 0$,
	\item (PostInc1) for all $ \cpt \in \overline Y_i$, $v'(\cpt) = 2^{2^i}$;
	\item (PostInc2) for all $\cpt \not \in Y_i $, $v'(\cpt) = v(\cpt)$.
\end{itemize}	
Moreover, if for all $0\leq j \leq n$, $v$ is $j$-bounded, then  for all $(\ell,v'')$ such that 
	$(\ell^{\mathtt{Inc},i}_{in},v) \transNbCM^* (\ell, v'')$ in $\mathtt{Inc}_i$, then $v''$ is $j$-bounded for all $0\leq j\leq n$.
\end{proposition}
\fi
%More formally, for any configuration $(\ell_{in^{\mathtt{Inc}, i}, v)$ where $\ell_{in}^{\mathtt{Inc}, i}$ is the initial state of the procedural-\NBCM~and $L_{out} = \{\ell_{out}\}$, there exists a path $(\ell_{in}^{\mathtt{Inc}, i},v) \transNbCM^* (\ell_{out}^{\mathtt{Inc}, i}, v')$ if and only if $v$ is such that:
%Note that for a valuation $v\in\mathbb{N}^{X'}$ that meets the requirements (PreTest1), (PreTest2) and (PreTest3), there is only one configuration
%$(\ell,v')$ with $\ell \in \{\ell^i_{\textit{z}},\ell^i_{\textit{nz}}\}$ such that $(\ell_{in},v) \transNbCM^* (\ell,v')$.
%Note that for a function $v$ respecting (PreInc1) and (PreInc2), there is an only $v'$ such that $(\ell_{in}^{\mathtt{Inc}, i},v) \transNbCM^+ (\ell_{out}^{\mathtt{Inc}, i},v')$.
%Furthermore, for any $0 \leq j \leq n$, and any counter $x \in Y_j \cup \bar Y_j$, its value never goes above $2^{2^j}$ during the execution. 
%Note also that, for any function in which one of the (PreInc1), (PreInc2) does not hold, there is simply no path exiting the procedural-\NBCM.\nas{utile?}

\smallskip

\noindent\emph{Procedural \NBCM~$\mathtt{RstInc}$.}
%Now that we know how to initialize all the counters, we can explain how we define the procedural-\NBCM~$\mathtt{RstInc}$, which makes use of the other ones we have already introduced. %
Finally, let $\mathtt{RstInc}$ be a procedural-\NBCM~ with initial location $\ell_a$ and output location $\ell_b$, over the set of counters $\Counters'$ and built as an 
alternation of $\mathtt{Rst}_i$  and $\mathtt{Inc}_i$ for $0\leq i<n$, finished by $\mathtt{Rst}_n$. It is depicted in \cref{fig:cover:expspace-hard:rstinc}. Thanks to the
properties of the machines $\mathtt{Rst}_i$ and $\mathtt{Inc}_i$, in the output location of each $\mathtt{Inc}_i$ machine, the counters in $\overline{Y}_i$ are set
to $2^{2^i}$, which allow counters in $Y_{i+1}\cup \overline Y_{i+1}$ to be set to 0 in the output location of $\mathtt{Rst}_{i+1}$. Hence, in location $\ell^{\mathtt{Inc},n}_{out}$, counters in $Y_n=\Counters$ are set to 0.
 %Properties of 
%all the procedural-\NBCM~that constitute it ensure that the counters are correctly reset, and that the valuation stays $i$-bounded, for all $0\leq i\leq n$. 
\begin{figure}[t]
	\resizebox{14cm}{!}{%
	\begin{tikzpicture}[->, >=stealth', shorten >=0.5pt,node distance=1.7cm,on grid,auto, initial text = {}] 
	\tikzstyle{initial}= [initial by arrow,initial text=]
	\node[state,initial, minimum width=0.1pt] (0) at (0,0) {$\ell_{a}$};
	\node[state] [right of=0] (1) {$\ellinit^{\mathtt{R},0}$};
	%\node[state] [right of=0, yshift=-20] (y) {};
	\node[state] [right of=1] (2) {$\ell_{out}^{\mathtt{R},0}$};
	\node[state] [right of=2] (3) {$\ellinit^{\mathtt{Inc},0}$};
	\node[state] [right of=3] (4) {$\ell_{out}^{\mathtt{Inc},0}$};
	\node[state] [right of=4,] (5) {$\ellinit^{\mathtt{R},1}$};
	\node[state] [right of=5] (6) {$\ell_{out}^{\mathtt{R},1}$};
	\node[state] [right of=6, xshift=14] (7) {$\ellinit^{\mathtt{Inc},n}$};
	\node[state] [right = 1.8 of 7] (8) {$\ell_{out}^{\mathtt{Inc},n}$};
	\node[state] [right = of 8] (lout) {$\ell_b$};

	\node[draw, fill = orange, fill opacity = 0.2, text opacity = 1, fit=(1) (2), text height=0.08 \columnwidth] (Rst0) {$\mathtt{Rst_0}$};
	\node[draw, fill = cyan, fill opacity = 0.2, text opacity = 1, fit=(3) (4), text height=0.09 \columnwidth] (Inc0) {$\mathtt{Inc_0}$};
	\node[draw, fill = orange, fill opacity = 0.2, text opacity = 1, fit=(5) (6), text height=0.08 \columnwidth] (Rst1) {$\mathtt{Rst_1}$};
	\node[draw, fill = orange, fill opacity = 0.2, text opacity = 1, fit=(7) (8), text height=0.09 \columnwidth] (Incn) {$\mathtt{Rst_n}$};
	
%	\path (6) -- node[auto=false]{\ldots} (7);
	\draw[->] (6) --node[fill=white,inner sep=1mm]{$\ldots$} (7);

	\path[->] 	
	(0) edge node {$\nop$} (1)

	(2) edge node {$\nop$} (3)

	(4) edge node {$\nop$} (5)
	(8) edge node {$\nop$} (lout)
	;

\end{tikzpicture}
	}
	\caption{$\mathtt{RstInc}$} \label{fig:cover:expspace-hard:rstinc}
\end{figure}

\iflong

\noindent\textbf{The reduction.}
To build the final \NRCM~$N$, we compose the procedural \NBCM~$\mathtt{RstInc}$ with the \testfreeCM~ $M$ in the way described \cref{fig:nrvass}, and
we add to every location $\ell$ of $\mathtt{RstInc}$ and $M$ a restore transition $(\ell, \emptyset,\ellinit')$ which is represented in the figure in an abstract way with dashed arrows, for readability's sake.
\fi
From \cite{lipton76reachability,esparza98decidability}, each procedural machine $\mathtt{TestSwap}_i(\cpt)$ and $\mathtt{Inc}_i$ has size at most $C \times n^2$ for some constant $C$. Hence, observe that $N$ is of size at most $B$ for some $B\in O(|M|^3)$.
One can show that  $(\ellinit, \mathbf{0}_\Counters) \transCM^*_M (\ell_f, v)$ for some $v\in\mathbb{N}^\Counters$, if and  only if $(\ellinit', \mathbf{0}_{\Counters'}) \transNbCM^*_N (\ell_f, v')$ for some $v'\in\mathbb{N}^{\Counters'}$. Using~\cref{th:expspace-hard-lipton}, we obtain:

\begin{theorem}\label{th:expspace-hard}
\CMCover[\NRCM] is \Expspace-hard.
\end{theorem}
%First, we argue that this machine is of polynomial size on $n$.
%\begin{lemma}
%	$N$ is of size at most $B$ for some $B \in O(|M|^3)$.
%\end{lemma}\nas{je n'ai pas relu ca}
%\begin{proof}
%	\cref{remark:expspace-hard:size-lipton-machines} allows us to bound size of each machine $\mathtt{Rst}_i$ by some $c \times n^2$ for some constant $c$. As a result, we can bound the machine $\mathtt{RstInc}$ by a some $d \times n^3$ for some constant $d$, which concludes the proof as $n \in O(|M|)$ and $|M'|\in O(|M|)$.
%\end{proof}
%

%\begin{figure}[t]
%	\includegraphics[width=14cm]{pic/nbrcm.png}
%	\centering
%	\caption{$N$}
%\end{figure}

%Together with \cref{thm:cover-nbcm-in-expspace}, it yields the main result of this section.
%\begin{theorem}\label{thm:cover-nrcm-expspace-complete}
%	\CMCover[\NRCM] is \Expspace-complete.
%\end{theorem}

    \section{Coverability for Rendez-Vous Protocols}\label{sec:cover-rdv-protocols}
    %\subsection{Cover problem}
In this section we prove that \Cover~and \CCover~problems are both \Expspace-complete for rendez-vous protocols. To this end, we present the following  reductions: \CCover~reduces to \CMCover[\NBCM] and \CMCover[\NRCM] reduces to \Cover.
%a reduction from \CCover~to \CMCover[\NRCM], and another one from \CMCover[\NRCM] to \Cover. 
This will prove that \CCover~is in \Expspace~and \Cover~is \Expspace-hard (from \cref{thm:cover-nbcm-in-expspace} and \cref{th:expspace-hard}). As \Cover~is an instance of \CCover, the two reductions suffice to prove \Expspace-completeness for both problems.

%The proof of this theorem is based on a two-senses reduction with the \Cover~problem for \NRCM.
%The rest of the section is devoted to the proof of this theorem. In fact we show that \Cover~over Rendez-Vous Protocols and \CMCover[\NRCM] are inter-reducible.
\subsection{From Rendez-vous Protocols to \NBCM}\label{subsec:rdv-to-nrcm}
\iftrue
\begin{figure}
	\begin{minipage}[t]{0.2\textwidth}
		
		\resizebox{1.5cm}{!}{%
			\begin{tikzpicture}[->, >=stealth', shorten >=0.5pt,node distance=2cm,on grid,auto, initial text = {}] 
	\tikzstyle{initial}= [initial by arrow,initial text=]
	\node[state,initial, minimum width=0.1pt] (0) at (0,0) {$\ellinit$};
%	\node[state] [right of=0] (1) {};
	%\node[state] [right of=0, yshift=-20] (y) {};
	%	\node[state] [right of=1] (2) {};
	%	\node[state] [right of=2] (3) {};
	%	\node[state] [right = of 3] (4) {};
	%	\node[state] [below right = of 0] (5) {$\ell_f$};
	
%	\path[->, bend left=20] 	
%	(0) edge node {$\dec{q}$} (1)
%	(1) edge node {$\inc{q'}$} (0)
%	;
	
		\path[->] 	
		(0) edge [loop above] node {$\inc{\qinit}$} (0)
		%	(0) edge [bend right =30] node [below left] {$\dec{q_f}$} (5)
		
		;
	
	%	\draw[->, dashed] (RstInc.south) |- ++(0,-1) -|(0.south); 
	%	\draw[->, dashed] (CM.south) |- ++(0,-1) -| (0.south);
	%	\node[anchor=east, xshift = -20, yshift=-40] at (CM.south) {Restore transitions};

\end{tikzpicture}
		}
		\caption{\\Incrementing $\qinit$} \label{fig:cover:translation-RDVtoMC-qinit}
	\end{minipage}
	\begin{minipage}[t]{0.3\textwidth}
		
	\resizebox{2.8cm}{!}{%
	\begin{tikzpicture}[->, >=stealth', shorten >=0.5pt,node distance=2cm,on grid,auto, initial text = {}] 
	\tikzstyle{initial}= [initial by arrow,initial text=]
	\node[state,initial, minimum width=0.1pt] (0) at (0,0) {$\ellinit$};
	\node[state] [right of=0] (1) {};
	%\node[state] [right of=0, yshift=-20] (y) {};
%	\node[state] [right of=1] (2) {};
%	\node[state] [right of=2] (3) {};
%	\node[state] [right = of 3] (4) {};
	%	\node[state] [below right = of 0] (5) {$\ell_f$};
	
	\path[->, bend left=20] 	
	(0) edge node {$\dec{q}$} (1)
	(1) edge node {$\inc{q'}$} (0)
	;
	
%	\path[->] 	
%	(0) edge [loop above] node {$\inc{\qinit}$} (0)
%	%	(0) edge [bend right =30] node [below left] {$\dec{q_f}$} (5)
%	
%	;
	
	%	\draw[->, dashed] (RstInc.south) |- ++(0,-1) -|(0.south); 
	%	\draw[->, dashed] (CM.south) |- ++(0,-1) -| (0.south);
	%	\node[anchor=east, xshift = -20, yshift=-40] at (CM.south) {Restore transitions};

\end{tikzpicture}
	}
	\caption{Transitions for\\$(q, \tau, q') \in T$} \label{fig:cover:translation-RDVtoMC-tau}
	\end{minipage}
\begin{minipage}[t]{0.5\textwidth}
	\resizebox{6cm}{!}{%
		\begin{tikzpicture}[->, >=stealth', shorten >=0.5pt,node distance=2cm,on grid,auto, initial text = {}] 
	\tikzstyle{initial}= [initial by arrow,initial text=]
	\node[state,initial, minimum width=0.1pt] (0) at (0,0) {$\ellinit$};
	\node[state] [right of=0] (1) {};
	%\node[state] [right of=0, yshift=-20] (y) {};
	\node[state] [right of=1] (2) {};
	\node[state] [right of=2] (3) {};
	%\node[state] [right = of 3] (4) {};
%	\node[state] [below right = of 0] (5) {$\ell_f$};
	
	\path[->, bend left=20] 	
	(0) edge node {$\dec{q}$} (1)
	(1) edge node {$\dec{p}$} (2)
	(2) edge node {$\inc{q'}$} (3)
	(3) edge node {$\inc{p'}$} (0)
	
	%(4) edge node {} (0)
	;
	
%	\path[->] 	
%	(0) edge [loop above] node {$\inc{\qinit}$} (0)
%%	(0) edge [bend right =30] node [below left] {$\dec{q_f}$} (5)
%	
%	;
%	
%	\draw[->, dashed] (RstInc.south) |- ++(0,-1) -|(0.south); 
%	\draw[->, dashed] (CM.south) |- ++(0,-1) -| (0.south);
%	\node[anchor=east, xshift = -20, yshift=-40] at (CM.south) {Restore transitions};

\end{tikzpicture}
	}	
	\caption{Transitions for a rendez-vous\\ $(q, !a, q')$, $(p, ?a, p') \in T$} \label{fig:cover:translation-RDVtoMC-rdv}
	
\end{minipage}

\vspace*{0.2cm}
\begin{minipage}[t]{0.5\textwidth}
	
	\resizebox{7cm}{!}{%
		\begin{tikzpicture}[->, >=stealth', shorten >=0.5pt,node distance=2cm,on grid,auto, initial text = {}] 
	\tikzstyle{initial}= [initial by arrow,initial text=]
	\node[state,initial, minimum width=0.1pt] (0) at (0,0) {$\ellinit$};
	\node[state] [right of=0] (1) {};
	%\node[state] [right of=0, yshift=-20] (y) {};
	\node[state] [right of=1] (2) {};
	\node[state] [right of=2] (3) {};
	\node[state] [right = of 3] (4) {};
%	\node[state] [below right = of 0] (5) {$\ell_f$};
	
	\path[->, bend left=20] 	
	(0) edge node {$\dec{q}$} (1)
	(1) edge node {$\nbdec{p_1}$} (2)
	(3) edge node {$\nbdec{p_k}$} (4)
	;
	\path[->, bend left = 18] 
	(4) edge node {$\inc{q'}$} (0)
	;
	
%	\path[->] 	
%	(0) edge [loop above] node {$\inc{\qinit}$} (0)
%%	(0) edge [bend right =30] node [below left] {$\dec{q_f}$} (5)
%	
%	;
%	
	\path (2) -- node[auto=false]{\ldots} (3);
	
	%	\draw[->, dashed] (RstInc.south) |- ++(0,-1) -|(0.south); 
	%	\draw[->, dashed] (CM.south) |- ++(0,-1) -| (0.south);
	%	\node[anchor=east, xshift = -20, yshift=-40] at (CM.south) {Restore transitions};

\end{tikzpicture}
	}
	\caption{Transitions for a non-blocking \\sending $(q, !a, q') \in T$ and $R(a) = \{p_1 \dots p_k\}$} \label{fig:cover:translation-RDVtoMC-nb}
\end{minipage}
\begin{minipage}[t]{0.45\textwidth}
	
	\resizebox{7cm}{!}{%
		\begin{tikzpicture}[->, >=stealth', shorten >=0.5pt,node distance=2cm,on grid,auto, initial text = {}] 
	\tikzstyle{initial}= [initial by arrow,initial text=]
	\node[state,initial, minimum width=0.1pt] (0) at (0,0) {$\ellinit$};
	\node[state] [right of=0] (1) {};
	%\node[state] [right of=0, yshift=-20] (y) {};
	\node[state] [right of=1] (2) {};
	\node[state] [right of=2] (3) {};
	\node[state] [right = of 3] (4) {$\ell_f$};
%	\node[state] [ right = of 4] (5) {$\ell_f$};
	
	\path[->, bend left=20] 	
	(0) edge node {$\dec{\mathbf{q}_1}$} (1)
	(1) edge node {$\dec{\mathbf{q}_2}$} (2)
	(3) edge node {$\dec{\mathbf{q}_s}$} (4)

	;
	
%	\path[->] 	
%	(0) edge [loop above] node {$\inc{\qinit}$} (0)
%	%	(0) edge [bend right =30] node [below left] {$\dec{q_f}$} (5)
%	
%	;
%	
	\path (2) -- node[auto=false]{\ldots} (3);
	
	%	\draw[->, dashed] (RstInc.south) |- ++(0,-1) -|(0.south); 
	%	\draw[->, dashed] (CM.south) |- ++(0,-1) -| (0.south);
	%	\node[anchor=east, xshift = -20, yshift=-40] at (CM.south) {Restore transitions};

\end{tikzpicture}
	}	\caption{Verification for the coverability of $C_F = \mset{\mathbf{q}_1} + \mset{\mathbf{q}_2} + \dots + \mset{\mathbf{q}_s}$} \label{fig:cover:translation-RDVtoMC-verif}
\end{minipage}

\end{figure}
\else
\begin{figure}
\resizebox{7cm}{!}{
	\begin{tikzpicture}[->, >=stealth', shorten >=0.5pt,node distance=2cm,on grid,auto, initial text = {}] 
	\tikzstyle{initial}= [initial by arrow,initial text=]
	\node[state,initial, minimum width=0.1pt] (0) at (0,0) {$\ellinit$};
	\node[state] [right of=0] (1) {};
	%\node[state] [right of=0, yshift=-20] (y) {};
	\node[state] [right of=1] (2) {};
	\node[state] [right of=2] (3) {};
	\node[state] [right = of 3] (4) {};
%	\node[state] [below right = of 0] (5) {$\ell_f$};
	
	\path[->, bend left=20] 	
	(0) edge node {$\dec{q}$} (1)
	(1) edge node {$\nbdec{p_1}$} (2)
	(3) edge node {$\nbdec{p_k}$} (4)
	;
	\path[->, bend left = 18] 
	(4) edge node {$\inc{q'}$} (0)
	;
	
%	\path[->] 	
%	(0) edge [loop above] node {$\inc{\qinit}$} (0)
%%	(0) edge [bend right =30] node [below left] {$\dec{q_f}$} (5)
%	
%	;
%	
	\path (2) -- node[auto=false]{\ldots} (3);
	
	%	\draw[->, dashed] (RstInc.south) |- ++(0,-1) -|(0.south); 
	%	\draw[->, dashed] (CM.south) |- ++(0,-1) -| (0.south);
	%	\node[anchor=east, xshift = -20, yshift=-40] at (CM.south) {Restore transitions};

\end{tikzpicture}
	}
		\caption{Transitions for a non-blocking sending $(q, !a, q') \in T$ and $R(a) = \{p_1 \dots p_k\}$ in $M$} \label{fig:cover:translation-RDVtoMC-nb}
	\end{figure}
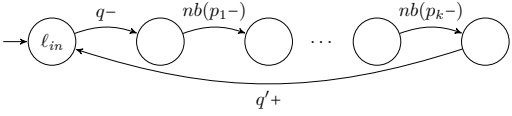
	\fi
%
%In fact, we consider \NBCM~rather than \NRCM, recall that in \cref{subsec:in-expspace}, we not only proved that \CMCover[\NRCM]~is in \Expspace~but that this is also the case for \CMCover[\NBCM]. 
%In this subsection, we present a reduction from \CCover~in rendez-vous protocols to \CMCover[\NBCM]. 
%
Let $\PP = (Q, \Sigma, \qinit, q_f, T)$ a rendez-vous protocol and $C_F$ a configuration of \PP~to be covered. We shall also decompose $C_F$ as a sum of multisets $\mset{\mathbf{q}_1} + \mset{\mathbf{q}_2} + \dots + \mset{\mathbf{q}_s}$. Observe that there might be $\mathbf{q}_i=\mathbf{q}_j$ for $i\neq j$. 
We build the \NBCM~$M = (\Loc, \Counters, \Delta_b, \Delta_{nb}, \ellinit)$ \iflong described  in \cref{fig:cover:translation-RDVtoMC-rdv,fig:cover:translation-RDVtoMC-nb,fig:cover:translation-RDVtoMC-tau,fig:cover:translation-RDVtoMC-verif}. Here, \else with \fi $\Counters=Q$. A configuration $C$ of $\PP$ is meant to be represented in $M$
by $(\ellinit,v)$, with $v(q)=C(q)$ for all $q\in Q$. The only meaningful location of $M$ is then $\ellinit$. The other ones are here to ensure correct updates of the 
counters when simulating a transition. 
We let $\Loc = \{\ellinit\}\cup \{\ell_{(t,t')}^1, \ell_{(t,t')}^2,\ell_{(t,t')}^3\mid t=(q,!a,q'), t'=(p,?a,p')\in T\}\cup\{\ell_t, \ell_{t,p_1}^a,\cdots,\ell_{t,p_k}^a\mid
t=(q,!a,q')\in T, \Read{a}=\{p_1,\dots, p_k\}\}\cup\{\ell_q\mid t=(q,\tau,q')\in T\} \cup \{\ell_1 \dots \ell_{s}\}$, with final location $\ell_f = \ell_s$, where $\Read{m}$ for a message $m \in \Sigma$ has been defined in \cref{section:definition-rdv}. The sets $\Delta_b$ and $\Delta_{nb}$ are shown \cref{fig:cover:translation-RDVtoMC-qinit,fig:cover:translation-RDVtoMC-rdv,fig:cover:translation-RDVtoMC-nb,fig:cover:translation-RDVtoMC-tau,fig:cover:translation-RDVtoMC-verif}. Transitions pictured \cref{fig:cover:translation-RDVtoMC-qinit,fig:cover:translation-RDVtoMC-rdv,fig:cover:translation-RDVtoMC-tau,fig:cover:translation-RDVtoMC-verif} show how to
simulate a rendez-vous protocol with the classical rendez-vous mechanism. %
The non-blocking
rendez-vous are handled by the transitions pictured~\cref{fig:cover:translation-RDVtoMC-nb}\iflong (where the
only non-blocking transitions of the \NBCM~occur): to simulate the occurrence of $(q,!a,q')$, the \NBCM~M decrements the value of $q$ by
a transition of the form $(3)$. It then takes a sequence of non-blocking decrements for each state in $\Read{a}$. The last transition of the simulation of a non-blocking rendez-vous is to increment the counter $q'$ by a transition of the form (3).\else. \fi If the \NBCM~$M$ faithfully simulates $\PP$, then this loop of non-blocking
decrements is taken when the values of the counters in $\Read{a}$ are equal to 0, and the configuration reached still corresponds to a configuration in $\PP$. 
However, it could be that this loop is taken in $M$ while some counters in $\Read{a}$ are strictly positive. In this case, a blocking rendez-vous has to be taken in $\PP$, e.g.\ $(q,!a,q')$ and $(p,?a,p')$ if the counter $p$ in $M$ is strictly positive. Therefore, the value of the reached configuration $(\ellinit, v)$ and the corresponding configuration $C$ in $\PP$
will be different: first, $C(p')>v(q')$, since the process in $p$ has moved in the state $p'$ in $\PP$ when there has been no increment of $p'$ in $M$. Furthermore, 
all other non-blocking decrements of counters in $\Read{a}$ in $M$ may have effectively decremented the counters, when in $\PP$ no other process has left a state of $\Read{a}$. However, this ensures that $C\geq v$. The reduction then ensures that if $(\ellinit, v)$ is reachable in $M$, then a configuration $C\geq v$ is reachable in $\PP$. \color{black} Then, if it is possible to reach a configuration $(\ellinit, v)$ in $M$ whose counters are high enough to cover $\ell_F$, then 
the corresponding initial execution in $\PP$ will reach a configuration $C\geq v$, which \color{teal} hence \color{black} covers $C_F$. 

\begin{theorem}\label{cor:ccover-expspace}
\CCover~over rendez-vous protocols is in \Expspace.
\end{theorem}

\subsection{From \NRCM~to Rendez-Vous Protocols}\label{subsec:nrcm-to-rdv}
%In this part, we prove that, from a \NRCM~and a final location, one can build a rendez-vous protocol such that the protocol is a positive instance for the \Cover~problem iff the \NRCM~is one aswell. As the \Cover~problem is \Expspace-hard for \NRCM, this shall prove that the \Cover~problem is \Expspace-hard for rendez-vous protocols aswell.
%
%In this section we present a reduction from \CMCover[\NRCM] to \Cover (in rendez-vous protocols). The reduction mainly relies on the mechanism which can isolate a process on a set of states as we have seen on the protocol of \cref{fig-rdv2}~for state $\qfive$. 
%The protocol shall ensure that for a set of states and for every reachable configuration, there is at most one process in this set of states. This process will simulate the machine execution, as some other (unbounded) set of processes shall simulate counters values. We now present the reduction.
%
The reduction from \CMCover[\NRCM]~to \Cover~in rendez-vous protocols mainly relies on the mechanism that can ensure that at most one process evolves in some given set of states, as explained in~\cref{example-verif-pbs}. This will allow to somehow select a ``leader'' among the processes that will simulate the behaviour
of the \NRCM~whereas other processes will simulate the values of the counters. 
Let  $M = (\Loc, \Counters, \Delta_b, \Delta_{nb}, \ellinit)$ a \NRCM~and $\ell_f \in \Loc$ a final target location.
We build the rendez-vous protocol $\PP$ pictured in \cref{fig:cover:protocol-MCtoRDV}, where $\PP(M)$ is the part that will simulate the \NRCM~$M$. The locations
$\{1_\cpt\mid \cpt\in\Counters\}$ will allow to encode the values of the different counters during the execution: for a configuration $C$, $C(1_\cpt)$ will represent
the value of the counter $\cpt$.
We give then $\PP(M)=(Q_M,\Sigma_M,\ellinit,\ell_f,T_M)$ with
\iflong
\begin{align*}
Q_M &= \Loc\cup\{\ell_\delta\mid \delta\in\Delta_b\}\\
\Sigma_M&= \{\textrm{inc}_\cpt,\overline{\textrm{inc}}_\cpt, \textrm{dec}_\cpt, \overline{\textrm{dec}}_\cpt, \textrm{nbdec}_\cpt\mid \cpt\in\Counters\}\\
T_M&=\{(\ell_i,!\textrm{inc}_\cpt, \ell_\delta), (\ell_\delta, ?\overline{\textrm{inc}}_\cpt, \ell_j)\mid \delta=(\ell_i, \inc{\cpt}, \ell_j)\in\Delta_b\}\\
	 &\cup \{(\ell_i, !\textrm{dec}_\cpt, \ell_\delta), (\ell_\delta, ?\overline{\textrm{dec}}_\cpt, \ell_j)\mid \delta = (\ell_i, \dec{\cpt}, \ell_j)\in\Delta_b\}\\
	 &\cup\{(\ell_i, !\textrm{nbdec}_\cpt, \ell_j)\mid (\ell_i, \nbdec{x},\ell_j)\in\Delta_{nb}\}\\
	 &\cup\{(\ell_i, \tau, \ell_j)\mid (\ell_i, \nop,\ell_j)\in\Delta_{b}, \ell_j \ne \ellinit\}
\end{align*}
\else
$Q_M=\Loc\cup\{\ell_\delta\mid \delta\in\Delta_b\}$, $\Sigma_M= \{\textrm{inc}_\cpt,\overline{\textrm{inc}}_\cpt, \textrm{dec}_\cpt, \overline{\textrm{dec}}_\cpt, \textrm{nbdec}_\cpt\mid \cpt\in\Counters\}$, and $T_M=\{(\ell_i,!\textrm{inc}_\cpt, \ell_\delta),(\ell_\delta, ?\overline{\textrm{inc}}_\cpt, \ell_j)\mid \delta=(\ell_i, \inc{\cpt}, \ell_j)\in\Delta_b\}\cup \{(\ell_i, !\textrm{dec}_\cpt, \ell_\delta), (\ell_\delta, ?\overline{\textrm{dec}}_\cpt, \ell_j)\mid$ $ \delta = (\ell_i, \dec{\cpt}, \ell_j) \linebreak[0]\in\Delta_b\}\cup\{(\ell_i, !\textrm{nbdec}_\cpt, \ell_j)\mid (\ell_i, \nbdec{\cpt},\ell_j)\in\Delta_{nb}\}\cup\{(\ell_i, \tau, \ell_j)\mid (\ell_i, \nop,\ell_j)\in\Delta_{b}\}$. Here, the reception of a message $\overline{\textrm{inc}}_\cpt$
(respectively $\overline{\textrm{dec}}_\cpt$) works as an acknowledgement, ensuring that a process has indeed received the message $\textrm{inc}_\cpt$ (respectively $\textrm{dec}_\cpt$), and that the corresponding counter has been incremented (resp.\ decremented). For 
non-blocking decrement, obviously no acknowledgement is required. 
%given by the smallest set that contains :
%\begin{itemize}
%\item for all $\delta=(\ell_i, \inc{\cpt}, \ell_j)\in\Delta_b$, $(\ell_i,!\textrm{inc}_\cpt, \ell_\delta)$ and $(\ell_\delta, ?\overline{\textrm{inc}}_\cpt, \ell_j)\in T_M$, with the 
%reception of a message $?\overline{\textrm{inc}}_\cpt$ working as an acknowledgement, ensuring that a process has indeed received the message $!\textrm{inc}_\cpt$, and that the counter $\cpt$ has been incremented.
%\item for all $\delta = (\ell_i, \dec{\cpt}, \ell_j)\in\Delta_b$, $(\ell_i, !\textrm{dec}_\cpt, \ell_\delta)$ and $(\ell_\delta, ?\overline{\textrm{dec}}_\cpt, \ell_j)\in T_M$,
%with, again, the 
%reception of a message $?\overline{\textrm{dec}}_\cpt$ ensuring that a process has indeed received the message $!\textrm{dec}_\cpt$, and that the counter $\cpt$ has been decremented.
%\item for all $(\ell_i, \nbdec{\cpt},\ell_j)\in\Delta_{nb}$, $(\ell_i, !\textrm{nbdec}_\cpt, \ell_j)\in T_M$. Here, no acknowledgement is needed since we implement a 
%non-blocking decrement.
%\item for all $(\ell_i, \nop,\ell_j)\in\Delta_{b}$, $(\ell_i, \tau, \ell_j)\in T_M$.
%\end{itemize}
\fi 
\iflong We define $\PP=(Q,\Sigma,T,\qinit, \ell_f)$ as follows.
\begin{align*}
Q &= Q_M\cup \{1_\cpt, q_\cpt, q'_\cpt\mid \cpt\in\Counters\}\cup \{\qinit, q, q_\bot\}\\
\Sigma &= \Sigma_M\cup \{L, R\}\\
T&=T_M\cup \{(\qinit, !L, q), (q, !R, \ellinit), (q, ?L, q_\bot)\} \cup \{(\ell, ?L, q_\bot)\mid \ell\in Q_M\}\\
 & \cup \{(\qinit, ?\textrm{inc}_\cpt, q_\cpt), (q_\cpt, !\overline{\textrm{inc}}_\cpt, 1_\cpt), (1_\cpt, ?\textrm{dec}_\cpt, q'_\cpt), (q'_\cpt, !\overline{\textrm{dec}}_\cpt, \qinit), (1\cpt, ?\textrm{nbdec}_\cpt, \qinit)\mid
 \cpt\in \Counters\}\\
& \{(q_\cpt, ?R, \qinit), (q'_\cpt, ?R, \qinit)\mid\cpt\in\Counters\}
\end{align*}
\fi
\iflong\else The protocol $\PP=(Q,\Sigma,\qinit, \ell_f,T)$ is then defined with $Q= Q_M\cup \{1_\cpt, q_\cpt, q'_\cpt\mid \cpt\in\Counters\}\cup \{\qinit, q, q_\bot\}$,
$\Sigma=\Sigma_M\cup \{L, R\}$ and $T$ is the set of transitions $T_M$ along with the transitions pictured in~\cref{fig:cover:protocol-MCtoRDV}. Note 
that there is a transition $(\ell,?L,q_\bot)$ for all $\ell\in Q_M$.\fi
%\begin{figure}[t]
%	\includegraphics[width=12cm]{pic/cover-protocol-MC-RDV.png}
%	\centering
%	\caption{$\PP$}
%\end{figure}
\begin{figure}[t]
	
	\resizebox{!}{2.8cm}{%
		\begin{tikzpicture}[->, >=stealth', shorten >=0.5pt,node distance=2cm,on grid,auto, initial text = {}] 
	\tikzstyle{initial}= [initial by arrow,initial text=]
	\node[state,initial, minimum width=0.1pt] (0) at (0,0) {$\qinit$};
	\node[state] [below left of=0, xshift = -10, yshift = 10] (1) {$q_\cpt$};
	%\node[state] [right of=0, yshift=-20] (y) {};
	\node[state] [below right of=0, xshift = 10, yshift = 10] (2) {$q'_\cpt$};
	\node[state] [below =2 of 0] (3) {$1_\cpt$};
	\node[state] [right =4 of 0] (q) {$q$};
	\node[state] [right = of q] (li) {$\ellinit$};
	\node[state] [right = of li] (lf) {$\ell_f$};
	\node[state] [below = 1.5 of li] (d) {$q_\bot$};
	
	\node[draw, fill = teal, fill opacity = 0.1, text opacity = 1, fit=(li) (lf), text height=0.06 \columnwidth] (P) {$\PP(M)$};
	
	%	\path (6) -- node[auto=false]{\ldots} (7);

	\path[->] 	
	(0) edge node {$!L$} (q)
	(q) edge node {$!R$} (li)
	(q) edge [bend right] node {$?L$} (d)
	(P) edge [bend left] node {$?L$} (d)
	
	(0) edge [bend right= 13] node [above left ] {$?\textrm{inc}_\cpt$} (1)
	(1) edge [bend right= 13] node [below  ] {$?R$} (0)
	(2) edge [bend left= 13] node [below  ] {$?R$} (0)
	(2) edge [bend right = 13] node [right  ] {$!\overline{\textrm{dec}}_\cpt$} (0)
	(1) edge [bend right= 13] node [left , yshift = -5] {$!\overline{\textrm{inc}}_\cpt$} (3)
	(3) edge [bend right= 13] node [right, yshift = -5 ] {$?\textrm{dec}_\cpt$} (2)
%	
%	(2) edge node {} (3)
%	
	;
		\draw[->] ( [xshift = 2, yshift = -6]3.west) -| ++(-2.5,0) |- ++ (0, +2.9)-|(0.north); 
%	\draw[->, dashed] (CM.south) |- ++(0,-1) -| (0.south);
	\node[anchor=east, xshift = -80, yshift=70] at (3.south) {$?\textrm{nbdec}_\cpt$};

\end{tikzpicture}
	}
	\caption{The rendez-vous protocol $\PP$ built from the \NRCM~$M$. Note that there is one gadget with states $\{q_\cpt$,
	$q'_\cpt$, $1_\cpt\}$ for each counter $\cpt\in\Counters$.}
 \label{fig:cover:protocol-MCtoRDV}
 \end{figure}

With two non-blocking transitions on $L$ and $R$ at the beginning, protocol $\PP$ can faithfully simulate the \NRCM~$M$ without further ado, provided that the initial configuration contains enough processes to simulate all the counters values during the execution: after having sent a process in state $\ellinit$, any transition of $M$ can be simulated in $\PP$\color{black}. Conversely, an initial execution of $\PP$ can send multiple processes into
the $\mathcal{P}(M)$ zone, which can mess up the simulation. However, each new process entering $\mathcal{P}(M)$ will send the message $L$, which will send the process already in
$\{q\}\cup Q_M$ in the deadlock state $q_\bot$, and send the message $R$, which will be received by any process in $\{q_\cpt,q'_\cpt\mid \cpt\in\Counters\}$. Moreover, the construction of the protocol ensures that there can only be one process in the set of states 
$\{q_\cpt,q'_\cpt\mid \cpt\in\Counters\}$. Then, if we have reached a configuration simulating the configuration $(\ell, v)$ of $M$, sending a new process in the $\mathcal{P}(M)$ zone will lead to a configuration $(\ellinit, v)$, and hence simply mimicks a restore transition of $M$. \color{black} So every initial execution of $\PP$ corresponds to an initial execution of $M$.

\begin{theorem}\label{th:ccover-expspace-complete}
	\Cover~\ifccover{and \CCover} over rendez-vous protocols are \Expspace~complete.
\end{theorem}

    \section{Coverability for Wait-Only Protocols}
    \label{sec:wo}

In this section, we study a restriction on rendez-vous protocols in which
we assume that a process waiting to answer a rendez-vous cannot perform another
action by itself. This allows for a polynomial
time algorithm for solving \CCover.

\subsection{Wait--Only Protocols}

We say that a protocol $\PP = (Q, \Sigma, \qinit,
q_f, T)$ is \emph{wait-only} if the set of states $Q$ can be
partitioned into $Q_A$ \textemdash \ the \emph{active states} \textemdash \ and $Q_W$ \textemdash\ the \emph{waiting} states \textemdash\ with $\qinit\in Q_A$ 
 and:
 \vspace*{-0.2cm}
\begin{itemize}
  \item for all $q \in Q_A$, for all $(q',?m,q'')\in T$, we have $q'\neq q$;
  %there
%does not exist in $T$ a transition of the form $(q,?m,q') \in T$, and,
\item for all $q\in Q_W$, for all $(q', !m, q'') \in T$, we have $q' \neq q$ and for all $(q', \tau, q'') \in T$, we have $q'\neq q$.\color{black}
%there exists $q'\in Q$ and $m\in\Sigma$ such
%  that $(q,?m,q')\in T$ and there does not exist $q'' \in Q$ such that
%  $(q,\tau,q'') \in T$ or $(q,!m',q'')\in T$ for some $m' \in \Sigma$.
\end{itemize}
From a waiting state, a process can only perform receptions (if it can perform anything), whereas in an active state, a process can only perform internal actions or send messages. \color{black}
%We call \emph{active states} the states in $Q_A$ and \emph{waiting
%  states} the states in $Q_W$.
%Hence, with such protocols, when a process is in a
%waiting state from $Q_W$, it is not able to request rendez-vous nor
%to perform an internal action. 
Examples of wait-only protocols are
given by Figures \ref{fig-example-wo} and \ref{fig-example-wo-2}.

In the sequel, we will often refer to the paths of the underlying
graph of the protocol. Formally,  a
\emph{path} in a  protocol $\PP = (Q, \Sigma, \qinit,
q_f, T)$ is either a control state $q \in Q$ or a finite sequence of
transitions in $T$ of the form
$(q_0,a_0,q_1)(q_1,a_1,q_2)\ldots(q_k,a_k,q_{k+1})$, the first case
representing a path from $q$ to $q$ and the second one from $q_0$ to
$q_{k+1}$.

\subsection{Abstract Sets of Configurations}

To solve the coverability problem for wait-only protocols in polynomial time, we
rely on a sound and complete
abstraction of the set of reachable configurations. In the sequel, we consider a wait-only protocol $\PP = (Q, \Sigma, \qinit,
q_f, T)$ whose set of states is partitioned into a set of active
states $Q_A$ and a set of waiting states $Q_W$. An \emph{abstract set of
configurations} $\gamma$ is a pair $(S,\Toks)$ such that:
\vspace*{-0.2cm}
\begin{itemize}
\item $S \subseteq Q$ is a subset of states, and,
\item $\Toks \subseteq Q_W \times \Sigma$ is a subset of pairs 
  composed of a waiting state and a message, and,
\item $q \not\in S$ for all $(q,m) \in \Toks$.
\end{itemize}
We then abstract the set of reachable configurations as a set
of states of the underlying protocol. However, 
as we have seen, some states, like states in $Q_A$, can host an unbounded number of processes together (this will
be the states in $S$), while some states can only host a bounded number (in fact, 1) of processes together (this will be the states stored in $\Toks$).
This happens when a waiting state $q$ answers a rendez-vous $m$, that has necessarily been requested for a process to be in $q$. 
%for waiting states that can answer requests on a message $m$, message that is necessarily sent for a process to be in this state. 
%Hence, we remember this message along with the state in the set $\Toks$. 
%The intuition for this abstraction is that, in a wait-only
%protocol, there are some states that can contain an unbounded number
%of processes (states in $S$), this is for instance the case of all the
%reachable active states, and other states whose number of processes they can contain
%at any given point is bounded -- 
%%that will be able to receive
%%at any moment a bounded number of processes, 
%these are the states
%appearing in $\Toks$. Furthermore for these states, we shall see that
%the bound is $1$ (this will be a consequence of the correction of our
%abstraction). 
Hence, in $\Toks$, along with a state $q$, we remember the last message $m$ having been sent in the path leading from $\qinit$ to $q$, which is necessarily in $Q_W$. Observe that, since several paths can lead to $q$, 
  there can be $(q,m_1),(q,m_2)\in\Toks$ with $m_1\neq m_2$.  We denote
  by $\Gamma$ the set of abstract sets of configurations. 

Let $\gamma=(S,\Toks)$ be an abstract set of
configurations. Before we go into the configurations represented by
$\gamma$, we need some preliminary definitions. We note $\mst(\mathit{{\kern-1pt}\Toks})$ the set $\set{q \in Q_W
  \mid\textrm{there exists } m\in \Sigma\textrm{ such that }(q,m) \in \Toks}$ of control states
appearing in $\Toks$. Given a state $q \in Q$, we let
$\Rec{q}$ be the set $\set{ m \in \Sigma \mid\textrm{there exists } q'\in Q \textrm{ such that }
(q,?m,
  q') \in T}$ of messages that can be received in state $q$ (if $q$ is
not a waiting state, this set is empty). Given two different waiting states $q_1$ and $q_2$ in
$\starg{\Toks}$, we say $q_1$ and $q_2$ are \emph{conflict-free} in
$\gamma$ if there exist $m_1,m_2 \in \Sigma$ such that $m_1 \neq m_2$, 
$(q_1,m_1),(q_2,m_2) \in \Toks$ and $m_1 \notin \Rec{q_2}$ and
$m_2 \notin \Rec{q_1}$. We now say that a configuration $C\in\CC(\PP)$ \emph{respects}
$\gamma$ if and only if for all $q \in Q$ such that $C(q)>0$ one of the following two
conditions holds:
\begin{enumerate}
	\vspace*{-0.2cm}
\item \label{ccover-wo-consistency-1} $q \in S$, or,
\item \label{ccover-wo-consistency-2}$q \in \starg{\Toks}$ and $C(q)=1$ and for all $q' \in \starg{\Toks} \setminus\set{q}$ such that
  $C(q')=1$, we have that $q$ and $q'$ are conflict-free.
\end{enumerate}

Note that the condition is on states $q$ such that $C(q) > 0$ and not all states $q \in Q$ because it might be that some states don't appear in $S\cup st(Toks)$ (non-reachable states for instance).
\color{black}
Let
$\Interp{\gamma}$ be the set of configurations respecting $\gamma$. Note
that in $\Interp{\gamma}$, for $q$ in $S$ there is no restriction on
the number of processes that can be put in $q$  and if $q$ in
$\starg{\Toks}$, it can host at most one process. Two
states from $\starg{\Toks}$  can both host
a process if they are conflict-free.

%We need a last notion to characterise the manipulated sets of
%configurations. We restrict indeed our reasoning to 
Finally, we will only consider abstract sets of configurations that
are \emph{consistent}. This property aims to ensure that concrete configurations
that respect it are indeed reachable from states of $S$.
%, which means that the token
%$(q,m) \in \Toks$ should really come from a 'feasible path' in the
%protocol starting by the request of a rendez-vous $!m$ and followed by
%reception of messages that can effectively be emitted. We as well add
%a property to ensure that if $q_1$ and $q_2$ are conflict-free thanks
%to $(q_1,m_1)$ and $(q_2,m_2)$ and
%$q_1$ and $q_3$ are conflict-free thanks to $(q_1m'_1)$ and $(q_3,m_3)$ then $q_1$ and $q_3$ are conflict-
%free thanks to $(q_1,m_1)$ and $(q_3,m_3)$ too. 
Formally, we
say that an abstract set of
configurations $\gamma=(S,\Toks)$ is \emph{consistent} if $(i)$ for all $(q,m) \in
\Toks$, there exists a path
$(q_0,a_0,q_1)(q_1,a_1,q_2)\ldots(q_k,a_k,q)$ in $\PP$
 such that $q_0
\in S$ and $a_0=\ !m$
and for all $1\leq i \leq k$, we have that $a_i=\ ?m_i$ and that  there
exists $(q'_i,!m_i,q''_i) \in T$ with $q'_i \in S$, and $(ii)$ for two tokens $(q,m), (q',m') \in \Toks$ either $m\in\Rec{q'}$ and $m'\in\Rec{q}$, or, $m\notin\Rec{q'}$ and $m'\notin\Rec{q}$.
Condition $(i)$ ensures that processes in $S$ can indeed lead to a process in the states from $\starg{\Toks}$. Condition $(ii)$ ensures that if in a configuration $C$,
 some states in $\starg{\Toks}$ are pairwise conflict-free, then they can all host a process together.

\begin{lemma}\label{lem:interp-cover-check}
Given  $\gamma\in \Gamma$ and a configuration $C$, there exists $C' \in
\Interp{\gamma}$ such  that $C' \geq C$ if and only if $C \in
\Interp{\gamma}$. Checking that $C\in\Interp{\gamma}$ can be done in polynomial time.
\end{lemma}

\subsection{Computing Abstract Sets of Configurations}

Our polynomial time algorithm is based on the computation of a
polynomial length sequence of consistent abstract sets of
configurations leading to a final
abstract set characterising in a sound and complete manner (with
respect to the coverability problem), an abstraction for the set of
reachable configurations. This will be achieved by a function $F:\Gamma \to \Gamma$,
that inductively computes  this final abstract set
starting from $\gamma_0=(\set{\qinit}, \emptyset)$. %For this matter, we rely on a function
%$F:\Gamma \mapsto \Gamma$ which allows to increase in a certain sense our abstract set of
%configurations. Our sequence will then start with the abstract set of
%configurations $(\set{q_{in}},\emptyset)$ and will be built by
%applying the function $F$ successively until saturation.

\iftable
\begin{table}[t]
\begin{center}

\label{tab:S''}
\makebox[\textwidth]{%
\scalebox{1}{
\begin{tabular}{ p{13.5cm}}
\toprule
\textbf{Construction of intermediate states $S''$ and $\Toks''$}\\
\midrule
\vspace*{-0.5cm}
\begin{enumerate}[itemsep=-0cm,itemindent=-0.2cm]
	\item $S\subseteq S''$ and $\Toks\subseteq \Toks''$
	\item \label{ccover-wo-F-cond-internal}for all $(p,\tau,p') \in T$ with $p \in S$, we have $p' \in S''$
	\item for all $(p,!a,p') \in T$ with $p \in S$, we have: 
\vspace{-0.2cm}
	\begin{enumerate}[itemsep=0cm,itemindent=-0.2cm]
		\item $p' \in S''$ if $a \notin \Rec{p'}$ or if there exists
		$(q,?a,q') \in T$ with $q \in S$;\label{ccover-wo-F-cond-send-S}
		\item $(p',a) \in \Toks''$ otherwise (i.e. when $a \in \Rec{p'}$ and for all $(q,?a,q') \in T$, $q \notin S$);\label{ccover-wo-F-cond-newtok}
	\end{enumerate}
	\item for all $(q,?a,q') \in T$ with $q \in S$ or $(q,a) \in \Toks$, we have $q' \in
	S''$ if there exists $(p,!a,p') \in T$ with $p \in S$;\label{ccover-wo-F-cond-reception-S}
	\item for all $(q,?a,q') \in T$ with $(q,m) \in \Toks$ with $m
	\neq a$, if there exists $(p, !a, p') \in T$ with $p \in S$, we have:\label{ccover-wo-F-cond-tok}
\vspace{-0.2cm}
	\begin{enumerate}[itemsep=0cm,itemindent=-0.2cm]
		\item $q' \in S''$ if $m \notin \Rec{q'}$;\label{ccover-wo-F-cond-tok-end}
		\item $(q',m) \in \Toks''$ if $m \in \Rec{q'}$.\label{ccover-wo-F-cond-tok-step}
	\end{enumerate}
	\vspace*{-0.5cm}
\end{enumerate}
%1. $S\subset S''$ and $\Toks\subseteq \Toks''$\\
%2. for all $(p,\tau,p') \in T$ with $p \in S$, we have $p' \in S''$\\
%3. for all $(p,!a,p') \in T$ with $p \in S$, we have: \\
%\textbf{(a)}  $p' \in S''$ if $a \notin \Rec{p'}$ or if there exists
%$(q,?a,q') \in T$ with $q \in S$;\\
% \textbf{(b)} $(p',a) \in \Toks''$ otherwise (i.e. when $a \in \Rec{p'}$ and there  does not exists $(q,?a,q') \in T$ with $q \in S$);\\
%4. for all $(q,?a,q') \in T$ with $q \in S$ or $(q,a) \in \Toks$, we have $q' \in
% S''$ if there exists $(p,!a,p') \in T$ with $p \in S$;\\
%5.  for all $(q,?a,q') \in T$ with $(q,m) \in \Toks$ with $m
%\neq a$, we have: \\
%\textbf{(a)} $q' \in S''$ if $m \notin \Rec{q'}$ and  there exists  $(p,!a,p') \in T$ with $p \in
%        S$;\\
%\textbf{(b)} $(q',m) \in \Toks''$ if $m \in \Rec{q'}$ and  there exists
%    $(p,!a,p') \in T$ with $p \in S$.\\%\label{ccover-wo-F-cond-tok-step}
\\
 \toprule
%\end{enumerate}
\end{tabular}
%}}
}
}
\caption{{Definition of $S'', \Toks''$ for $\gamma=(S,\Toks)$.}}\label{table:F}
\end{center}
\vspace*{-1cm}
\end{table}

Formal definition of the function $F$ relies on intermediate sets
$S''\subseteq Q$ and $\Toks''\subseteq Q_W \times\Sigma$, which are the smallest sets satisfying the conditions described in \cref{table:F}. 
From $S$ and $\Toks$, rules described in \cref{table:F} add states and tokens to $S''$ and $\Toks''$ from the outgoing transitions from states in $S$ and $\mst(\Toks)$.
It must be that every state added to $S''$ can host an unbounded number of processes, and every state added to $\Toks''$ can host at least one process, furthermore, two conflict-free states in $\Toks''$ should be able to host at least one process at the same time.  
\color{black}

\else
We now provide the formal definition of this function.  For an
abstract set of configurations $\gamma=(S,\Toks)$, we will have
$\gamma'=F(\gamma)$ if and only if $\gamma'=(S',\Toks')$ where $S'$
and $\Toks'$ are built as follows. First we use some intermediate sets
of states
$S'' \subseteq Q$ and $\Toks'' \subseteq Q_W \times \Sigma$ which are
the smallest sets 
satisfying the following conditions  $S \subseteq S''$ and $\Toks
\subseteq \Toks''$ and:
\begin{enumerate}
\item \label{ccover-wo-F-cond-internal}for all $(p,\tau,p') \in T$ with $p \in S$, we have $p' \in S''$;
\item \label{ccover-wo-F-cond-send-S}  for all $(p,!a,p') \in T$ with $p \in S$, we have: \textbf{(a)}  $p' \in S''$ if $a \notin \Rec{p'}$ or if there exists
    $(q,?a,q') \in T$ with $q \in S$; \textbf{(b)} $(p',a) \in \Toks''$ otherwise (i.e. when $a \in \Rec{p'}$ and there
       does not exists $(q,?a,q') \in T$ with $q \in S$);
   \item \label{ccover-wo-F-cond-reception-S}for all $(q,?a,q') \in T$ with $q \in S$ or $(q,a) \in \Toks$, we have $q' \in
     S''$ if there exists $(p,!a,p') \in T$ with $p \in S$;
    \item for all $(q,?a,q') \in T$ with $(q,m) \in \Toks$ with $m
      \neq a$, we have: \label{ccover-wo-F-cond-tok} \textbf{(a)} $q' \in S''$ if $m \notin \Rec{q'}$ and  there exists
          $(p,!a,p') \in T$ with $p \in
          S$;\label{ccover-wo-F-cond-tok-end} \textbf{(b)} $(q',m) \in \Toks''$ if $m \in \Rec{q'}$ and  there exists
          $(p,!a,p') \in T$ with $p \in S$.\label{ccover-wo-F-cond-tok-step}
\end{enumerate}
We have then that $S'$ is the smallest set including $S''$ and such
that:
\begin{enumerate}
  \setcounter{enumi}{5}
  
\item \label{ccover-wo-F-cond-2toks-1} for all $(q_1, m_1), (q_2, m_2) \in \Toks''$ such that $m_1
    \ne m_2$ and $m_2 \notin \Rec{q_1}$ and  $m_1 \in \Rec{q_2}$, we
    have $q_1 \in S'$;
 \item \label{ccover-wo-F-cond-3toks-1}for all $(q_1, m_1), (q_2, m_2), (q_3,m_2) \in \Toks''$ s.t $m_1 \ne m_2$
   and $(q_2, ?m_1, q_3) \in T$, we have $q_1 \in S'$;
  \item \label{ccover-wo-F-cond-3toks-2}for all $(q_1, m_1), (q_2, m_2), (q_3, m_3) \in \Toks''$ such
    that $m_1 \ne m_2$ and  $m_1\ne m_3$ and $m_2 \ne m_3$ and 
    $m_1 \notin \Rec{q_2}$, $m_1 \in \Rec{q_3}$ and  $m_2\notin \Rec{q_1}$, $m_2 \in \Rec{q_3}$, and $m_3 \in \Rec{q_2}$ and $m_3 \in \Rec{q_1}$,
    we have $q_1 \in S'$.
\end{enumerate}
And finally $\Toks'=\set{(q,m) \in \Toks'' \mid q \not\in
  S'}$.
\fi

\begin{figure}[htbb]
  \begin{minipage}[c]{.49\columnwidth}
	\resizebox*{!}{2.3cm}{
	\begin{tikzpicture}[->, >=stealth', shorten >=1pt,node distance=2cm,on grid,auto, initial text = {}] 
	\node[state, initial above] (q0) {$\qinit$};
	\node[state] (q1) [ left = of q0] {$q_1$};
	\node[state] (q2) [ left = of q1, yshift = -0.75cm] {$q_2$};
	\node[state] (q3) [below = 1.5 of q1] {$q_3$};
	\node[state] (q4) [ below = 1.5 of q0] {$q_4$};
	\node[state] (q5) [ right = of q0] {$q_5$};
	\node[state] (q6) [ below = 1.5 of q5] {$q_6$};
	\node[state] (q7) [ right = of q5] {$q_7$};

%	\node[state] (q6) [ left = of q5] {$q_3$};
%	\node[state] (q6) [ right = of q5] {$q_6$}; 
%	\node[state] (q7) [ right = of q6] {$q_7$}; 	
%	\node[state] (q8) [ right =of q5, xshift = 1cm, yshift = -1cm] {$q_8$};
%	\node[state] (q9) [ right = of q8 ] {$q_6$};
%	\node[state] (q10) [  right= of q8, yshift = -30] {$q_6$};
%	
%	\node[state] (q2) [right = of q1] {$q_2$};
	
	\path[->] 
	(q0) edge [bend right = 15] node [above] {$!a$} (q1)
        	edge [bend left = 15] node {$!b$} (q1)
        	edge node {$!d$} (q4)
        	edge node {$!c$} (q5)
	(q1) edge [bend right = 15] node [above] {$?a,?b$} (q2)
			edge node {$?c$} (q3)
	(q3) edge [bend left = 15] node {$?a,?b$} (q2)

	(q5) edge node {$?c$} (q6)
	(q5) edge node {$?d$} (q7)
%	edge   node [above] {$?b$} (q4)
%	edge   node  {$!b$} (q6)
%	(q6) edge [bend right] node {$?c$} (q2)
%	%					 edge  [bend right] node [below left] {$?d$} (q6)
	;

\end{tikzpicture}
  }
  \caption{Wait-only protocol $\PP_1$.}\label{fig-example-wo}
\end{minipage}
\begin{minipage}[c]{.49\columnwidth}
  \resizebox*{!}{2.3cm}{
	\begin{tikzpicture}[->, >=stealth', shorten >=1pt,node distance=2cm,on grid,auto, initial text = {}] 
	\node[state, initial above] (q0) {$\qinit$};
	\node[state] (q1) [ above = 1 of q0, xshift = -1.5cm] {$q_1$};
	\node[state] (q2) [ below = 1 of q0, xshift = -1.5cm] {$q_2$};
	\node[state] (q3) [left = 3  of q0] {$q_3$};
	\node[state] (q5) [ right = 2.5 of q0] {$p_2$};
	\node[state] (q4) [ above  = 1 of q5] {$p_1$};
	\node[state] (q6) [ below = 1 of q5] {$p_3$};
	\node[state] (q7) [ right = 5 of q0] {$p_4$};
	
	%	\node[state] (q6) [ left = of q5] {$q_3$};
	%	\node[state] (q6) [ right = of q5] {$q_6$}; 
	%	\node[state] (q7) [ right = of q6] {$q_7$}; 	
	%	\node[state] (q8) [ right =of q5, xshift = 1cm, yshift = -1cm] {$q_8$};
	%	\node[state] (q9) [ right = of q8 ] {$q_6$};
	%	\node[state] (q10) [  right= of q8, yshift = -30] {$q_6$};
	%	
	%	\node[state] (q2) [right = of q1] {$q_2$};
	
	\path[->] 
	(q0) edge [bend right = 15] node [above, xshift = 5] {$!a$} (q1)
	edge [bend left = 15] node {$!b$} (q2)
	edge [bend left = 15]  node {$!m_1$} (q4)
	edge node {$!m_2$} (q5)
	edge [bend right = 15]  node [below] {$!m_3$} (q6)
	(q1) edge [bend right = 15] node [above, xshift = - 5] {$?a$} (q3)
	(q2) edge [bend left = 15] node {$?a, ?b$} (q3)
	
	(q4) edge [bend left = 15] node [xshift= -5] {$?m_1, ?m_3$} (q7)
	(q5) edge node {$?m_2, ?m_3$} (q7)
	(q6) edge [bend right = 15] node [below ,xshift = 12] {$?m_1, ?m_2, ?m_3$} (q7)
	%	edge   node [above] {$?b$} (q4)
	%	edge   node  {$!b$} (q6)
	%	(q6) edge [bend right] node {$?c$} (q2)
	%	%					 edge  [bend right] node [below left] {$?d$} (q6)
	;

\end{tikzpicture}
	}
	\caption{Wait-only protocol $\PP_2$.}\label{fig-example-wo-2}
\end{minipage}
\vspace*{-0.3cm}
\end{figure}

\begin{example}
%We provide now some intuition on how we defined $F$ using different
%examples. 
Consider the wait-only protocol $\PP_1$ depicted on Figure
\ref{fig-example-wo}. From $(\set{q_{in}},\emptyset)$, rules described in \cref{table:F}~construct the following pair $(S_1'', \Toks_1'') = (\set{q_{in},q_4},\set{(q_1,a),\linebreak[0](q_1,b),(q_5,c)})$.
%We have
%$F((\set{q_{in}},\emptyset))=(\set{q_{in},q_4},\set{(q_1,a),\linebreak[0](q_1,b),(q_5,c)})$. 
In
$\PP_1$, it is indeed possible to reach a configuration with as
many processes as one wishes in the state $q_4$
by repeating the transition $(q_{in},!d,q_4)$ (rule \ref{ccover-wo-F-cond-send-S}). On the other hand, it
is possible to put \emph{at most} one process in the waiting state $q_1$
(rule \ref{ccover-wo-F-cond-newtok}), because any other attempt from a process in $\qinit$ will yield a reception
of the message $a$ (resp. $b$) by the process already in $q_1$. % For
%instance, to add one token to $q_1$, one needs to use transitions
%$(q_{in},!a,q_1)$ or $(q_{in},!b,q_1)$, by repeating one of the transitions, a process will be added to  $q_1$ but
%another one will be removed from $q_1$ because of the transitions
%$(q_1,?a,q_2)$ and $(q_1,?b,q_2)$. 
Similarly, we can put at most
one process in $q_5$. Note that in
$\Toks_1''$, the states $q_1$
and $q_5$ are conflict-free and it is hence possible to have
simultaneously one process in both of them.

If we apply rules of \cref{table:F} one more time to $(S''_1, \Toks''_1)$, we get $S_2''=\set{\qinit, \textcolor{blue}{q_2}, {q_4}, \textcolor{blue}{q_6},\textcolor{blue}{q_7}}$ and 
$\Toks_2''=\set{{(q_1,a)}, {(q_1,b)} ,\textcolor{blue}{(q_3,a)},\textcolor{blue}{(q_3,b)},{(q_5,c)}}$.
We can put at most one process in $q_3$: to add one, a process will take the transition
$(q_1,?c,q_3)$. Since $(q_1,a)$, $(q_1,b)\in\Toks''_1$, there can be at most one process in
state $q_1$, and this process arrived by a path in which
the last request of rendez-vous was $!a$ or $!b$. % (this is witnessed by
%the tokens  $(q_1,a)$ and $(q_1,b)$), since these two rendez-vous can
%be accepted from state $q_3$ we cannot put a great number of processes
%in it. 
Since $\{a,b\}\subseteq\Rec{q_3}$, by rule~\ref{ccover-wo-F-cond-tok-step}, $(q_3,a),(q_3,b)$ are added. On the other hand we can put as many processes as we want in the
state $q_7$ (rule \ref{ccover-wo-F-cond-tok-end}): from a configuration with one process on state $q_5$, successive non-blocking request on letter $c$, and 
rendez-vous
on letter $d$ will allow to increase the number of processes in state $q_7$.

However, one can observe that $q_5$ can in fact host an unbounded number of processes:
 once two processes have been put on states $q_1$ and $q_5$ respectively (remember that $q_1$ and $q_5$ are conflict-free in $(S''_1, \Toks''_1)$), iterating rendez-vous on letter $c$ (with transition $(q_1, ?c, q_3)$) and rendez-vous on letter $a$ put as many processes as one wants on state $q_5$.

This is why we need another transformation from $S_2'', \Toks_2''$ to $F(S''_1, \Toks''_1)$.
As we shall see, this transformation does not have any impact on $S''_1$ and $\Toks''_1$ and so it holds that $F((\set{\qinit}, \emptyset)) = (S''_1, \Toks''_1)$.
\end{example}

Note $F(\gamma) = (S', \Toks')$, \cref{table2:F} describes the construction of $S'$ from $(S'', \Toks'')$, while $\Toks' = \Toks'' \setminus (S \times \Sigma)$, i.e.\ all states added to $S'$ are removed from $\Toks'$ so a state belongs either to $S'$ or to $\starg{\Toks'}$.

\begin{table}[]
	\begin{center}
		
		\label{tab2:S''}
		\makebox[\textwidth]{%
			\scalebox{1}{
				\begin{tabular}{ p{13.5cm}}
					\toprule
					\textbf{Construction of state $S'$, the smallest set including $S''$ and such
						that:
					}\\
					\midrule
					\vspace*{-0.5cm}
					\begin{enumerate}[itemsep=-0cm,itemindent=-0.2cm]\addtocounter{enumi}{5}
						\item for all $(q_1, m_1), (q_2, m_2) \in \Toks''$ such that $m_1
						\ne m_2$ and $m_2 \notin \Rec{q_1}$ and  $m_1 \in \Rec{q_2}$, we
						have $q_1 \in S'$;\label{ccover-wo-F-cond-2toks-1}
						\item for all $(q_1, m_1), (q_2, m_2), (q_3,m_2) \in \Toks''$ s.t $m_1 \ne m_2$
						and $(q_2, ?m_1, q_3) \in T$, we have $q_1 \in S'$;\label{ccover-wo-F-cond-3toks-1}
						\item for all $(q_1, m_1), (q_2, m_2), (q_3, m_3) \in \Toks''$ such
						that $m_1 \ne m_2$ and  $m_1\ne m_3$ and $m_2 \ne m_3$ and 
						$m_1 \notin \Rec{q_2}$,
						$m_1 \in \Rec{q_3}$ and  $m_2\notin \Rec{q_1}$, $m_2 \in \Rec{q_3}$, and $m_3 \in \Rec{q_2}$ and $m_3 \in \Rec{q_1}$,
						we have $q_1 \in S'$.\label{ccover-wo-F-cond-3toks-2}
						\vspace*{-0.4cm}
					\end{enumerate}
					%  6. for all $(q_1, m_1), (q_2, m_2) \in \Toks''$ such that $m_1
					%    \ne m_2$ and $m_2 \nin \Rec{q_1}$ and  $m_1 \in \Rec{q_2}$, we
					%    have $q_1 \in S'$;\\
					%    7. for all $(q_1, m_1), (q_2, m_2), (q_3,m_2) \in \Toks''$ s.t $m_1 \ne m_2$
					%   and $(q_2, ?m_1, q_3) \in T$, we have $q_1 \in S'$;\\
					% 8. for all $(q_1, m_1), (q_2, m_2), (q_3, m_3) \in \Toks''$ such
					%  that $m_1 \ne m_2$ and  $m_1\ne m_3$ and $m_2 \ne m_3$ and 
					%  $m_1 \nin \Rec{q_2}$,\\
					%   $m_1 \in \Rec{q_3}$ and  $m_2\nin \Rec{q_1}$, $m_2 \in \Rec{q_3}$, and $m_3 \in \Rec{q_2}$ and $m_3 \in \Rec{q_1}$,
					%   we have $q_1 \in S'$.\\
					%    \midrule
					%     \textbf{Construction of state $\Toks'$}\\
					%     \midrule
					%    $\Toks'=\set{(q,m) \in \Toks'' \mid q \not\in
						%  S'}$.
					\\
					%\midrule
					%\textbf{Construction of state $\Toks'$}\\
					%\midrule
					%$\Toks'=\set{(q,m) \in \Toks'' \mid q \not\in
						%	S'}$.
					%\\
					\toprule
					%\end{enumerate}
				\end{tabular}
				%}}
	}
}
\caption{{Definition of $F(\gamma)=(S',\Toks')$ for $(S'', \Toks'')$.}}\label{table2:F}
\end{center}
\vspace*{-1cm}
\end{table}

%Now, observe that the tokens $(q_5,c)$, $(q_1,a)$, $(q_3,a)$ allow for application of rule~\ref{ccover-wo-F-cond-3toks-1}, since $(q_1,?c,q_3)\in T$, and yields $q_5$ in $S'$. Once two processes have been put on states $q_1$ and $q_5$ respectively (remember that $q_1$ and $q_5$ are conflict-free in $F(\gamma)$), iterating rendez-vous on letter $c$ (with transition $(q_1, ?c, q_3)$) and rendez-vous on letter $a$ put as many processes as one wants on state $q_5$.
%Finally, $F({F(\set{q_{in}},\emptyset)})=(\set{{q_{in}}, q_2,{q_4}, q_5, q_6,q_7},\linebreak[0]\set{{(q_1,a)}, {(q_1,b)} ,(q_3,a),(q_3,b)})$. Since $q_1$ and $q_3$ are not
%conflict-free, they won't be reachable together in a configuration. % from a configuration with a
%great number of processes on state $\qinit$ and one process on state
%$q_5$, then, one can put a great number of processes on state $q_7$ by
%doing successively rendez-vous with letter $d$ and a non-blocking
%request on letter $c$.
%\end{example}

%Observe that it might be that a state is both added to $S''$ and $\Toks''$; in that case, it will be removed from $\Toks'$ by application of the last rule of $F$. Hence, 
%a state belongs either to $S'$ or to $\starg{\Toks'}$. 

\color{black}
\begin{example}

Now the case of state $q_5$ evoked in the previous example leads to application of  rule~\ref{ccover-wo-F-cond-3toks-1}, since $(q_5,c)$, $(q_1,a) \in \Toks''_2$, and $(q_3,a)$ $(q_1,?c,q_3)\in T$.
Finally, $F({F(\set{q_{in}},\emptyset)})=(\set{{q_{in}}, q_2,{q_4}, q_5, q_6,q_7},\linebreak[0]\set{{(q_1,a)}, {(q_1,b)} ,(q_3,a),(q_3,b)})$. Since $q_1$ and $q_3$ are not
conflict-free, they won't be reachable together in a configuration. 
\color{black}

We consider now the wait-only protocol $\PP_2$ depicted on Figure
\ref{fig-example-wo-2}. In that case, to compute
$F((\set{q_{in}},\emptyset))$ we will first have $S''=\set{q_{in}}$
and $\Toks''=\set{(q_1,a),(q_2,b),(p_1,m_1),(p_2,m_2),\linebreak[0](p_3,m_3)}$
(using rule \ref{ccover-wo-F-cond-newtok}), to finally get $F((\set{q_{in}},\emptyset))=(\set{q_{in},q_1,p_1},\set{(q_2,b),(p_2,m_2),\linebreak[0](p_3,m_3)}))$.
Applying rule \ref{ccover-wo-F-cond-2toks-1}~to tokens $(q_1,
a)$ and $(q_2, b)$ from $\Toks''$, we obtain that $q_1\in S'$: whenever one manages
to obtain one process in state $q_2$, this process can answer the requests on message $a$ instead of processes in state $q_1$, allowing one
to obtain as many processes as desired in state $q_1$.  
% appear on the tokens set, then applying $F$ should
%add $q_1$ to the set of unbounded states. 
%Indeed, take the reachable
%configuration with one process on state $q_1$, many processes on state
%$\qinit$ and no process on state $q_2$. With successive non-blocking
%requests on letter $b$, and rendez-vous on letter $a$ with transitions
%$(\qinit, !a, q_1)$ and $(q_2, ?a, q_3)$, we can reach a configuration
%with many processes on state $q_1$.
%
Now since $(p_1,m_1)$, $(p_2, m_2)$ and $(p_3, m_3)$ are in $\Toks''$
and respect the conditions of rule \ref{ccover-wo-F-cond-3toks-2}, $p_1$ is added to the set $S'$ of unbounded states.
% (we have indeed
%$m_1 \nin \Rec{p_2}$, $m_1 \in \Rec{p_3}$ and  $m_2\nin \Rec{p_1}$, $m_2 \in \Rec{p_3}$, and $m_3 \in \Rec{p_2}$ and $m_3 \in \Rec{p_1}$).
This case is a generalisation of the previous one, with 3 processes. Once one process has
been put on state $p_2$ from $\qinit$, iterating the following actions: rendez-vous over $m_3$, rendez-vous over $m_1$, non-blocking request of $m_2$,
will ensure as many processes as one wants on state $p_1$.
%
%To justify this, take the reachable configuration with one process on
%state $p_1$, one on state $p_2$ and many processes on state
%$\qinit$. By doing successively: rendez-vous on letter $m_3$ with
%transitions $(\qinit, !m_3, p_3)$ and $(p_2, ?m_3, p_4)$, rendez-vous
%on letter $m_1$ with transitions $(\qinit, !m_1, p_1)$ and $(p_3,
%?m_1, p_4)$ and non-blocking request on message $m_2$ (note that $m_2
%\nin \Rec{p_1}$), we can reach a configuration with many processes on
%state $p_1$. 
Finally applying successively $F$, we get in this case
the abstract set $(\set{q_{in},q_1,q_3,p_1,p_2,p_3,p_4},\set{(q_2,b)})$.
\end{example}

 We show that $F$ satisfies the following properties.
\begin{lemma}\label{lem:function-F}
\begin{enumerate}
\item \label{lem:function-F-poly}$F(\gamma)$ is consistent and can be computed in polynomial time
  for all consistent $\gamma \in \Gamma$.
\item \label{lem:function-F-included}If $(S',\Toks')=F(S,\Toks)$ then $S \neq S'$ (and $S \subseteq S'$) \color{black}or
  $\Toks \subseteq \Toks'$.
  %\nas{pourquoi un $\subsetneq$ et un $\subseteq$?}
 \item \label{lem:function-F-a-step}For all consistent $\gamma \in \Gamma$, if $C \in
   \Interp{\gamma}$ and $C \arrowP{}C'$ then $C' \in
   \Interp{F(\gamma)}$.
 \item\label{lem:function-F-bck-step} For all consistent $\gamma \in \Gamma$, if $C' \in
 \Interp{F(\gamma)}$, then there exists $C'' \in \CC$ and $C \in
 \Interp{\gamma}$ such that $C'' \geq C'$ and $C \arrowP{}^\ast C''$.
\end{enumerate}
\end{lemma}

\iflong
Point 1. and 2, ensures us that if we apply successively the function
$F$ to $(\set{q_{in}},\emptyset)$ then the computation will reach a consistent
abstract set $\gamma_f$ such  that $\gamma_f=F(\gamma_f)$ and it will
take a polynomial time. Points 3. ensures that the computed abstraction
is complete whereas Point 4. guarantees its soundness.
\fi

\subsection{Polynomial Time Algorithm}

We now present our polynomial time algorithm to solve \CCover~for
wait-only protocols. %It consists in computing a finite sequence of
%abstract sets of configurations using the fonction $F$ iteratively. 
We
define the sequence $(\gamma_n)_{n \in \nat}$ as follows:
$\gamma_0=(\set{\qinit},\emptyset)$ and $\gamma_{i+1}=F(\gamma_i)$ for all
$i \in \nat$. First note that $\gamma_0$ is consistent and that $\Interp{\gamma_0}=\Init$ is the set of
initial configurations. Using Lemma \ref{lem:function-F}, we deduce
that $\gamma_i$ is consistent for all $i \in \nat$. Furthermore, each time we apply $F$ to an
abstract set of configurations $(S,\Toks)$ either $S$  or $\Toks$
increases, or $(S, \Toks)$ stabilises. Hence for all $n \geq |Q|^2*|\Sigma|$, we have
$\gamma_{n+1}=F(\gamma_n)=\gamma_n$. Let
$\gamma_f=\gamma_{|Q|^2*|\Sigma|}$. Using Lemma \ref{lem:function-F},
we get:

\begin{lemma}\label{lem:correct-abstraction}
Given $C \in \CC$, there exists $C_0 \in \Init$ and $C' \geq C$
such that $C_0 \arrowP{}^\ast C'$ if and only if there exists $C'' \in
\Interp{\gamma_f}$ such  that $C'' \geq C$.
\end{lemma}

We need  to iterate $|Q|^2*|\Sigma|$ times the function $F$ to
compute $\gamma_f$ and each computation of $F$ can be done in
polynomial time. Furthermore checking whether there exists $C'' \in
\Interp{\gamma_f}$ such  that $C'' \geq C$ for a configuration $C \in
\CC$ can be done in polynomial time by Lemma
\ref{lem:interp-cover-check}, hence using the previous lemma we obtain the desired result.

\begin{theorem}
 \CCover~and \Cover~restricted to wait-only protocols are in \Ptime.
\end{theorem}
%
%We get the immediate corollary.
%\begin{corollary}
%	\Cover~restricted to wait-only protocols is in \Ptime.
%\end{corollary}

%	\input{old-wo}

    %\section{About synchronization: undecidability and open questions}\nas{open questions : on en parle?}
    \section{Undecidability of \Target}\label{sec:target}

	%We will see in this section that the synchronization and
%termination problems are undecidable. To obtain this result,
%we rely in the both cases on a reduction from the coverability
%problem for Minsky machines.

It is known that \CMCover[CM] is undecidable in its full generality~\cite{Minsky67}. This result holds for a very restricted class of counter machines, namely Minsky machines (Minsky-CM for short), which are CM over 2 counters, $\cpt_1$ and $\cpt_2$.
 Actually, it is already
undecidable whether there is an execution $(\ellinit,\mathbf{0}_{\{\cpt_1,\cpt_2\}})\transCM^* (\ell_f, \mathbf{0}_{\{\cpt_1,\cpt_2\}})$. 
Reduction from this last problem gives the following result.
%
%For readability's sake, when it comes to Minsky Machines, we will often write a configuration $(\ell, v)$ as $(\ell, v_1, v_2)$ where $v_1 = v(\cpt_1)$ and $v_2 = v(\cpt_2)$ (and $\Counters = \{\cpt_1, \cpt_2\})$.
%We will consider \emph{deterministic} \MinskyMachine, i.e machines satisfying the following condition: for all $\ell \in \Loc$, there exists a transition $(\ell, \dec{\cpt}, \ell') \in \Delta$ iff there exists a transiton $(\ell, \testz{\cpt}, \ell'') \in \Delta$ for some $\cpt \in \Counters$ and $\ell', \ell''\in \Loc$.\lug{je ne sais pas c'est nécessaire à définit?}

%\subsection{\TTarget~problem}
%
%This subsection presents a proof for the undecidability of the \tTarget~problem for the non-blocking semantics in rendez-vous protocols.

\begin{theorem} \label{target:thm}
\Target~is undecidable, even for wait-only protocols. 
	%The \TTarget~problem for rendez-vous protocols is undecidable and it is already the case for wait-only protocols.
\end{theorem}
%We prove that one can simulate a Minsky machine with a rendez-vous protocol, such that, there exists an execution of the machine from an initial configuration to a final configuration iff there exists an initial execution of the protocol leading to a configuration for which all processes are in the final state.
%
%
%
%Let $M$ be a Minsky machine. We define a rendez-vous protocol \PP~such that there exists an execution from $(\ell_0, 0, 0)$ from $(\ell_f, 0, 0)$ in the Minsky machine if and only if there exists $c \in \Init(\PP)$, $c' \in \FinalA(\PP)$ such that $c \arrowPnb{}^* c'$. 
%
%
Fix $M = (\Loc, \ell_0, \{ \cpt_1, \cpt_2\}, \Delta )$ with $\ell_f \in \Loc$ the final state. W.l.o.g., we assume that there is no outgoing transition from state $\ell_f$ in the machine. 
The protocol \PP~is described in \cref{wait-only:fig:target,wait-only:fig:target:incr,wait-only:fig:target:decr}. 
\begin{figure}
		\begin{minipage}[c]{.55\columnwidth}
				\resizebox*{!}{3.7cm}{
				\begin{tikzpicture}[->, >=stealth',  shorten >=1pt,node distance=1.5cm,on grid,auto, initial text = {}] 
					\node[state, initial] (q0) {$\qinit$};
					
					\node[state] (w) [above = 0.6 of q0, xshift = 40] {$w$};
					\node[state] (wp) [ right = 2.5 of w] {$w'$};
					\node[state, accepting] (qf) [right = 2 of wp] {$\ell_f$};
%					\node[state] (b) [right =of qf] {$\frownie$};
					
					\node[state] (q1) [below = 0.8 of q0, xshift = 40] {$q_1$};
					\node[state] (q2) [right = 2 of q1] {$q_2$};
					\node[state] (l0) [below = 1.4 of qf] {$\ellinit$};
					
					\node[state] (0) [above = 2 of q0] {$0_i$};
					\node[state] (01) [ right =of 0] {$p_i$};
					\node[state] (1) [ right = 2.5of 01] {$1_i$};
					\node[state] (10) [ right = 2 of 1] {$p'_i$};
%					\node[state, accepting] (lf) [right = 2.5of 10] {$\ell_f$};
					
					\node[state] (b1) [right =3.5 of 1] {$\frownie$};
					\node[draw, fill = teal, opacity = 0.1, text opacity = 1, fit=(l0) (qf), text height=0.12 \columnwidth] (P) {};

					\path[->]
					(q0) edge node {$\tau$} (q1)
					edge node [xshift = 5]{\small $!\textrm{init}$} (w)
					edge node {$\tau$} (0)
					
					(q1) edge node {\small $?\textrm{init}$} (q2)
					(q2) edge node {\small $!\textrm{ackinit}$} (l0)
					
					(w) edge node {\small $?\textrm{ackinit}$} (wp)
					(wp) edge node {\small $!\textrm{w}$} (qf)
					(qf) edge [in = -60, out = -30]  node  [below] {\small $?\textrm{w}$} (b1)
					
					(0) edge node {\small $?\textrm{inc}_i$} (01)
					(01) edge node {$!\textrm{ackinc}_i$} (1)
					(1) edge node {$?\textrm{dec}_i$} (10)
					edge [in = 140, out = 40]  node  {$?\textrm{zero}_i$} (b1)
					(10) edge node  [yshift = 5]{$!\textrm{ackdec}_i$} (qf)
					;
					
				\end{tikzpicture}
			}

			\caption{The protocol $\PP$ -- The coloured zone \\contains transitions pictured in~\cref{wait-only:fig:target:incr,wait-only:fig:target:decr,wait-only:fig:target:test}}\label{wait-only:fig:target}
		\end{minipage}
		\begin{minipage}[c]{.41\columnwidth}
\resizebox*{3.7cm}{!}{
				\begin{tikzpicture}[->, >=stealth',  shorten >=1pt,node distance=1.9cm,on grid,auto, initial text = {}] 
					\node[state] (l) {$\ell$};
					\node[state] (q) [right =of l] {};
					\node[state] (l') [right = 2.2 of q] {$\ell'$};
					
					\path[->]
					(l) edge node {$!\textrm{inc}_i$} (q)
					(q) edge node {$?\textrm{ackinc}_i$} (l');
					
				\end{tikzpicture}
}
				\caption{Translation of $(\ell, \inc{\cpt_i}, \ell')$.}\label{wait-only:fig:target:incr}
				\resizebox*{3.7cm}{!}{
				\begin{tikzpicture}[->, >=stealth',  shorten >=1pt,node distance=1.9cm,on grid,auto, initial text = {}] 
					\node[state] (l2) [] {$\ell$};
					\node[state] (q2) [right =of l2] {};
					\node[state] (l2') [right = 2.2 of q2] {$\ell'$};

					\path[->]
					(l2) edge node {$!\textrm{dec}_i$} (q2)
					(q2) edge node {$?\textrm{ackdec}_i$} (l2')
					;

				\end{tikzpicture}
			}
				\caption{Translation of $(\ell, \dec{\cpt_i}, \ell')$.}\label{wait-only:fig:target:decr}
	\resizebox*{2cm}{!}{
		\begin{tikzpicture}[->, >=stealth',  shorten >=1pt,node distance=1.9cm,on grid,auto, initial text = {}] 
			\node[state] (l2) [] {$\ell$};
			\node[state] (l2') [right = of l2] {$\ell'$};

			\path[->]
			(l2) edge node {$!\textrm{zero}_i$} (l2');

		\end{tikzpicture}
	}
	\caption{Translation of $(\ell, \testz{\cpt_i}, \ell')$.}\label{wait-only:fig:target:test}
	
		\end{minipage}
		
\end{figure}
The states $\{0_i,p_i,1_i,p'_i\mid i=1,2\}$ will be visited by processes simulating values of counters, while the states in $\Loc$ will be visited by a process simulating
the different locations in the Minsky-CM. If at the end of the computation, the counters are equal to 0, it means that each counter has been incremented and 
decremented the same number of times, so that all processes simulating the counters end up in the state $\ell_f$. The first challenge is to appropriately check when 
a counter equals 0. This is achieved thanks to the non-blocking semantics: the process sends a message $!\textrm{zero}_i$ to check if the counter $i$ equals 0. If it is does not, the message will be received by a process that will end up in the deadlock state $\frownie$. The second challenge is to ensure that only
one process simulates the Minsky-CM in the states in $\Loc$. This is ensured by the states $\{w, w'\}$. Each time a process arrives in the $\ellinit$ state, another must arrive in the $w'$ state, as a witness that the simulation has begun. This witness must reach $\ell_f$ for the computation to be a testifier of a positive instance of \Target, but
it should be the first to do so, otherwise a process already in $\ell_f$ will receive the message ``$\textrm{w}$'' and reach the deadlock state $\frownie$. Thus, if two processes simulate the Minsky-CM, there will be two witnesses, and they won't be able to reach $\ell_f$ together.

	\section{Conclusion}
	We have introduced the model of parameterised networks communicating by non-blocking rendez-vous, and showed that safety analysis of such networks becomes
much harder than in the framework of classical rendez-vous. Indeed, \CCover~and \Cover~become \Expspace-complete and \Target~undecidable in our 
framework, while these problems are solvable
in polynomial time in the framework of~\cite{german92}. We have introduced a natural restriction of protocols, in which control states
are partitioned between \emph{active} states (that allow requesting of rendez-vous) and \emph{waiting}  states (that can only answer to rendez-vous) and showed that \CCover~can then be solved in polynomial time. Future work includes finding further restrictions that would yield decidability of
\Target. A candidate would be protocols in which waiting states can only receive \emph{one} message. Observe that in that case, the reduction of~\cref{sec:target} can be adapted to simulate a \testfreeCM, hence \Target~for this subclass of protocols is as hard as reachability in Vector Addition Systems with States, i.e. non-primitive recursive \cite{leroux-reachability-focs21}. Decidability remains open though. 
\bibliography{main}

\newpage
\appendix

\section{Proofs of \cref{sec:cover-nb-machines}}

We present here the omitted proofs of \cref{sec:cover-nb-machines}.
\subsection{ Proof of~\cref{thm:cover-nbcm-in-expspace}}
We will in fact prove the \Expspace~upper bound for a more general model: Non-Blocking Vector Addition Systems (\NBVAS).
A \NBVAS~is composed of a set of transitions over vectors of dimension $d$, sometimes called counters, and an initial vector of $d$ non-negative integers, like in VAS. However, in a \NBVAS, a transition is a pair of vectors: one is a vector of $d$ integers and is called the \emph{blocking} part of the transition and the other one is a vector of $d$ \emph{non-negative} integers and is called the \emph{non-blocking} part of the transition. %Formally it is defined as follows.
\begin{definition}
	Let $d \in \mathbb{N}$. A Non-blocking Vector Addition System (\NBVAS) of dimension $d$ is a tuple $(T, v_0)$ such that $T \subseteq \mathbb{Z}^d \times \mathbb{N}^d$ and $v_{init}\in \mathbb{N}^d$.
\end{definition}

Formally, for two vectors $v, v' \in \mathbb{N}^d$, and a transition $t=(t_b, t_{nb}) \in T$, we write $v \xtransVas{t} v'$ if there exists $v'' \in \mathbb{N}^d$ such that $v'' = v + t_b$ and, for all $i \in [1,d]$, $v'(i) = \max(0, v''(i) - t_{nb}(i))$. We write $\transVas$ for $\bigcup_{t \in T} \xtransVas{t}$. We define an execution as a sequence of vectors $v_1 v_2 \dots v_k$ such that for all $1 \leq i < k$, $v_i \transVas v_{i+1}$.  

Intuitively, the blocking part $t_b$ of the transition has a strict semantics: to be taken, it needs to be applied to a vector large enough so no value goes below 0. The non-blocking part $t_{nb}$ can be taken even if it decreases some component below 0: the corresponding component will simply be set to 0. 

We can now define what is the \Cover~problem on \NBVAS.

\begin{definition}
\Cover~problem for a \NBVAS~$V = (T,v_{init})$ of dimension $d \in \mathbb{N}$~and a target vector $v_f$, asks if there exists $v\in \mathbb{N}^d$, such that $v \geq v_f$ and $v_{init} \transVas^\ast v$.
\end{definition}

%We are ready to adapt proof of \Expspace~membership of \Cover~problem in VAS ( \cite{rackoff78covering}) to our problem: \Cover~in \NBVAS.

%This proof	is an adaptation of the proof of [rackoff]\lug{do refs} for VAS (Vector Addition System) to our model. 
%To this end, we will use \NBVAS~(Non-Blocking Vector Addition System).

Adapting the proof of \cite{rackoff78covering} to the model of \NBVAS~yields the following result.

\begin{lemma}\label{lemma:cover:nbvas-expspace}
	The \Cover~problem for \NBVAS~is in \Expspace.
\end{lemma}

\begin{proof}
	Fix a \NBVAS~$(T,v_{init})$ of dimension $d$, we will extend the semantics of \NBVAS~to a slightly \emph{relaxed} semantics: let $v,v' \in \mathbb{N}^d$ and $t = (t_b, t_{nb}) \in T$, we will write $v \xtransVasZ{t} v'$ when for all $1\leq j \leq d$, $v'(j) = \max(0, (v+t_b -t_{nb})(j))$. 
	
	Note that $v \xtransVas{t} v'$ implies that $v \xtransVasZ{t} v'$ but the converse is false: consider an \NBVAS~ of dimension $d = 2$, with $t = (t_b, t_{nb}) \in T$
	such that $t_b =(-3, 0)$ and $t_{nb} = (0, 1)$, and let $v = (1, 2)$ and $v' =(0, 1)$. One can easily see that there does not exist $v'' \in \mathbb{N}^2$ such that $v'' = v + t_b$, as $1 - 3<0$. So, $t$ cannot be taken from $v$ and it is not the case that $v\xtransVas{t} v'$, however, $v \xtransVasZ{t} v'$. 
	
	We use $\transVasZ$ for $\bigcup_{t \in T} \xtransVasZ{t}$.
	
	Let $J \subseteq [1,d]$, a path $v_0\transVasZ v_1 \transVasZ \dots \transVasZ v_m$ is said to be \emph{$J$-correct} if for all $v_i$ such that $i < m$, there exists $t = (t_b, t_{nb}) \in T$, such that $v_i \xtransVasZ{t} v_{i+1}$ and for all $j \in J$, $(v_i + t_b)(j) \geq 0$. We say that the path is correct if the path is $[1,d]$-correct.
	
	It follows from the definitions that for all $v,v'\in\mathbb{N}^d$, $v\transVas^* v'$ if and only if there exists a correct path between $v$ and $v'$.

	Fix a target vector $v_f \in \mathbb{N}^d$, and define $\constantM= |v_f| + \max_{(t_b, t_{nb})\in T}(|t_b| + |t_{nb}|)$, where $|\cdot|$ is the norm 1 of vectors in $\mathbb{Z}^d$. Let $\rho = v_0 \transVasZ v_1 \transVasZ \dots \transVasZ v_m$ and $J \subseteq [1,d]$. We say the path $\rho$ is \emph{$J$-covering} if it is $J$-correct and for all $j \in J$, $v_m(j) \geq v_f(j)$. Let $r \in \mathbb{N}$, we say that $\rho$ is $(J,r)$-bounded 
	%\nas{n'a t il pas besoin d'être aussi $J$-correct?}\lug{pas comme utilisé dans la preuve}
	if for all $v_i$, for all $j \in J$, $v_i(j) < r$. Let $v \in \mathbb{N}^d$, we define $m(J,v)$ as the length of the shortest $J$-covering path starting with $v$, 0 if there is none.
	
	Note $\mathcal{J}_i = \{J\subseteq [1,d]\mid |J| = i \}$ and define the function $f$ as follows: for $1 \leq i \leq d$, $f(i) = \max \{m(J_i, v) \mid J_i \in \mathcal{J}_i, v\in \mathbb{N}^d\}$. We will see that $f$ is always well defined, in $\nat$. 
	
	\begin{claim}
		$f(0) = 1$.
	\end{claim}
	\begin{proof}
		From any vector $v \in \mathbb{N}^d$, the path with one element $v$ is $\emptyset$-covering.
	\end{proof}
	
	\begin{claim}
		For all $0 \leq i < d$, $f(i+1) \leq (\constantM\cdot f(i))^{i+1} + f(i)$.
	\end{claim}
	\begin{proof}
		Let $J \in \mathcal{J}_{i+1}$ and $v\in \mathbb{N}^d$ such that there exists a $J$-covering path starting with $v$. Note $\rho = v_0\xtransVasZ{t^1} \dots \xtransVasZ{t^m}v_m$ the shortest such path.\\
		
		\textbf{First case: $\rho$ is $(J, \constantM.f(i))$-bounded.} %We are going to bound the length of $\rho$ by $(Mf(i))^{i+1}$. In order to do this, we will prove that there is no two indices $k < \ell$ such that for all $j \in J$, $v_k(j) = v_\ell(j)$.
		Assume, for sake of contradiction, that for some $k < \ell$, for all $j\in J$, $v_k(j)=v_\ell(j)$. Then we show that $v_0\transVasZ\dots v_k\transVasZ \overline{v}_{\ell+1}\dots \transVasZ \overline{v}_m$ is also a $J$-correct path, with  
		the vectors $(\overline{v}_{\ell'})_{\ell< \ell'\leq m}$, defined as follows. 
		$$\overline{v}_{\ell+1}(j)=\begin{cases}v_{\ell+1}(j) & \textrm{for all $j\in J$}\\
			\max(0,(v_k(j)+t^{\ell+1}_b(j)-t^{\ell+1}_{nb}(j))) & \textrm{otherwise.}\end{cases}$$ 
		And for all $\ell + 1< \ell'\leq m$, 
		$$\overline{v}_{\ell'}(j)=\begin{cases}v_{\ell'}(j) &  \textrm{ for all $j\in J$}\\
			\max(0, (\overline{v}_{\ell'-1}(j)+t_b^{\ell'}(j)-t_{nb}^{\ell'}(j))) & \textrm{ otherwise.}\end{cases}$$ 
		
		Then $v_0\transVasZ\dots v_k\transVasZ \overline{v}_{\ell+1}\dots \transVasZ \overline{v}_m$ is also a $J$-correct path. Indeed, since $v_k(j)=v_\ell(j)$
		for all $j\in J$, we have that $\overline{v}_{\ell+1}(j)=v_{\ell+1}(j)=\max(0,(v_\ell(j) + t^{\ell+1}_b(j) - t^{\ell+1}_{nb}(j)))=\max(0,(v_k(j) + t^{\ell+1}_b(j) - t^{\ell+1}_{nb}(j)))$. 
		Moreover, for $j\in J$, since $v_{\ell}(j)+t^{\ell+1}_b(j)\geq 0$, we get that $v_k(j)+ t^{\ell+1}_b(j)\geq 0$. 
		By definition, for $j\notin J$, $\overline{v}_{\ell+1}(j)=\max(0,(v_k(j) + t^{\ell+1}_b(j) - t^{\ell+1}_{nb}(j)))$. Hence, $v_k\transVasZ^{t^{\ell+1}} \overline{v}_{\ell+1}$, and 
		$v_0\transVasZ^{t^1} \dots v_k\transVasZ^{t^{\ell+1}} \overline{v}_{\ell+1}$ is $J$-correct. Now let $\ell<\ell'<m$.
		By definition, for $j\in J$, $\overline{v}_{\ell'+1}(j)=v_{\ell'+1}(j)$. Then, $\overline{v}_{\ell'+1}(j)=\max(0,(v_{\ell'}(j)+t^{\ell'+1}_b(j) - t^{\ell'+1}_{nb}(j))) = \max(0,(\overline{v}_{\ell'}(j)+t^{\ell'+1}_b(j) - t^{\ell'+1}_{nb}(j)))$. Again, since $\rho$ is $J$-correct, we deduce that for $j\in J$, $v_{\ell'}(j)+t^{\ell'+1}_b(j)\geq 0$, hence $\overline{v}_{\ell'}(j)+t^{\ell'+1}_b(j)\geq 0$. For $j\notin J$, $\overline{v}_{\ell'+1}(j)=\max(0, (\overline{v}_{\ell'}(j)+t_b^{\ell'+1}(j)-t_{nb}^{\ell'+1}(j)))$. So $\overline{v}_{\ell'}\transVasZ^{t^{\ell'+1}} \overline{v}_{\ell'+1}$, and $v_0\transVasZ^{t^1} \dots v_k\transVasZ^{t^{\ell'+1}} \overline{v}_{\ell'+1}$ is $J$-correct.
		
		Then, $\rho'=v_0\transVasZ\dots v_k\transVasZ \overline{v}_{\ell+1}\dots \transVasZ \overline{v}_m$ is a $J$-correct path, and since $\overline{v}_m(j)=v_m(j)$ for all 
		$j\in J$, it is also $J$-covering, contradicting the fact that $\rho$ is minimal.

		Hence, for all $k < \ell$, there exists $j \in J$ such that $v_k(j) \not = v_\ell(j)$. The length of such a path is at most  $(\constantM.f(i))^{i+1}$, so $m(J,v)\leq (\constantM.f(i))^{i+1}\leq (\constantM.f(i))^{i+1}+f(i)$.\\
		
		\textbf{Second case: $\rho$ is not $(J, \constantM.f(i))$-bounded.} We can then split $\rho$ into two paths $\rho_1 \rho_2$ such that $\rho_1$ is $(J,\constantM.f(i))$-bounded and $\rho_2 = v'_0 \dots v'_n$ is such that $v'_0(j) \geq \constantM.f(i)$ for some $j \in J$.
		As we have just seen, $|\rho_1|\leq (\constantM.f(i))^{i+1}$.
		
		Note $J' = J \setminus \{j\}$ with $j$ such that $v'_0(j) \geq \constantM.f(i)$. Note that $\rho_2$ is $J'$-covering, therefore, by definition of $f$, there exists 
		a $J'$-covering execution $\overline{\rho} = w_0 \dots w_k$ with $w_0=v'_0$, and such that $|\overline{\rho}|\leq f(i)$. Also, by definition of $\constantM$, for all $1\leq j' \leq d$, for all $(t_b,t_{nb})\in T$,  $\constantM \geq |t_b(j')|+|t_{nb}(j')|$, then $t_b(j')\geq -\constantM$, and $t_b(j')-t_{nb}(j')\geq -\constantM$. Hence, for all $v\in\nat^d$,
		$1\leq j'\leq d$, and $c\in\nat$ such that $v(j')\geq \constantM + c$, for all $(t_b,t_{nb})\in T$, $(v+t_b)(j') \geq c$ and $(v+t_b-t_{nb})(j') \geq c$.
		Now, since $w_0 = v'_0$, we get $w_0(j)\geq \constantM.f(i)$.
		We deduce two things: first, for all $0 \leq \ell < k$, if $t=(t_b,t_{nb})\in T$ is such that $w_\ell \transVasZ^t w_{\ell+1}$, 
		it holds that $(w_{\ell} + t_b)(j)\geq \constantM.(f(i)- \ell - 1)$. Since $k = f(i) - 1$, it yields that $\overline{\rho}$ is $J$-correct.			
		Second, for all $0 \leq \ell \leq k$, $w_\ell(j)\geq \constantM(f(i) - \ell)$. Again, $k = f(i) - 1$, so $w_k(j) \geq\constantM\geq v_f(j)$. Hence 
		$\overline{\rho}$ is also $J$-covering. 
		
		Since $\rho$ is the shortest $J$-covering path, we conclude that $|\rho|\leq (\constantM.f(i))^{i+1} + f(i)$, and so $m(J,v)\leq (\constantM.f(i))^{i+1} + f(i)$.
	\end{proof}
	We define a function $g$ such that $g(0) = 1$ and $g(i+1) = (\constantM+1)^d(g(i))^d$ for $0 \leq i < d$; then $f(i)\leq g(i)$ for all $1 \leq i \leq d$. Hence, $f(d) \leq g(d) \leq (\constantM+1)^{d^{d+1}} \leq 2^{2^{cn\log n}}$ for some $n \geq \max( d,  \constantM, |v_{init}|)$ and a constant $c$ which does not depend on $d$, $v_0$, nor $v_f$ or the \NBVAS. Hence, we can cover vector $v_f$ from $v_{init}$ if and only if there exists a path (from $v_{init}$) of length $\leq 2^{2^{cn \log n}}$ which covers $v_f$.  Hence, there is a non-deterministic procedure that guesses a path of length $\leq 2^{2^{cn \log n}}$, checks if it is a valid path and accepts it if and only if it covers $v_f$. As $|v_{init}|\leq n$, $|v_f| \leq n$ and for all $(t_b, t_{nb}) \in T$, $|t_b| + |t_{nb}| \leq n$, this procedure takes an exponential space in the size of the protocol. By Savitch theorem, there exists a deterministic procedure in exponential space for the same problem.
\end{proof}

We are now ready to prove that the \Cover~problem for \NBVAS~is as hard as the \Cover~problem for \NBCM. 

%To this end, given a \NBCM, we will build a \NBVAS~for which one is a positive instance of \Cover~iff the other one is. The reduction goes as follows: we build a \NBVAS~of dimension $|\Counters|+|\Loc|$ where $\Counters$ is the set of counters of the machine and $\Loc$ the set of locations, this way, we can associate to each component a counter or a location of the machine. The initial vector should be the vector with the only non-null component is the one corresponding to $\ellinit$ and is equal to 1. Each transition of the machine is translated in one transition of the \NBVAS~in the following way: if the transition goes from $\ell$ to $\ell'$, the blocking part of the \NBVAS~transition has a -1 on the component corresponding to $\ell$ and a +1 on the one corresponding to $\ell'$. If the operation on counters is an increment (resp. a decremnt) on a counter $\cpt$, the blocking part also has a +1 (resp. -1) on the corresponding component. If it is a non-blocking decrement, then the non-blocking part has a -1 on the corresponding counter.\nas{je pense que ce paragraphe est inutile. Tout le monde sait ce que c'est qu'une réduction. L'intuition de la réduction peut être donnée dans la preuve. J'ajoute d'ailleurs des intuitions.}

\begin{lemma}\label{lemma:cover:nbvas-nbcm}
	\CMCover[\NBCM] reduces to \Cover~in \NBVAS.
\end{lemma}
\begin{proof}
	
Let a \NBCM~$M = (\Loc, \Counters, \Delta_b,\Delta_{nb}, \ellinit)$, for which we assume wlog that it does not contain any self-loop 
(replace a self loop on a location by a cycle
using an additional internal transition and an additional location).
We note $\Counters = \{\cpt_1, \dots, \cpt_m\}$, and $\Loc = \{\ell_1\dots \ell_{k}\}$, with $\ell_1=\ellinit$ and $\ell_k=\ell_f$, and let $d = k+m$.
We define the \NBVAS~$V = (T, v_{init})$ of dimension $d$ as follows: it has one counter by location of the \NBCM, and one counter by
counter of the \NBCM. The transitions will ensure that the sum of the values of the counters representing the locations of $M$ will always be equal to 1, hence
a vector during an execution of $V$ will always represent a configuration of $M$. First, for a transition $\delta = (\ell_i, \textit{op}, \ell_{i'})\in\Delta$, we define $(t_\delta, t'_\delta)\in \mathbb{Z}^d\times\mathbb{N}^d$ by $t_\delta(i) = -1,
t_\delta(i')= 1$ and,
\begin{itemize}
	\item if $\textit{op}=\nop$, then 
	$t_\delta(y)= 0 \textrm{ for all other $1\leq y\leq d$}$, and $t'_\delta=\mathbf{0}_d$ (where $\mathbf{0}_d$ is the null vector of dimension $d$), i.e. no other modification is made on the counters.	

	\item if $\textit{op}=\inc{\cpt_j}$, then $t_\delta(k+j)=1$, and 
	$t_\delta(y)= 0 \textrm{ for all other $1\leq y\leq d$}$, and $t'_\delta=\mathbf{0}_d$, i.e. the blocking part of the transition ensures the increment of the corresponding counter, while the non-blocking part does nothing.
	
	\item if $\textit{op}=\dec{\cpt_j}$, then $t_\delta(k+j)=-1$, and 
	$t_\delta(y)= 0 \textrm{ for all other $1\leq y\leq d$}$, and $t'_\delta=\mathbf{0}_d$, i.e. the blocking part of the transition ensures the decrement of the corresponding counter, while the non-blocking part does nothing.
.
	
	\item if $\textit{op}=\nbdec{\cpt_j}$, then  
	$t_\delta(y)= 0 \textrm{ for all other $1\leq y\leq d $}$, and $t'_\delta(k+j)=1$ and $t'_\delta(y)=0$ for all other $1\leq y\leq d$,
	%$t'_\delta(y)=\begin{cases}-1 & \textrm{ if $y=k+j$}\\
		%0 & \textrm{otherwise.}\end{cases}$, 
		i.e. the blocking part of the transition only ensures the change in the location, and the non-blocking decrement of the 
		counter is ensured by the non-blocking part of the transition.
\end{itemize}

We then let $T=\{t_\delta\mid \delta\in\Delta\}$,
%	\begin{align*}
	%	T & =  \{(t_{\delta},\mathbf{0}_d) \mid \delta = (\ell_i, \nop, \ell_{i'}) \in \Delta \wedge t_{\delta}(m + i) = -1 \wedge  t_{\delta}(m  +i') = 1  \wedge \forall k \not = m+i, m+i', t_{\delta}(k) = 0\}\\
	%	&  \cup \{(t_{\delta},\mathbf{0}_d) \mid \delta = (\ell_i, c_j++, \ell_{i'}) \in \Delta \wedge t_{\delta}(m + i) = -1 \wedge  t_{\delta}(m  +i') = 1 \wedge t_{\delta}(j) = 1 \wedge \forall k \not = m+i, m+i', j, t_{\delta}(k) = 0\}\\
	%	& \cup  \{(t_{\delta},\mathbf{0}_d) \mid \delta = (\ell_i, c_j--, \ell_{i'}) \in \Delta \wedge t_{\delta}(m+i) = -1 \wedge  t_{\delta}(m+i') = 1 \wedge t_{\delta}(j) = -1 \wedge \forall k \not = m+i, m+i', j, t_{\delta}(k) = 0\}\\
	%	& \cup  \{(t_{\delta}, t'_{\delta}) \mid \delta = (\ell_i, nb(\cpt_j--), \ell_{i'}) \in \Delta \wedge t_{\delta}(m+i) = -1 \wedge  t_{\delta}(m+i') = 1 \wedge t'_{\delta}(j) = -1\wedge \forall k \not = m+i, m+i' ,t_{\delta}(k) = 0\\& \ \ \  \wedge \forall k \not = j, t'_{\delta}(k) = 0\}
	%	\end{align*}
and $v_0$ is defined by $v_{init}(1)=1$ and $v_{init}(y)=0$ for all $2\leq y\leq d$. %= \begin{cases} 1 \textrm{ if $y=1$}\\ 0 &\textrm{otherwise.}\end{cases}$
We also fix $v_f$ by $v_f(k)=1$, and $v_f(y)=0$ for all other $1\leq y\leq d$.
One can prove that $v_f$ is covered in $V$ if and only if $\ell_f$ is covered in $M$.
\iflong
\begin{claim}
	If there exists $w \in \mathbb{N}^\Counters$ such that $(\ellinit, \mathbf{0}_\Counters) \transNbCM^* (\ell_f, w)$, then there exists $v \in \mathbb{N}^d$ such that $v_0 \transVas^* v$ and $v \succeq v_f$.
\end{claim}
\begin{sketch-proof}
Any configuration $(\ell,w)$ of $M$ can be turned into a valuation $v(\ell_i,w)$ of $T$ such that $v(\ell_i,w)(i)=1$, for all $1\leq i\leq m$, $v(\ell_i,w)(k+i)=w(\cpt_i)$
and for all other $1\leq y\leq k$, $v(\ell_i,w)(y)=0$. Observe that $v(\ellinit,\mathbf{0}_\Counters)=v_0$. It follows from the definitions that $(\ell_i,w)\transCM(\ell_{i'},w')$ if and only if $v(\ell_i,w)\transCM v(\ell_{i'},w')$. Hence, $v_0\transCM^*v(\ell_f,w)\geq v_f$. 
%	Starting from $v_0$, one can mimick the execution of the machine as follows: a configuration of the machine $ \gamma = (\ell_i, w)$ is translated as a vector $v_\gamma$ such that, $v_\gamma(i) = 1$, $v_\gamma(k+j) = w(\cpt_j)$ for all $1 \leq j \leq m$ and for all other counters $1\leq i'\leq k$, if $i'\neq i$, 
%	$v_\gamma(i') = 0$.
%	Note $(\ellinit, \mathbf{0}_\Counters) = \gamma_0 \transNbCM \gamma_1 \transNbCM \dots \transNbCM \gamma_n = (\ell_f, w)$. One can easily prove by induction that $v_{\gamma_0} \transVas v_{\gamma_1} \transVas \dots \transVas v_{\gamma_n}$, which concludes the proof.
\end{sketch-proof}

\begin{claim}
	If there exists $v \in \mathbb{N}^d$ such that $v_0 \transVas^* v$ and $v \succeq v_f$, then there exists $w \in \mathbb{N}^\Counters$ such that $(\ellinit, \mathbf{0}_\Counters) \transNbCM^* (\ell_f, w)$.
\end{claim}
\begin{sketch-proof}
	One can prove by induction that every vector $v$ reachable from $v_0$ is such that there exists only one $1 \leq i \leq k$ such that $v(i) = 1$ and for all $1 \leq i' \leq k$ such that $i \neq i'$, $v(i') = 0$. Hence, given a reachable vector $v$, one can define $\gamma_v$ a machine configuration as $(\ell_i, w)$ where $i$ is the unique index $1\leq i\leq k$ such that $v(i) = 1$ and, for all $1 \leq j \leq m$, $w(\cpt_j) = v(k+j)$. Note $v_0  \transVas v_1 \transVas \dots \transVas v_n = v$, and observe that $\gamma_{v_n} = (\ell_f, w)$ for some $w \in \mathbb{N}^\Counters$. Again, by a simple induction, one can prove that $\gamma_{v_0} \transNbCM \gamma_{v_1} \transNbCM \dots \transNbCM \gamma_{v_n}$, which concludes the proof.
\end{sketch-proof}	\fi
\end{proof}

Putting together Lemma \ref{lemma:cover:nbvas-expspace} and Lemma \ref{lemma:cover:nbvas-nbcm}, we obtain the proof of~\cref{thm:cover-nbcm-in-expspace}.

%\subsection{Proofs of \Expspace-hardness of \Cover[\NRCM]}
\subsection{Proof of \cref{th:expspace-hard}}
In this subsection, we prove \cref{th:expspace-hard}~by proving that the \Cover[\NRCM] problem is \Expspace~hard. Put together with \cref{thm:cover-nbcm-in-expspace}, it will prove the \Expspace-completeness of \Cover[\NRCM].

%We start by formalize the properties ensured by the procedural \NBCM~given in \cite{lipton76reachability,esparza98decidability}.
\subsubsection{Proofs on the Pocedural \NBCM~Defined in \cref{sec:cover-nb-machines}}
We formalize some properties on the procedural \NBCM~presented in \cref{sec:cover-nb-machines}~used in the proof.

%\paragraph*{Procedural \NBCM~$\mathtt{TestSwap}_i$.}
As for the procedural \NBCM~$\mathtt{TestSwap}_i$, we use this proposition from \cite{lipton76reachability,esparza98decidability}. 
\begin{proposition}[\cite{lipton76reachability,esparza98decidability}]\label{prop:test-swap}
Let $0\leq i < n$, and $\overline\cpt\in \overline Y_i$.
For all $v,v'\in\mathbb{N}^{X'}$, for $\ell\in\{\ell^{\mathtt{TS},i,\cpt}_{\textit{z}},\ell^{\mathtt{TS},i,\cpt}_{\textit{nz}}\}$, 
we have $(\ellinit^{\mathtt{TS},i},v)\transCM^*(\ell,v')$ in ${\mathtt{TestSwap}_i(\overline\cpt)}$
if and only if:
\begin{itemize}
\item (PreTest1): for all $0 \leq j < i$, for all $\overline\cpt_j \in \overline Y_j$, $v(\overline\cpt_j) = 2^{2^j}$ and for all $\cpt_j \in  Y_j$, $v(\cpt_j) = 0$;
\item (PreTest2): $v(\overline \cpts_i) = 2^{2^i}$ and $ v( \cpts_i) = 0$;
	\item (PreTest3): $v(\cpt) + v(\overline\cpt) = 2^{2^i}$;
\item (PostTest1): For all $\cpty\notin\{\cpt,\overline\cpt\}$, $v'(\cpty) = v(\cpty)$; 
	\item (PostTest2): either $(i)$ $v(\overline\cpt) = v'(\cpt) = 0$, $v(\cpt) = v'(\overline\cpt)$ and $\ell = \ell^i_z$, or $(ii)$ $v'(\overline \cpt) = v(\overline \cpt) >0$, $v'(\cpt) = v(\cpt)$ and $\ell = \ell^{\mathtt{TS},i,\cpt}_{nz}$.
\end{itemize}
Moreover, if for all $0 \leq j \leq n$, and any counter $\cpt \in Y_j \cup \overline Y_j$, $v(\cpt)\leq 2^{2^j}$, then for all $0 \leq j \leq n$, and any counter $\cpt\in Y_j \cup \overline Y_j$, the value of $\cpt$ will never go above $2^{2^j}$ during the execution. 
\end{proposition}

Note that for a valuation $v\in\mathbb{N}^{X'}$ that meets the requirements (PreTest1), (PreTest2) and (PreTest3), there is only one configuration
 $(\ell,v')$ with $\ell \in \{\ell^{\mathtt{TS},i,\cpt}_{\textit{z}},\ell^{\mathtt{TS},i,\cpt}_{\textit{nz}}\}$ such that $(\ell_{in},v) \transNbCM^* (\ell,v')$.

\paragraph*{Procedural \NBCM~$\mathtt{Rst}_i$.}

We shall now prove that the procedural \NBCM s~we defined and displayed in \cref{sec:cover-nb-machines}~meet the desired requirements.
For all $0\leq i\leq n$, any procedural \NBCM~$\mathtt{Rst}_i$ has the following property:
\begin{proposition}\label{lemma:cover:expspace-hard:nbcm:rst-spec}
	%Let  $0\leq i \leq n$ and $\ellinit^{\mathtt{R},i}$ (resp. $\ell^{\mathtt{R},i}_{out}$) be the entering (resp. exiting) state of the procedural-\NBCM~$\mathtt{Rst_i}$.
	 For all $0\leq i\leq n$, for all $v\in \mathbb{N}^{\Counters'}$ such that 
	 \begin{itemize}
		\item (PreRst1): for all $0 \leq j < i$, for all $\overline \cpt \in \overline Y_j$, $v(\overline \cpt) = 2^{2^j}$ and for all $\cpt \in Y_j$, $v(\cpt) = 0$,
	\end{itemize}
	for all $v' \in \mathbb{N}^{\Counters'}$, if $(\ellinit^{\mathtt{R},i}, v) \transNbCM^* (\ell^{\mathtt{R},i}_{out},v')$ in $\mathtt{Rst_i}$ then
	%then, there for all paths $(\ellinit^{\mathtt{R},i}, v) \transNbCM^+ (\ell^{\mathtt{R},i}_{out},v')$ in $\mathtt{Rst_i}$, $v'$ satisfies the following:
	\begin{itemize}
		\item (PostRst1): for all $\cpt \in Y_i \cup \overline Y_i$, $v'(\cpt) = \max(0, v(\cpt) - 2^{2^i})$,
		\item (PostRst2): for all $\cpt \not \in Y_i \cup \overline Y_i$, $v'(\cpt) = v(\cpt)$.
		\end{itemize}
 \end{proposition}

\begin{proof}[Proof of \cref{lemma:cover:expspace-hard:nbcm:rst-spec}]
	For $\mathtt{Rst_0}$, (PreRst1) trivially holds, and it is easy to see that (PostRst1) and (PostRst2) hold. Now fix $0 \leq i < n$, and consider the procedural-\NBCM~$\mathtt{Rst_{i+1}}$. Let $v_0 \in \mathbb{N}^{\Counters'}$ such that for all $0 \leq j < i+1$, for all $\overline \cpt \in \overline Y_j$, $v_0(\overline \cpt) = 2^{2^j}$ and for all $\cpt \in Y_j$, $v_0( \cpt) = 0$, and let $v_f$ such that $(\ellinit^{\mathtt{R},i}, v_0) \transNbCM^+ (\ell^{\mathtt{R},i}_{out},v_f)$ in $\mathtt{Rst}_i$.
	
	First, we show the following property.
	
	\noindent\emph{Property $(\ast)$}: if there exist $v,v'\in\mathbb{N}^{\Counters'}$ such that $v(\overline \cptz_i)=k$, %$v(\overline\mathtt{z}_i)=2^{2^i}-k$ and
	$(\ellinit^{\mathtt{TS},i,\cptz},v)\transCM^*(\ell^{\mathtt{TS},i,\cptz}_z,v')$ with no other visit of
	$\ell^{\mathtt{TS},i,\cptz}_z$ in between, then $v'(\overline \cptz_i)=2^{2^i}$, $v'({\cptz_i})=0$, for all $\cpt\in Y_{i+1}\cup\overline Y_{i+1}$, $v'(\cpt)=\max(0, v(\cpt)-k)$, and $v'(\cpt)=v(\cpt)$ for all other $\cpt\in \Counters'$.\\

	If $k=0$, then Proposition~\ref{prop:test-swap} ensures that $v'(\overline \cptz_i)=2^{2^i}$, $v'({\cptz_i})=0$, and for all other $\cpt\in \Counters'$, 
	$v'(\cpt)=v(\cpt)$. Otherwise, assume that the property holds for some $k\geq 0$ and consider $(\ellinit^{\mathtt{TS},i,\overline \cptz},v)\transCM^*(\ell^{\mathtt{TS},i,\overline \cptz}_z,v')$ with no other visit of
	$\ell^{\mathtt{TS},i,\cptz}_z$ in between, and $v(\overline \cptz_i)=k+1$. Here, since $v(\overline \cptz_i)=k+1$, Proposition~\ref{prop:test-swap} and the construction
	 of the procedural-\NBCM~ensure that $(\ellinit^{\mathtt{TS},i,\cptz},v)\transCM^*(\ell^{\mathtt{TS},i,\cptz}_{nz},v)\transCM(\ell^{\mathtt{R},i+1}_2,v)\transCM^*(\ellinit^{\mathtt{TS},i,\cptz},v_1)$ with $v_1(\overline \cptz_i)=k$, $v_1({\cptz}_i)=v({\cptz}_i)+1$, for all $\cpt\in Y_{i+1}\cup\overline Y_{i+1}$, 
	$v_1(\cpt)=\max(0, v(\cpt)-1)$, and for all other $\cpt\in \Counters'$, $v_1(\cpt)=v(\cpt)$. Induction hypothesis tells us that $(\ellinit^{\mathtt{TS},i,\cptz},v_1)\transCM^* (\ell^{\mathtt{TS},i,\cptz}_z,v')$ with  $v'(\overline \cptz_i)=2^{2^i}$, $v'({\cptz_i})=0$, for all $\cpt\in Y_{i+1}\cup\overline Y_{i+1}$, $v'(\cpt)=\max(0, v(\cpt)-k-1)$, and $v'(\cpt)=v(\cpt)$ for all other $\cpt\in \Counters'$. 
	
	Next, we show the following. 
	
	 \noindent\emph{Property $(\ast\ast)$}: if there exist $v,v'\in\mathbb{N}^{\Counters'}$ such that $v(\overline \cpty_i)=k$, $v(\overline \cptz_i)=2^{2^i}$, $v({\cptz_i})=0$,
	%$v(\overline\mathtt{z}_i)=2^{2^i}-k$ and
	and $(\ellinit^{\mathtt{TS},i,\cpty},v)\transCM^*(\ell^{\mathtt{TS},i,\cpty}_z,v')$ with no other visit of
	$\ell^{\mathtt{TS},i,\cpty}_z$ in between, then $v'(\overline \cpty_i)=2^{2^i}$, $v'(\overline{y_i})=0$, for all $\cpt\in Y_{i+1}\cup\overline Y_{i+1}$, $v'(\cpt)=\max(0, v(\cpt)- k.2^{2^i})$, and $v'(\cpt)=v(\cpt)$ for all other $\cpt\in \Counters'$. \\
	
	 If $k=0$, then Proposition~\ref{prop:test-swap} ensures that
	$v'(\overline \cpty_i)=2^{2^i}$, $v'({\cpty_i})=0$, and $v'(\cpt)=v(\cpt)$ for all other $\cpt\in \Counters'$. Otherwise, assume that the property holds for some $k\geq 0$ and consider $(\ellinit^{\mathtt{TS},i,\cpty},v)\transCM^*(\ell^{\mathtt{TS},i,\cpty}_z,v')$ with no other visit of
	$\ell^{\mathtt{TS},i,\cpty}_z$ in between, and $v(\overline \cpty_i)=k+1$. Again, since $v(\overline \cpty_i)=k+1$, Proposition~\ref{prop:test-swap} and the construction
	 of the procedural-\NBCM~ensure that $(\ellinit^{\mathtt{TS},i,\cpty},v)\transCM^*(\ell^{\mathtt{TS},i,\cpty}_{nz},v)\transCM(\ellinit^{\mathtt{R},i+1},v)\transCM^*(\ellinit^{\mathtt{TS},i,\cptz},v_1)\transCM^* (\ell^{\mathtt{TS},i,\cptz}_z,v'_1)\transCM (\ellinit^{\mathtt{TS},i,\cpty},v'_1)$, with $v_1(\overline \cpty_i)=v(\overline \cpty_i)-1=k$, $v_1({\cpty_i})=v({\cpty_i})+1$, $v_1(\overline \cptz_i)=v(\overline \cptz_i)-1=2^{2^i}-1$, $v_1({\cptz_i})=v({\cptz_i})+1=1$, for all $\cpt\in Y_{i+1}\cup\overline Y_{i+1}$, $v_1(\cpt)=\max(0,v(\cpt)-1)$, and for all other $\cpt\in\Counters'$, $v_1(\cpt)=v(\cpt)$. By Property ($\ast$), $v'_1(\overline \cptz_i)=2^{2^i}$, $v'_1({\cptz_i})=0$, for all $\cpt\in Y_{i+1}\cup\overline Y_{i+1}$, $v'_1(\cpt)=\max(0, v(\cpt)-2^{2^i})$, and $v'_1(\cpt)=v_1(\cpt)$ for all other $\cpt\in \Counters'$. Induction hypothesis allows to conclude that 
	 since $(\ellinit^{\mathtt{TS},i,\cpty},v'_1)\transCM^* (\ell^{\mathtt{TS},i,\cpty}_z,v')$, $v'(\overline \cpty_i)=2^{2^i}$, $v'({\cpty_i})=0$, for all $\cpt\in Y_{i+1}\cup\overline Y_{i+1}$, $v'(\cpt)=\max(0, v'_1(\cpt)- k.2^{2^i}) = \max(0, v(\cpt) - (k+1).2^{2^i})$, and $v'(\cpt)=v'_1(\cpt)=v(\cpt)$ for all other $\cpt\in \Counters'$.
	
	Since $(\ellinit^{\mathtt{R},i}, v_0) \transNbCM^+ (\ell^{\mathtt{R},i}_{out},v_f)$, we know that $(\ellinit^{\mathtt{R},i}, v_0) \transNbCM^* (\ellinit^{\mathtt{TS},i,\cptz},v)\transCM^*(\ell^{\mathtt{TS},i,\cptz}_z,v')\transCM(\ellinit^{\mathtt{TS},i,\cpty},v')\transCM^*(\ell^{\mathtt{TS},i,\cpty}_z,v'')\transCM (\ell^{\mathtt{R},i}_{out},v_f)$.
	By construction, $v(\overline \cpty_i)=2^{2^i}-1$, $v(\overline \cptz_i)=2^{2^i}-1$, $v({\cptz_i})=1$, $v({\cptz_i})=1$, for all $\cpt\in Y_{i+1}\cup\overline Y_{i+1}$,
	$v(\cpt)=\max(0,v_0(\cpt)-1)$, and for all other counter $\cpt$, $v(\cpt)=v_0(\cpt)$. By Property ($\ast$), $v'(\overline \cptz_i)=2^{2^i}=v_0(\overline\cptz_i)$, $v'({\cptz_i})=0=v_0({\cptz_i})$,
	for all $\cpt\in Y_i\cup \overline{Y_{i+1}}$, $v'(\cpt)=\max(0, v_0(\cpt)-2^{2^i})$ and for all other $\cpt\in\Counters'$, $v'(\cpt)=v(\cpt)$. By Property ($\ast\ast)$,
	$v''(\overline\cpty_i)=2^{2^i}=v_0(\overline\cpty_i)$, $v''({\cpty_i})=0=v_0({\cpty_i})$, for all $\cpt\in Y_i\cup\overline{Y_{i+1}}$, 
	$v''(\cpt)=\max(0, v_0(\cpt)-2^{2^i} - (2^{2^i}-1).2^{2^i})=\max(0, v_0(\cpt)-2^{2^i}.2^{2^i})=\max(0, v_0(\cpt)-2^{2^{i+1}})$, and for all other $\cpt\in\Counters'$, $v''(\cpt)=v'(\cpt)=v_0(\cpt)$.
%	First, we show that if there exist $v,v'\in\mathbb{N}^{\Counters'}$ such that $(\ellinit^{\mathtt{R},i},v)\transCM^*(\ell^{\mathtt{TS},i,z}_z,v')$ during the
%	execution, then $v'(\mathtt{z}_i)=2^{2^i}$, $v'(\overline{z_i})=0$, for all $\cpt\in Y_{i+1}\cup\overline Y_{i+1}$, $v'(\cpt)=\max(0, v(\cpt)-2^{2^i}+1)$, and $v'(\cpt)=v(\cpt)$ for all other $\cpt\in \Counters'$. First observe that $v(\mathtt{z}_i)=2^{2^i}-1$. If it is the first time that $\ell_7^{\mathtt{R},i+1}$ is visited, then $v(\mathtt{z}_i)=v_0(\mathtt{z}_i)-1= 2^{2^i}-1$. If it is not, 	
	\end{proof}
	
	We get the immediate corollary:
\begin{lemma}\label{lemma:cover:expspace-hard:nbcm:rst-spec-2}
Let $0\leq i\leq n$, and $v\in\mathbb{N}^{\Counters'}$ satisfying (PreRst1) for $\mathtt{Rst}_i$. If $v$ is $i$-bounded, then 
	the unique configuration such that $(\ell^{\mathtt{R},i}_{in},v) \transNbCM^+ (\ell^{\mathtt{R},i}_{out}, v')$ in $\mathtt{Rst}_i$ is defined $v'(\cpt) = 0$ for all 
	$\cpt \in Y_i \cup \overline Y_i$ and $v'(\cpt) = v(\cpt)$ for all $\cpt \notin Y_i \cup \overline Y_i$.
	
%	Moreover, for all $0 \leq j \leq n$, and any counter $\cpt_j \in Y_j \cup \overline Y_j$, its value never goes above $2^{2^j}$ during the execution. \lug{rajouté ici pour ne pas avoir une propriété juste pour ça}
\end{lemma}	

%%%%%%%%%%%%%%%%%%%%%%%%%%%%%%%%
\begin{proposition}\label{lem:rst-bounded}
	Let $0\leq i \leq n$, and let $v\in\mathbb{N}^{\Counters'}$ satisfying (PreRst1) %such that (PreRst1) holds 
	%and (PreRst2) holds 
	for $\mathtt{Rst_i}$. If for all $0\leq j \leq n$, $v$ is $j$-bounded, then  for all $(\ell,v')\in\Loc^{\mathtt{R},i}\times\mathbb{N}^{\Counters'}$ such that 
	$(\ell^{\mathtt{R},i}_{in},v) \transNbCM^* (\ell, v')$ in $\mathtt{Rst}_i$, $v'$ is $j$-bounded for all $0\leq j \leq n$.
\end{proposition}

\begin{proof}%[Proof of~\cref{lem:rst-spec-bounded}]
	
	We will prove the statement of the property along with some other properties: (1) if $\ell$ is not a state of $\mathtt{TestSwap}_i(\overline{\cptz}_i)$ or  $\mathtt{TestSwap}_i(\overline{\cpty}_i)$, then for all $0 \leq j < i$, for all $\cpt \in \overline{Y_j}$, $v'(\cpt) = 2^{2^j}$ and for all $\cpt \in Y_j$, $v'(\cpt) = 0$, and $v'(\overline{\cpts_i}) =2^{2^i}$ and $v'(\cpts_i) = 0$.
	(2) if $\ell$ is not a state of $\mathtt{TestSwap}_i(\overline{\cptz}_i)$ or  $\mathtt{TestSwap}_i(\overline{\cpty}_i)$ and if $\ell \ne \ell_1^{\mathtt{R}, i+1}$, then $v'(\cpty_i) + v'(\overline{\cpty}_i) = 2^{2^i}$, and if $\ell \ne \ell_3^{\mathtt{R}, i+1}$, then $v'(\cptz_i) + v'(\overline{\cptz}_i) = 2^{2^i}$.
	
	For $\mathtt{Rst_0}$, the property is trivial. Let $0\leq i <n$, and a valuation $v\in\mathbb{N}^{\Counters'}$
	such that for all $0 \leq j \leq i$, for all $\overline \cpt \in \overline Y_j$, $v(\overline \cpt) = 2^{2^j}$ and for all $\cpt \in Y_j$, $v(\cpt) = 0$, and such that, for all $0\leq j\leq n$, $v$ is $j$-bounded. Let now $(\ell,v')$ such that  $(\ell^{\mathtt{R},{i+1}}_{in},v) \transNbCM^* (\ell, v')$ in $\mathtt{Rst}_{i+1}$. 
	We prove the property by induction on the number of occurences of $\ellinit^{\mathtt{TS},i,z}$ and  $\ellinit^{\mathtt{TS},i,y}$. If there is no occurence of such state between in $(\ell^{\mathtt{R},{i+1}}_{in},v) \transNbCM^* (\ell, v')$, then, for all $\cpt \in Y_j \cup \overline{Y_j} \cup \{\cpts_i, \overline{\cpts_i}\} $ and $j \ne i$, $j \ne i+1$, then $v'(\cpt) = v(\cpt)$ and so $v'$ is $j$-bounded. Furthermore, for $ \cpt \in Y_i \cup Y_{i+1} \cup \overline{Y_{i+1}}$, $v'(\cpt) \leq v(\cpt)$,  and for all $\cpt \in \overline{Y_{i}}$, $v'(\cpt) \leq v(\cpt) + 1 = 1$. The property (2) is easily verified. 
	Hence the properties hold.
	
	Assume now we proved the properties for $k$ occurrences of $\ellinit^{\mathtt{TS},i,z}$ and  $\ellinit^{\mathtt{TS},i,y}$, and let us prove the clam for $k+1$ such occurrences. Note $\ell_{k+1} \in \{\ellinit^{\mathtt{TS},i,z},\ellinit^{\mathtt{TS},i,y} \}$ the last occurence such that: 
	$(\ell^{\mathtt{R},{i+1}}_{in},v) \transNbCM^+ (\ell_k, v_k) \transNbCM (\ell_{k+1}, v_{k+1}) \transNbCM^* (\ell, v')$. By induction hypothesis, $v_k$ is $j$-bounded for all $0 \leq j \leq n$ and it respects (1) and (2), and by construction, $(\ell_k, \nop, \ell_{k+1})$ and $\ell_k \ne \ell_1^{\mathtt{R},i+1}$, $\ell_k \ne \ell_3^{\mathtt{R}, i+1}$, hence $v_{k+1}$ is $j$-bounded for all $0 \leq j \leq n$ and respects (PreTest1), (PreTest2), and (PreTest3) for $\mathtt{TestSwap}_i(\overline{\cptz}_i)$ and $\mathtt{TestSwap}_i(\overline{\cpty}_i)$. As a consequence, if $\ell$ is a state of one of this machine such that $(\ell_{k+1}, v_{k+1})\transNbCM^* (\ell, v')$, then by \cref{prop:test-swap},  for all $0 \leq j \leq n$, as $v_{k+1}$ is $j$-bounded, so is $v'$.
	
	Assume now $\ell$ to not be a state of one of the two machines. And keep in mind that $v_{k+1}$ respects (1) and (2). Then, either $\ell = \ell_{out}^{\mathtt{R}, i+1}$ and so $v'(\cpt) = v_{k+1}(\cpt)$ for all $\cpt \in Y_j \cup \overline{Y}_j$ for all $j \ne i$, and $v'(\overline{\cpty_i}) = 2^{2^i}$ and $v'(\cpty_i) = 0$ and so the claim holds,  either $\ell \in \{\ell_{in}^{\mathtt{R,i+1}}, \ell_{j'}^{\mathtt{R}, i+1}\}_{j' = 1, 2, 3, 4, 5, 6, \dots, r}$. In this case, the execution is such that: $(\ell_{k+1}, v_{k+1}) \transNbCM^+ (\ell_{nz, k+1}, v_{k+1}) \transNbCM^* (\ell, v')$, where if $\ell_{k+1} =\ellinit^{\mathtt{TS},i,z}$, $\ell_{nz, k+1} = \ell^{\mathtt{TS}, i ,z}_{nz}$ and otherwise $\ell_{nz, k+1} = \ell^{\mathtt{TS}, i ,y}_{nz}$.
	In any cases, for all $j \ne i$, $j \ne i+1$, $\cpt \in Y_{j} \cup \bar Y_j \cup \{\cpts_i, \overline{\cpts_i}\}$, $v'(\cpt) = v_{k+1}(\cpt)$, hence (1) holds and $v'$ is $j$-bounded for all $j < i$ and $j > i+1$.
	
	Observe as well that for all $\cpt \in Y_{i+1} \cup \overline{Y}_{i+1}$, $v'(\cpt) \leq v_{k+1}(\cpt)$, and so $v'$ is $i+1$-bounded. The last thing to prove is that (2) holds. This is direct from the fact that $v_{k+1}$ respects (2).

\end{proof}

%%%%%%%%%%%

%\paragraph*{Procedural \NBCM~$\mathtt{Inc}_i$}
About the procedural \NBCM~$\mathtt{Inc}_i$, we use this proposition from \cite{lipton76reachability,esparza98decidability}.
\begin{proposition}[\cite{lipton76reachability,esparza98decidability}]\label{proposition:inc}
For all $0\leq i< n$, for all $v,v'\in\mathbb{N}^{\Counters'}$, $(\ellinit^{\mathtt{Inc}, i},v) \transNbCM^* (\ell_{out}^{\mathtt{Inc}, i}, v')$ in $\mathtt{Inc}_i$ if and only if:
\begin{itemize}
	\item (PreInc1) for all $0 \leq j < i$, for all $\cpt \in \overline Y_j$, $v(\cpt) = 2^{2^j}$ and for all $\cpt \in  Y_j$, $v(\cpt) = 0$;
	\item (PreInc2) for all $\cpt \in \overline Y_i$, $v( \cpt) = 0$,
	\item (PostInc1) for all $ \cpt \in \overline Y_i$, $v'(\cpt) = 2^{2^i}$;
	\item (PostInc2) for all $\cpt \not \in Y_i $, $v'(\cpt) = v(\cpt)$.
\end{itemize}	
Moreover, if for all $0\leq j \leq n$, $v$ is $j$-bounded, then  for all $(\ell,v'')$ such that 
	$(\ell^{\mathtt{Inc},i}_{in},v) \transNbCM^* (\ell, v'')$ in $\mathtt{Inc}_i$, then $v''$ is $j$-bounded for all $0\leq j\leq n$.
\end{proposition}

\paragraph*{Procedural \NBCM~$\mathtt{RstInc}$.}
We shall now prove the properties in the procedural \NBCM~$\mathtt{RstInc}$ defined in \cref{sec:cover-nb-machines}.
The next proposition establishes the correctness of the construction $\mathtt{RstInc}$.
\begin{proposition}\label{prop:cover:proof:expspace-hard:rstinc}
	Let $v \in \mathbb{N}^{\Counters'}$ be a valuation such that for all $0\leq i \leq n$ and for all $\cpt \in Y_i \cup \overline Y_i$, $v(\cpt) \leq 2^{2^i}$. Then the
	unique valuation $v' \in \mathbb{N}^{\Counters'}$ such that $(\ell_a, v) \transNbCM^* (\ell_b, v')$ in $\mathtt{RstInc}$ satisfies the following: for all $0\leq i \leq n$,
	for all $\cpt \in \overline Y_i $, $v'(\cpt) = 2^{2^i}$ and for all $\cpt \in Y_i$, $v'(\cpt) = 0$. 
	Moreover, for all $(\ell,v'')$ such that $(\ell_a, v) \transNbCM^* (\ell, v'')$ in $\mathtt{RstInc}$, for all $0\leq i\leq n$, $v''$ is $i$-bounded. 	\end{proposition}%\nas{c'est le lemme le plus important : il etablit la correction de RstInc, il faut le mettre + en valeur}

\begin{proof}[Proof of~\cref{prop:cover:proof:expspace-hard:rstinc}]
We can split the execution in $(\ell_a,v)\transCM (\ellinit^{\mathtt{R},0},v)\transCM^*(\ell^{\mathtt{R},0}_{out}, v_0)\transCM (\ellinit^{\mathtt{Inc},0},v_0)\transCM^*
(\ell_{out}^{\mathtt{Inc},0},v'_0)\transCM
(\ellinit^{\mathtt{R},1},v'_0)\transCM^*(\ell^{\mathtt{R},1}_{out},v_1)\transCM^*(\ellinit^{\mathtt{Inc},n-1}, v_{n-1})\transCM^*(\ell^{\mathtt{Inc},n-1}_{out}, v'_{n-1})\transCM 
(\ellinit^{\mathtt{R},n}, v'_{n-1})\transCM^*(\ell_{out}^{\mathtt{R},n},v_n)\transCM(\ell_b,v')$, with $v'=v_n$ and $v=v'_{-1}$. We show that for all $0\leq i\leq n$:
\begin{itemize}
\item $P_1(i)$: For all $\cpt\in Y_i\cup\overline Y_i$, $v_i(\cpt)=0$, and for all $\cpt\notin (Y_i\cup\overline Y_i)$, $v_i(\cpt)=v'_{i-1}(\cpt)$. 
\item $P_2(i)$: For all $0\leq j <i$, for all $\cpt\in Y_j$, $v'_{i-1}(\cpt)=0$ and for all $\cpt\in \overline Y_j$, $v'_{i-1}(\cpt)=2^{2^j}$, and
for all other $\cpt\in\Counters'$, $v'_i(\cpt)=v_i(\cpt)$. 
\item $P_3(i)$: For all $v''$ such that $(\ell_a, v) \transNbCM^* (\ell, v'')\transNbCM^* (\ell^{\mathtt{R},i}_{out}, v_{i})$, $v''$ is $i$-bounded, for all $0\leq i\leq n$.
\end{itemize}

For $k=0$, \cref{lemma:cover:expspace-hard:nbcm:rst-spec-2} implies that for all $\cpt\in Y_0\cup\overline Y_0$, $v_0(\cpt)=0$, and that for all other $\cpt\in\Counters'$,
$v_0(\cpt)=v(\cpt)$. Moreover, for all $v''$ such that $(\ellinit^{\mathtt{R},0},v)\transCM^*(\ell, v'')\transCM^*(\ell_{out}^{\mathtt{R},0},v_0)$, \cref{lem:rst-bounded} 
ensures that $v''$ is $i$-bounded, for all $0\leq i\leq n$. $P_2(0)$ is trivially true. 
%By \cref{proposition:inc}, we have that for all $v''$ such that $(\ell_{out}^{\mathtt{R},0},v_0)\transCM(\ellinit^{\mathtt{Inc},0},v_0)\transCM^*(\ell, v'')\transCM^*(\ell_{out}^{\mathtt{Inc},i},v'_0)$, $v''$ is $i$-bounded, for all $0\leq i\leq n$. 
%for all $0\leq j\leq n$, $v''$ is $j$-bounded. 

%Now by \cref{proposition:inc}, we deduce then that $v'_0(\cpt)=2^{2^0}=2$, for all $\cpt\in\overline Y_0$, and that $v'_0(\cpt)=v_0(\cpt)$ for all other $\cpt\in\Counters'$. So $v'_0(\cpt)=0$ for $\cpt\in Y_0$, and $v'_0(\cpt)=v(\cpt)$ for all other $\cpt\in \Counters'$.
%\cref{proposition:inc} also allows to conclude that for all $v''$ such that $(\ellinit^{\mathtt{Inc},0},v_0)\transCM^*(\ell, v'')\transCM^*(\ell_{out}^{\mathtt{Inc},i},v'_0)$,
%for all $0\leq j\leq n$, $v''$ is $j$-bounded. 

 Let $0\leq k< n$, and assume that $P_1(k)$, $P_2(k)$ and $P_3(k)$ hold. $P_1(k)$ and $P_2(k)$ and \cref{proposition:inc} imply that for all $\cpt\in \overline Y_k$, $v'_k(\cpt)= 2^{2^k}$, and that for all other counter $\cpt\in\Counters'$, $v'_k(\cpt)=v_k(\cpt)$. Thanks to $P_1(k)$, $P_2(k+1)$ holds. Moreover, we also know by~\cref{proposition:inc} that for all $v''$ such that $(\ell_{out}^{\mathtt{R},k},v_k)\transCM (\ellinit^{\mathtt{Inc},k}, v_k)\transCM^*(\ell, v'')\transCM^*(\ell_{out}^{\mathtt{Inc},k}, v'_k)$, $v''$ is $i$-bounded for all $0\leq i\leq n$. Since $v'_k$ is then $i$-bounded for all $0\leq i\leq n$, and since $P_2(k)$ holds, \cref{lemma:cover:expspace-hard:nbcm:rst-spec-2}
 implies that $v_{k+1}(\cpt)=0$ for all $\cpt\in Y_{k+1}\cup\overline Y_{k+1}$, and that, for all other $\cpt\in\Counters'$, $v_{k+1}(\cpt)=v'_k\cpt)$. So $P_1(k+1)$ holds.
 Moreover, by~\cref{lem:rst-bounded}, for all $v''$ such that $(\ell_{out}^{\mathtt{Inc},k}, v'_k)\transCM(\ellinit^{\mathtt{R},k+1},v'_k)\transCM^*(\ell,v'')\transCM^*
 (\ell_{out}^{\mathtt{R},k+1},v_{k+1})$, $v''$ is $i$-bounded for all $0\leq i\leq n$. Hence $P_3(k+1)$ holds. 
 
 By $P_1(n)$, $v'(\cpt)=0$ for all $\cpt\in Y_n$, and since $\overline Y_n=\emptyset$, $v'(\cpt)=2^{2^n}$ for all $\cpt\in \overline Y_n$. Let $\cpt\notin (Y_n\cup
 \overline Y_n)$. Then $v'(\cpt)=v'_{n-1}(\cpt)$, and by $P_2(n)$, for all $0\leq i <n$, for all $\cpt\in \overline Y_i$, $v'(\cpt)=2^{2^i}$, and for all $\cpt\in Y_i$,
 $v'(\cpt)=0$. By $P_3(n)$, for all $(\ell,v'')$ such that $(\ell_a, v) \transNbCM^* (\ell, v'')$ in $\mathtt{RstInc}$, for all $0\leq i\leq n$, $v''$ is $i$-bounded.
\end{proof}

\subsubsection{Proofs of the Reduction}
We are now ready to prove \cref{th:expspace-hard}, i.e.\ that the reduction is sound and complete.
For some subset of counters $Y$, we will note $v_{| Y}$ for the valuation $v$ on counters $Y$, formally, $v_{| Y} : Y \rightarrow \mathbb{N}$ and is equal to $v$ on its domain.
\begin{lemma}
	If there exists $v \in \mathbb{N}^\Counters$ such that $(\ellinit, \mathbf{0}_\Counters) \transCM^*_M (\ell_f, v)$, then there exists $v' \in \mathbb{N}^{\Counters'}$ such that $(\ellinit', \mathbf{0}_{\Counters'}) \transNbCM^*_N (\ell_f, v')$.
\end{lemma}
\begin{proof}
	From \cref{prop:cover:proof:expspace-hard:rstinc}, we have that $(\ellinit', \mathbf{0}_\Counters') \transNbCM^*_N (\ellinit, v_0)$ where $v_0$ is such that, for all $0 \leq j \leq n$, for all $\cpt \in \overline Y_j$, $v_0(\cpt) = 2^{2^j}$ and for all $ \cpt \in  Y_j$, $v_0( \cpt) = 0$. By construction of $N$, $(\ellinit, v_0)\transCM^*_N (\ell_f,v')$ with $v'$ defined by: for all $0\leq i <n$, for all $\cpt \in\overline Y_j$, $v'(\cpt) = 2^{2^j}$, for all $\cpt \in  Y_j$, $v'(\cpt) = 0$, and,  for all $\cpt \in \Counters$, $v'(\cpt) = v(\cpt)$. Note that in this path, there is no restore step.
\end{proof}

\begin{lemma}
	If there exists $v' \in \mathbb{N}^{\Counters'}$ such that $(\ellinit', \mathbf{0}_{\Counters'}) \transNbCM^*_N (\ell_f, v')$, then there exists $v \in \mathbb{N}^\Counters$ such that $(\ellinit, \mathbf{0}_\Counters) \transCM^*_M (\ell_f, v)$.
\end{lemma}
\begin{proof}
We will note $v_0$ the function such that for all $0\leq i \leq n$, and for all $\cpt \in \overline Y_i $, $v_0(\cpt) = 2^{2^i}$ and for all $ \cpt \in  Y_i$, $v_0(\cpt) = 0$.
Observe that there might be multiple visits of location $\ellinit$ in the execution of $N$, because of the restore transitions. The construction of $\mathtt{RstInc}$ 
ensures that, every time a configuration $(\ellinit,v)$ is visited, $v=v_0$.
	Formally, we show that for all $(\ellinit, v)$ such that $(\ellinit',\mathbf{0}_{\Counters'})\transCM^*_N(\ellinit,v)$, we have that $v=v_0$. First let $(\ellinit',w)\transCM^*_N(\ellinit',w')$, with $w(\cpt)\leq 2^{2^i}$, and $\ellinit'$, $\ellinit$ not visited in between. Then for all $0\leq i\leq n$, for all $\cpt\in Y_i\cup\overline Y_i$,  $w'(\cpt)\leq 2^{2^i}$. Indeed, let $(\ell,\overline{w})$ be such that $
	(\ellinit',w)\transCM^*_N(\ell, \overline{w})\transCM_N(\ellinit',w')$. By \cref{prop:cover:proof:expspace-hard:rstinc}, we know that, for 
	all $0\leq i\leq n$, for all $\cpt\in Y_i\cup\overline Y_i$, $\overline{w}(\cpt)\leq2^{2^i}$. Since the last transition is a restore transition, we deduce that,
	for all $0\leq i\leq n$, for all $\cpt\in Y_i\cup\overline Y_i$,  $w'(\cpt)=\overline{w}(\cpt)\leq 2^{2^i}$. 
	\begin{itemize}
	\item Let $v\in\mathbb{N}^{\Counters'}$
	be such that $(\ellinit',\mathbf{0}_{\Counters'})\transCM^*_N(\ellinit,v)$, and $(\ellinit,v)$ is the first configuration where $\ellinit$ is visited. The
	execution is thus of the form  $(\ellinit',\mathbf{0}_{\Counters'})\transCM^*_N(\ellinit',w)\transCM^*_N(\ellinit,v)$, with $(\ellinit',w)$ the last time $\ellinit'$ is
	visited. We have stated above that  $w(\cpt)\leq 2^{2^i}$. Then, we have that $(\ellinit',\mathbf{0}_{\Counters'})
	\transCM^*_N(\ellinit',w)\transCM_N(\ell_a,w)\transCM^*_N(\ell_b,v)\transCM_N(\ellinit,v)$, and by \cref{prop:cover:proof:expspace-hard:rstinc}, $v=v_0$. 	
	\item Let now $v_k,v_{k+1}\in\mathbb{N}^{\Counters'}$ be such that $(\ellinit',\mathbf{0}_{\Counters'})\transCM^*_N(\ellinit,v_k)\transCM^*_N(\ellinit,v_{k+1})$,
	and $v_k$ and $v_{k+1}$ are respectively the $k^\textrm{th}$ and the $(k+1)^\textrm{th}$ time that $\ellinit$ is visited, for some $k\geq 0$. Assume that 
	$v_k=v_0$. We have $(\ellinit, v_k)\transCM^*_N(\ell,v)\transCM_N(\ellinit',v)\transCM^*_N(\ellinit',\overline v)\transCM_N(\ell_a,\overline v)\transCM_N^*(\ell_b,v_{k+1})
	\transCM_N(\ellinit, v_{k+1})$. Since the \testfreeCM~$M$ is 2EXP-bounded, and $v_k=v_0$, we obtain that for all $\cpt\in\Counters=Y_n$, 
	$v(\cpt)\leq 2^{2^n}$. For all $0\leq i<n$, for all $\cpt\in Y_i\cup\overline Y_i$, $v(\cpt)=v_0(\cpt)$, then for all $0\leq i\leq n$, for all $\cpt\in Y_i\cup\overline Y_i$,
	$v(\cpt)\leq 2^{2^i}$. Then, as proved above, $\overline v(\cpt)\leq 2^{2^i}$ for all $0\leq i\leq n$, for all $\cpt\in Y_i\cup\overline Y_i$. By \cref{prop:cover:proof:expspace-hard:rstinc}, $v'=v_0$.
	\end{itemize}
	Consider now the execution $(\ellinit',\mathbf{0}_{\Counters'})\transNbCM^*_N(\ellinit,v)\transNbCM^*_N(\ell_f,v')$, where $(\ellinit,v)$ is the last time the location
	$\ellinit$ is visited. Then, as proved above, $v=v_0$. From the execution $(\ellinit,v)\transNbCM^*_N(\ell_f,v')$, we can deduce an execution 
	$(\ellinit, v_{|\Counters})\transCM^*_M (\ell_f, v'_{|\Counters})$. Since $v=v_0$ and for all $\cpt\in\Counters=Y_n$, $v(\cpt)=0$, we can conclude the proof.
%	
%	We proved that for all reachable configuration $(\ellinit, v)$, $v = v_0$, and then by \cref{lemma:cover:expspace-hard:m'-to-m}, for a reachable configuration $(\ell, v')$ in $N$, if $\ell \in \Loc$, then there is a path from $(\ellinit, \mathbf{0}_\Counters)$ to $(\ell_f, v_{|\Counters})$ in $M$, which concludes the proof.
\end{proof}

The two previous lemmas prove that the reduction is sound and complete. By \cref{th:expspace-hard-lipton}, we proved the \Expspace-hardness of the problem, and so \cref{th:expspace-hard}.

%%%%%%%%%%%%%%%%%%%%

\section{Proofs of~\cref{sec:cover-rdv-protocols}}
In this section, we present proofs omitted in \cref{sec:cover-rdv-protocols}.
\subsection{Proof of \cref{cor:ccover-expspace}}}
We present here the proof of \cref{cor:ccover-expspace}. The two lemmas of this subsection prove the soundness and completeness of the reduction presented in \cref{subsec:rdv-to-nrcm}. Put together with \cref{thm:cover-nbcm-in-expspace}, we prove \cref{cor:ccover-expspace}.
\begin{lemma}
	Let $C_0 \in \Init$, $C_f \geq C_F$. If $C_0 \arrowP{\PP}^* C_f$, then there exists $v\in\mathbb{N}^Q$ such that $(\ellinit, \mathbf{0}_\Counters)\transNbCM^*(\ell_f, v)$.
	\end{lemma}

\begin{proof}

For all $q\in Q$, we let $v_q(q)=1$ and $v_q(q')=0$ for all $q'\in \Counters$ such that $q'\neq q$. Let $n=||C_0||=C_0(\qinit)$, and let $C_0C_1\cdots C_mC_f$ be the configurations visited in \PP. Then, applying
the transition $(\ellinit, \inc{\qinit}, \ellinit)$, we get $(\ellinit, \mathbf{0}_\Counters)\transNbCM (\ellinit, v^1) \transNbCM \dots \transNbCM (\ellinit, v^n)$ with $v_0 = v^n$ and $v_0(\qinit)=n$ and $v_0(\cpt)=0$ for all $\cpt\neq \qinit$. Let $i\geq 0$ and assume that $(\ellinit,\mathbf{0}_\Counters)\transNbCM^*(\ellinit, C_i)$. We show that
$(\ellinit, C_i)\transNbCM^*(\ellinit, C_{i+1})$. 
\begin{itemize}
\item If $C_i\arrowPlab{\PP}{m} C_{i+1}$, let $t=(q_1,!m,q'_1), t'=(q_2, ?m, q'_2)\in T$ such that $C_i(q_1)>0$, $C_i(q_2)>0$, $C_i(q_1)+C_i(q_2)\geq 2$, and $C_{i+1}=
C_i - \mset{q_1,q_2}+\mset{q'_1,q'_2}$. Then $(\ellinit, C_i)\transNbCM
(\ell_{(t,t')}^1, v_i^1)\transNbCM(\ell_{(t,t')}^2, v_i^2)\transNbCM(\ell_{(t,t')}^3, v_i^3)\transNbCM(\ellinit, v_i^4)$, with $v_i^1= C_i - v_{q_1}$, $v_i^2=v_i ^1 - v_{q_2}$,
$v_i^3 = v_i^2 + v_{q'_1}$, $v_i^4 = v_i^3+v_{q'_2}$. Observe that $v_i^4=C_{i+1}$ and then $(\ellinit, C_i)\transNbCM^*(\ellinit, C_{i+1})$.
\item If $C_i\arrowPlab{\PP}{\tau} C_{i+1}$, let $t=(q,\tau,q')$ such that $C_i(q)>0$ and $C_{i+1}=C_i-\mset{q}+\mset{q'}$. Then,  $(\ellinit, C_i)\transNbCM (\ell_q, v_i^1)\transNbCM (\ellinit, v_i^2)$ with $v_i^1=C_i- v_q$ and $v_i^2 = v_i^1+ v_{q'}$. Observe that $v_i^2 = C_{i+1}$, then $(\ellinit, C_i)\transNbCM^*(\ellinit, C_{i+1})$. 

\item If $C_i\arrowPlab{\PP}{\nb{m}} C_{i+1}$, let $t=(q,!m,q')$ such that $C_{i+1}=C_i-\mset{q}+\mset{q'}$, and $\Read{m} = \{q_1,\dots, q_k\}$. Then $C_i(p_j)=0$ for all $1\leq j\leq k$. We then have that $(\ellinit, C_i)\transNbCM (\ell_t, v_i^1)\transNbCM (\ell_{t,q_1}^m, v_i^1)\transNbCM\cdots \transNbCM(\ell_{t,q_k}^m, v_i^1) \transNbCM (\ellinit, v_i^2)$ with $v_i^1= C_i - v_q$ and $v_i^2= v_i^1 + v_{q'}$. 
Indeed, $v_i^1(q_j)=0$ for all $q_j\in \Read{m}$, so the transitions 
$(\ell^m_{t,q_j},
\nbdec{q_{j+1}}), \ell^m_{t,q_{j+1}})$ do not change the value of the counters. Hence, $v_i^2= C_{i+1}$ and $(\ellinit, C_i)\transNbCM^* (\ellinit, C_{i+1})$. 
\end{itemize}

So we know that $(\ellinit, \mathbf{0}_\Counters)\transNbCM^* (\ellinit, C_f)$. 
Moreover, since $C_f \geq C_F$, it holds that $C_f \geq v_{\mathbf{q}_1} + v_{\mathbf{q}_2} + \dots + v_{\mathbf{q}_s}$. Then $(\ellinit, C_f)\transNbCM^s (\ell_f, v)$ with $v=C_f-(v_{\mathbf{q}_1} + v_{\mathbf{q}_2} + \dots + v_{\mathbf{q}_s})$. 
\end{proof}

%\begin{sketch-proof}
%	Starting from $(\ellinit, \mathbf{0})$, one can mimick the execution between $C_0$ and $C_f$ in the following manner: at the begining, increment the counter $q_0$ as many times as the number of processes participating in the path from $C_0$ to $C_f$. Then, for every $C_i \arrowPlab{}{} C_{i+1}$ in the path of the rendez-vous protocol, we build a path in the \NBCM~such that $(\ell_0, v_i) \transNbCM^+ (\ell_0, v_{i+1})$ where $v_i = C_i$ and $v_{i+1} = C_{i+1}$\lug{à revoir en fonction de comment est définie une configuration}. At step $i$, if a rendez-vous occurs between states $q$ and $p$ leading to states $q'$ and $p'$, take the corresponding cycle which decrements $q$ and then $p$, and increments $p'$ and then $q'$. If it is a non-blocking sending which occurs from $q$ to $q'$ with a message $a$, take the corresponding cycle which decrements $q-1$ and does a non-blocking decrement for every state $p\in R(a)$, by hypothesis, all those counters are equal to $0$. If a process takes an internal transition between states $q$ and $q'$, then take the corresponding cycle starting from $\ell_0$ which decrements counter $q$ and increments counter $q'$. When the execution of the \NBCM~reaches $(\ell_0, v_f)$, then the execution takes its final step $(\ell_0, v_f) \transNbCM^{q_f-1} (\ell_f,v)$ which is possible because $v_f = C_f$ and so $v_f(q_f) > 0$. 
%\end{sketch-proof}

\begin{lemma}
	Let $v\in\nat^Q$. If $(\ellinit, \mathbf{0}_\Counters)\transNbCM^*(\ell_f, v)$, then there exists $C_0 \in \Init$, $C_f \geq C_F$ such that $C_0 \arrowPlab{\PP}{}^* C_f$.
\end{lemma}
\begin{proof}
	%Observe first that from $(\ell_0,\mathbf{0}_\Counters)$ the only transitions that can be taken in $M$ are $(\ell_0,\nop,\ell_0)$ and $(\ell_0, q_0++, \ell_0)$. 
	%Moreover, since we assume that $(\ell_0, \mathbf{0}_\Counters)\transNbCM^*(\ell_f, v)$, the transition $(\ell_0, q_0++, \ell_0)$ has been taken at least once.
	%Hence the execution of $M$ looks like $(\ell_0, \mathbf{0}_\Counters)\transNbCM^*(\ell_0, \mathbf{0}_\Counters)\transNbCM(\ell_0, v_0)\transNbCM^*(\ell_f,v)$, with $v_0=v_{q_0}$. 
	Let $(\ellinit, v_0), (\ellinit, v_1) \dots (\ellinit, v_n)$ be the projection of the execution of $M$ on $\{\ellinit\} \times \mathbb{N}^\Counters$.
	%, starting from the first occurrence of $(\ell_0, v_0)$.
	 We prove that, for all 
	$0\leq i\leq n$, there exists $C_0 \in \Init$, and $C\geq v_i$ 
	%(i.e for all $q \in Q$, $C_n(q) \geq v_n(q)$) 
	such that $C_0 \arrowP{\PP}^* C$. For $i = 0$, we let $C_0$ be the empty multiset, and the property is trivially true. Let $0\leq i < n$, and assume that there exists $C_0 \in \Init$, $C\geq v_i$ such that $C_0 \arrowP{\PP}^* C$. 
	%We prove the existence of $C_{n+1}\succeq v_{n+1}$ depending on what happens in the \NBCM~execution between $(\ell_0, v_n)$ and $(\ell_0, v_{n+1})$.
	\begin{itemize}
		%\item If $(\ellinit, v_i)\xtransNbCM{\delta} (\ellinit, v_{i+1})$ with $v_i=v_{i+1}$ (i.e. $\delta=(\ellinit, \nop, \ellinit)$), we let $C'=C$ and
		%$C_0\arrowP{\PP}^*C'$ with $C' \geq v_{i+1}$.
		% if the step between $(\ell_0, v_n)$ and $(\ell_0, v_{n+1})$ is a restore set, then the induction is trivial;
	
		\item \label{proof:rdvToCM:enum:rdv}If $(\ellinit, v_i)\xtransNbCM{\delta}(\ellinit, v_{i+1})$ with $\delta=(\ellinit, \inc{\qinit}, \ellinit)$, then $v_{i+1} = v_i +v_{\qinit}$. The
		execution $C_0\arrowP{\PP}^* C$ built so far cannot be extended as it is, since it might not include enough processes. Let
		$N$ be such that $C_0\arrowP{\PP} C_1\arrowP{\PP} \dots \arrowP{\PP} C_N = C$, and let $C'_0\in\Init$ with $C'_0(\qinit)=C_0(\qinit)+N+1$. We build, for all $0\leq j \leq N$, a configuration 
		$C'_j$ such that $C'_0\arrowP{\PP}^j C'_j$, 
		$C'_j\geq C_j$ and $C'_j(\qinit)>C_j(\qinit)+N-j$. For $j=0$ it is trivial. Assume now that, for $0\leq j < N$, $C'_j\geq C_j$ and that $C'_j(\qinit) > C_j(\qinit)+N-j$.
		
		If $C_j\arrowPlab{\PP}{m} C_{j+1}$ for $m\in\Sigma$, with $t_1=(q_1,!m, q'_1)$ and $t_2=(q_2,?m,q'_2)$. Then, $C_{j+1}=C_j - \mset{q_1,q_2} + \mset{q'_1,q'_2}$. Moreover, $C'_j(q_1) \geq C_j(q_1)>0$ and $C'_j(q_2) \geq C_j(q_2) >0$ and $C'_j(q_1) + C'_j(q_2)\geq C_j(q_1) + C_j(q_2) \geq 2$. We let
		$C'_{j+1} = C'_j - \mset{q_1,q_2} + \mset{q'_1,q'_2}$, and $C'_j\arrowPlab{\PP}{m} C'_{j+1}$. It is easy to see that $C'_{j+1}\geq C_{j+1}$. Moreover,
		$C'_{j+1}(\qinit) > C_{j+1}(\qinit) +N -j > C_{j+1} + N -j-1$. 
		
		If $C_j\arrowPlab{\PP}{\nb{m}} C_{j+1}$ and 
		for all $q\in\Read{m}$, $C'_j-\mset{q_1}(q)=0$,  with $t=(q_1,!m,q_2)$, (respectively $C_j\arrowPlab{\PP}{\tau} C_{j+1}$ with $t=(q_1,\tau,q_2)$), we let 
		$C'_{j+1}=C'_j - \mset{q_1}+\mset{q_2}$, and $C'_j\arrowPlab{\PP}{\nb{m}}C'_{j+1}$ (respectively $C'_j\arrowPlab{\PP}{\tau} C'_{j+1}$).
		Again, thanks to the induction hypothesis, we get that $C'_{j+1} \geq C_{j+1}$, and $C'_{j+1} (\qinit)> C_{j+1}(\qinit) + N - j >  C_{j+1}(\qinit) + N - j-1$. 
		
		If now $C_j\arrowPlab{\PP}{\nb{m}} C_{j+1}$, with $t_1=(q_1,!m,q_2)$ and there exists $q'_1\in \Read{m}$ such that $C'_j - \mset{q_1}(q'_1) >0$.
		Let $(q'_1,?m,q'_2)\in T$, and then $C'_{j+1}=C'_j - \mset{q_1,q'_1} + \mset{q_2, q'_2}$. Since $C'_j\geq C_j$, $C'_j(q_1)\geq 1$, and since
		$C'_j-\mset{q_1}(q'_1) >0$, $C'_j(q'_1)\geq 1$ and $C'_j(q_1) + C'_j(q'_1) \geq 2$. Hence, $C'_j\arrowPlab{\PP}{{m}} C'_{j+1}$. We have that $C'_j(q'_1) > C_j(q'_1)$, so $C'_{j+1}(q'_1) \geq C_{j+1}(q'_1)$ and $C'_{j+1}(q)\geq C_{j+1}(q)$ for all other $q\in Q$. Hence $C'_{j+1} > C_{j+1}$. Also, $C_{j+1}(\qinit) = C_j(\qinit) + x$, with $x\in\{0,1\}$. If $q'_1\neq \qinit$, then $C'_{j+1}(\qinit) = C'_j(\qinit) + y$, with $y\geq x$. Hence, since $C'_j(\qinit) > C_j(\qinit) + N - j$,
		we get $C'_{j+1}(\qinit) > C_{j+1}(\qinit) + N - j > C_{j+1}(\qinit) + N -j - 1$. If $q'_1 = \qinit$, then we can see that $C'_{j+1}(\qinit) = C'_j(\qinit) +y$, with
		$x-1\leq y \leq x$. In that case, $C'_{j+1}(\qinit) > C_j(\qinit) + N-j+y\geq C_j(\qinit) + N- j + x-1 \geq C_{j+1}(\qinit) + N-j-1$.  	
		
		So we have built an execution $C'_0 \arrowP{\PP}^* C'_N$ such that $C'_N\geq C_N$ and $C'_N(\qinit) > C_N(\qinit)$. Hence, $C'_N\geq v_{i+1}$.

		%		$(C'_j - \mset{q_1})(\qinit)>0$, then let $(\qinit, ?m, q')\in T$, we let $C'_{j+1}=C'_j - \mset{q_1,\qinit}+\mset{q_2,q'}$. Then $C'_{j+1}\geq C_{j=1}$. If 
%		$\qinit\neq q_1$, then 
%		$C'_{j+1}(\qinit)= (C'_j+\mset{q_2})(\qinit)-1$, and since $C'_j(\qinit) > C_j(\qinit) + N- j$ by induction hypothesis, we get that $C'_{j+1}(\qinit)> C_j(\qinit) + N- j-1$. 
%		If $\qinit=q_1$, then $C'_{j+1}(\qinit)\geq C'_j(\qinit)-2$
				
%		is an increment of the state $\qinit$, then, note $M$ the number of steps between $C_0$ and $C_n$. Consider $C'_0 \in \Init$ such that $C'_0(\qinit) = C_0(\qinit) + M +1$. Then, there exists $C'_{n+1}$ such that there is an execution between $C'_0$ and $C'_{n+1}$ with the same steps as between $C_0$ and $C_f$, but some new receptions (which were non-blocking sendings in the original execution) may occur from $\qinit$. As at most $M$ new receptions occur, $C'_{n+1}(\qinit) = C_n(\qinit) + 1$, which concludes this case;
		\item If $(\ellinit,v_i)\transNbCM (\ell_{(t,t')}^1, v_i^1)\transNbCM (\ell_{(t,t')}^2, v_i^2)\transNbCM (\ell_{(t,t')}^3, v_i^3)\transNbCM (\ellinit, v_{i+1})$, with
		$t= (q_1,!m,q_2)$ and $t'=(q'_1, ?m, q'_2)$, then
		$v_i^1 = v_i - v_{q_1}$, $v_i^2= v_i^1 - v_{q'_1}$, $v_i^3 = v_i^2 + v_{q_2}$, and $v_{i+1} = v_i^3+ v_{q'_2}$. Then by induction hypothesis, 
		$C(q_1)\geq 1$, $C(q'_1)\geq 1$, and $C(q_1) + C(q'_1) \geq 2$. We let $C' = C - \mset{q_1, q'_1} + \mset{q_2, q'_2}$. We have
		$C\arrowPlab{\PP}{m} C'$ and $C' \geq v_{i+1}$. 
		%\item  if a cycle is executed between $(\ellinit, v_n)$ and $(\ellinit, v_{n+1})$ and the cycle corresponds to a rendez-vous, by induction hypothesis, this rendez-vous is feasible from $C_n$ and by executing it, we find $C_{n+1}$ satisfying the property;
		\item If $(\ellinit, v_i)\transNbCM (\ell_q, v_i^1) \transNbCM (\ellinit, v_{i+1})$ with $(q,\tau, q')\in T$ and $v_i^1 = v_i - v_q$ and $v_{i+1} = 
		v_i^1 + v_{q'}$, then by induction hypothesis, $C\geq 1$,
		and if we let $C'=C- \mset{q}+\mset{q'}$, then $C\arrowPlab{\PP}{\tau}C'$, and $C'\geq v_{i+1}$.
		%\item  if a cycle is executed between $(\ellinit, v_n)$ and $(\ellinit, v_{n+1})$ and the cycle corresponds to an internal transition, by induction hypothesis, this transition is feasible from $C_n$ and by executing it, we find $C_{n+1}$ satisfying the property;
		
		\item If $(\ellinit, v_i)\transNbCM (\ell_t, v_i^1)\transNbCM (\ell_{t,p_1}^m, v_i^2)\transNbCM\dots\transNbCM (\ell_{t,p_k}^m, v_i^{k+1})\transNbCM
		(\ellinit, v_{i+1})$ with $t=(q,!m,q')$ and $\Read{m} = \{p_1,\dots,p_k\}$, and $(C-\mset{q})(p)=0$ for all $p\in\Read{m}$. We let 
		$C' = C- \mset{q}+\mset{q'}$, hence $C\arrowPlab{\PP}{\nb{m}} C'$. Moreover, $v_i^1 = v_i - v_q$,
		and, for all $1\leq j <k$, it holds that $v_i^{j+1}(p_j) = \max(0, v_i^j(p_j) - 1)$ and $v_i^{j+1}(p)=v_i^j(p)$ for all $p\neq p_j$. By induction hypothesis, $C\geq v_i$,
		hence $v_i^j(p)=0$ for all $p\in\Read{m}$, for all $1\leq j\leq k+1$. Hence, $v_{i+1} = v_i^{k+1} + v_{q'} = v_i^1 + v_{q'}$, and $C' \geq v_{i+1}$. 

		\item If $(\ellinit, v_i)\transNbCM (\ell_t, v_i^1)\transNbCM (\ell_{t,p_1}^m, v_i^2)\transNbCM\dots\transNbCM (\ell_{t,p_k}^m, v_i^{k+1})\transNbCM
		(\ellinit, v_{i+1})$ with $t=(q,!m,q')$ and $\Read{m} = \{p_1,\dots,p_k\}$, and $(C-\mset{q})(p_j)>0$ for some $p_j\in\Read{m}$. Let 
		$(p_j,?m,p'_j)\in T$ and
		$C' = C - \mset{q,p_j}+\mset{q',p'_j}$. Obviously, $C\arrowPlab{\PP}{m} C'$. It remains to show that $C'\geq v_{i+1}$. This is due to the fact that in 
		the \NRCM{}
		$M$, the counter $p'_j$ will not be incremented, unlike $C(p'_j)$. Moreover, in the protocol \PP, only $p_j$ will lose a process, whereas in $M$,
		other counters corresponding to processes in $\Read{m}$ may be decremented. Formally, by definition and by induction hypothesis, $C-\mset{q} \geq v_i^1$. Also, for all $p\in 
		\Read{m}$, either $v_i^1(p)=v_i^{k+1}(p) = 0$, or $v_i^{k+1}(p) = v_i^1(p)-1$. Remark
		that since $C\geq v_i$, then $C-\mset{q}\geq v_i-v_q = v_i^1$, hence $(C-\mset{q,p_j})(p_j) = (C-\mset{q})(p_j) - 1 \geq v_i^1(p_j)-1$. Also, $(C-\mset{q})(p_j) - 1\geq 0$, hence $(C-\mset{q})(p_j) - 1\geq \max(0,v_i^1(p_j)-1)=v_i^{k+1}(p_j)$. Observe also that, for all $p\neq p_j\in\Read{m}$, if $v_i^1(p)>0$, then $(C-\mset{q,p_j})(p)=
		(C-\mset{q})(p) \geq v_i^1(p) > v_i^{k+1}(p)$. If $v_i^1(p) = 0$, then  $(C-\mset{q,p_j})(p)\geq v_i^1(p)= v_i^{k+1}(p)$. For all other $p\in Q$, 
		$(C-\mset{q,p_j})(p) = (C-\mset{q})(p) \geq v_i^1(p)= v_i^{k+1}(p)$. Hence, $C-\mset{q,p_j} \geq v_i^{k+1}$. By definition, $v_{i+1} = v_i^{k+1} + v_{q'}$.
		Hence, $(C-\mset{q,p_j}+\mset{q',p'_j})(p)\geq v_{i+1}(p)$, for all $p\neq p'_j$, and $(C-\mset{q,p_j}+\mset{q',p'_j})(p'_j)> v_{i+1}(p'_j)$.
		So, $C'> v_{i+1}$.

	\end{itemize}
	
	Now we know that the initial execution of $M$ is: $(\ellinit, \mathbf{0}_\Counters)\transNbCM^\ast(\ellinit, v_n)\transNbCM^\ast (\ell_f, v_f)$ with $v_f = v_n - (v_{\mathbf{q}_1} + v_{\mathbf{q}_2} + \dots + v_{\mathbf{q}_s})$. Thus 
	$v_n>v_{\mathbf{q}_1} + v_{\mathbf{q}_2} + \dots + v_{\mathbf{q}_s}$. We have proved that we can build an initial execution of $P$: $C_0\arrowP{\PP}^*C_n$ and that $C_n\geq v_{\mathbf{q}_1} + v_{\mathbf{q}_2} + \dots + v_{\mathbf{q}_s}$. Hence $C_n \geq C_F$.
\end{proof}

%%%%%%%%%%%%%%%%
\subsection{Proofs of \cref{th:ccover-expspace-complete}}
To prove \cref{th:ccover-expspace-complete}, we shall use \cref{cor:ccover-expspace}~along with the reduction presented in \cref{subsec:nrcm-to-rdv}. If the reduction is sound and complete, it will prove that \Cover~is \Expspace-hard. As \Cover~is a particular instance of the \CCover~problem, this is sufficient to prove \cref{th:ccover-expspace-complete}.
The two lemmas of this subsection prove the soundness and completeness of the reduction presented in \cref{subsec:nrcm-to-rdv}, put together with \cref{th:expspace-hard}, it proves that \Cover~is \Expspace-hard.

%\nastext{j'ai l'impression qu'on n'a pas besoin de cette hypothèse.}
%We can suppose wlog that there is no transition leading to the location $\ellinit$ that modifies the counters (we can add, if needed, a new initial location with an only internal transition leading to the original initial location).
\begin{lemma}
	For all $v\in\mathbb{N}^d$, if $(\ellinit, \mathbf{0}_\Counters)\transCM_M^*(\ell_f, v)$, then there exists $C_0 \in \Init$, $C_f \in \FinalE$ such that $C_0 \arrowPlab{\PP}{}^* C_f$.
\end{lemma}
\begin{proof}
For all $\cpt\in\Counters$, we let $N_\cpt$ be the maximal value taken by $\cpt$ in the initial execution $(\ellinit, \mathbf{0}_\Counters)\transCM^*(\ell_f, v)$, and $N=\Sigma_{\cpt\in\Counters} N_\cpt$. Now, we let $C_0\in \Init\cap C_{N+1}$ be the initial configuration with $N+1$ processes. In the initial execution of $\PP$ that we will build, one of the processes will evolve in the $\PP(M)$ part of the protocol, simulating the execution of the \NRCM, the others will simulate the values of the counters in the execution.

Now, we show by induction on $k$ that, for all $k\geq 0$, if $(\ellinit, \mathbf{0}_\Counters)\transCM^k (\ell, w)$, then $C_0\arrowP{}^* C$, with $C(1_\cpt)=w(\cpt)$
for all $\cpt\in\Counters$, $C(\ell)=1$, $C(\qinit)=N-\Sigma_{\cpt\in\Counters} w(\cpt)$, and $C(s)=0$ for all other $s\in Q$.

$C_0\arrowPlab{}{\nb{L}} C_0^1\arrowPlab{}{\nb{R}} C_0^2$, and $C_0^2(\qinit)=N$, $C_0^2(\ellinit)=1$, and $C_0^2(s)=0$ for all other $s\in Q$. So the property
holds for $k=0$. Suppose now that the property holds for $k\geq 0$ and consider $(\ellinit, \mathbf{0}_\Counters)\transCM^k (\ell,w)\xtransNbCM{\delta} (\ell',w')$.

\begin{itemize}
\item if $\delta=(\ell,\inc{\cpt},\ell')$, then $C\arrowPlab{\PP}{\textrm{inc}_\cpt}C_1$ with $C_1=C-\mset{\ell, \qinit}+\mset{\ell_\delta,q_\cpt}$. Indeed, by induction 
hypothesis, $C(\ell)=1> 0$, and $C(\qinit)>0$, otherwise $\Sigma_{\cpt\in\Counters} w(\cpt)=N$ and $w(\cpt)$ is already the maximal value taken by $\cpt$ so
no increment  of $\cpt$ could have happened at that point of the execution of $M$. We also have $C_1\arrowPlab{\PP}{\overline{\textrm{inc}}_\cpt}C'$, since 
$C_1(\ell_\delta)>0$ and $C_1(q_\cpt)>0$ by construction, and $C'=C_1-\mset{\ell_\delta,q_\cpt}+\mset{\ell', 1_\cpt}$. So $C'(\ell')=1$, for all $\cpt\in\Counters$, 
$C'(1_\cpt)=w'(\cpt)$, and $C'(\qinit)=N-\Sigma_{\cpt\in\Counters} w'(\cpt)$.
\item if $\delta=(\ell,\dec{\cpt},\ell')$, then $C(\ell)=1>0$ and $C(1_\cpt)>0$ since $w(\cpt)>0$. Then $C\arrowPlab{\PP}{\textrm{dec}_\cpt}C_1$ with 
$C_1=C-\mset{\ell,1_\cpt}+\mset{\ell_\delta,q'_\cpt}$. Then $C_1\arrowPlab{\PP}{\overline{\textrm{dec}_\cpt}}C'$, with $C'=C_1-\mset{q'_\cpt, \ell_\delta}+\mset{\qinit, \ell'}$. So $C'(\ell')=1$, $C'(1_\cpt)=C(1_\cpt)-1$, $C'(\qinit)=C(\qinit)+1$.
\item if $\delta=(\ell,\nbdec{\cpt},\ell')$ and $w(\cpt)>0$ then $C\arrowPlab{\PP}{\textrm{nbdec}_\cpt}C'$, and $C'=C-\mset{\ell, 1_\cpt}+\mset{\ell',\qinit}$ and
the case is proved.
\item if  $\delta=(\ell,\nbdec{\cpt},\ell')$ and $w(\cpt)=0$ then by induction hypothesis, $C(1_\cpt)=0$ and $C\arrowPlab{\PP}{\nb{\textrm{nbdec}_\cpt}}C'$,
with $C'=C-\mset{\ell}+\mset{\ell'}$. Then, $C'(1_\cpt)=0=w'(\cpt)$, and $C'(\ell')=1$. 
\item if $\delta=(\ell,\nop,\ell')$, then $C\arrowPlab{\PP}{\tau}C'$, 
avec $C'=C-\mset{\ell}+\mset{\ell'}$. This includes the restore transitions. 

\end{itemize}
Then $C_0\arrowP{}^* C$ with $C(\ell_f)=1$ and $C\in\FinalE$.
\end{proof}

\begin{lemma}\label{lem:correction}
	Let $C_0 \in \Init$, $C_f \in \FinalE$ such that $C_0 \arrowPlab{\PP}{}^* C_f$, then $(\ell_0, \mathbf{0}_\Counters)\transCM^*_M(\ell_f, v)$ for some $v\in\mathbb{N}^\Counters$. 
\end{lemma}
Before proving this lemma we establish the following useful result. 

	\begin{lemma}\label{lem:leader}
		Let $C_0 \in \Init$. For all $C\in \CC$ such that $C_0\arrowPlab{\PP}{}^+ C$, we have $\Sigma_{p\in \{q\} \cup Q_M} C(p)= 1$. 
	\end{lemma}
		
%		\begin{proof}
%	By induction of the execution.  For $C_0$ it is obvious. Let now $C,C'\in\CC$ such that $C_0\arrowPlab{\PP}{}^*C\arrowPlab{\PP}{} C'$. By induction hypothesis, $\Sigma_{p\in \{q\} \cup \PP(M)} C(p)\leq 1$. If $\Sigma_{p\in \{q\} \cup \PP(M)} C(p)=0$, the only possible transitions from $C$ to $C'$ are $C\arrowPlab{\PP}{\nb{L}} C'$, $C\arrowPlab{\PP}{\nb{\overline{!\textrm{inc}_\cpt}}} C'$ or $C\arrowPlab{\PP}{\nb{\overline{!\textrm{dec}_\cpt}}} C'$.
%In the first case, $C'(q)=1$, and for all other $p\in Q_M$, $C'(q)=0$. In the other cases, $\Sigma_{p\in \{q\} \cup \PP(M)} C'(p)=0$.
%If $\Sigma_{p\in \{q\} \cup \PP(M)} C(p)= 1$. 
%	\end{proof}

\begin{proof}[Proof of~\cref{lem:correction}]
	Note $C_0\arrowPlab{}{}C_1\arrowPlab{}{}\dots\arrowPlab{}{}C_n = C_f$. 
	%We deduce the followig from the previous observation: we can suppose that for any $1 \leq i \leq n$, $C_i(q) + \sum_{p \in Q_M} C_i(p)=  1$. 
	Now, thanks to~\cref{lem:leader}, for all $1\leq i\leq n$, we can note $\mathsf{leader}(C_i)$ the unique state $s$ in $\{q\} \cup Q_M$ such that $C_i(s) = 1$. In particular, note that $\mathsf{leader}(C_n) = \ell_f$. We say that a configuration $C$ is
	$M$-compatible if $\mathsf{leader}(C)\in \Loc$. 
	For any $M$-compatible configuration $C\in\CC$, we define the configuration of the \NRCM~$\pi(C_i)=(\mathsf{leader}(C), v)$ with $v=C(1_\cpt)$ for all $\cpt\in\Counters$. 
	
	We let $C_{i_1}\cdots C_{i_k}$ be the projection of $C_0C_1\dots C_n$ onto the $M$-compatible configurations. 
	
	We show by induction on $j$ that: 
	
	$P(j)$: For all $1\leq j\leq k$, $(\ellinit,\mathbf{0}_\Counters)\transNbCM^*_M \pi(C_{i_j})$, and $\Sigma_{\cpt\in\Counters}C_{i_j}(q_\cpt)+C_{i_j}(q'_\cpt)=0$. Moreover, for all $C$ such that $C_0\arrowPlab{\PP}{}^*C\arrowPlab{\PP}{}C_{i_j}$, $\Sigma_{\cpt\in\Counters}C(q_\cpt)+C(q'_\cpt)\leq 1$.
	
	By construction of the protocol, $C_0\arrowPlab{}{\nb{L}} C_1(\arrowPlab{}{L})^k C_2\arrowPlab{}{\nb{R}} C_{i_1}$ for some $k \in \mathbb{N}$. So $\pi(C_{i_1})=(\ellinit, \mathbf{0}_\Counters)$, and for all  $C$ such that $C_0\arrowPlab{\PP}{}^*C\arrowPlab{\PP}{}C_{i_1}$, $\Sigma_{\cpt\in\Counters}C(q_\cpt)+C(q'_\cpt)=0$, so $P(0)$ holds true.
	
	Let now $1\leq j <k$, and suppose that $(\ellinit,\mathbf{0}_\Counters)\transNbCM^*_M \pi(C_{i_j})$, and $\Sigma_{\cpt\in\Counters}C_{i_j}(q_\cpt)+C_{i_j}(q'_\cpt)=0$. We know that $C_{i_j}\arrowPlab{}{}^+C_{i_{j+1}}$.
\begin{itemize}
	\item If there is no $C\in\CC$ such that $C(q)=1$ and $C_{i_j}\arrowPlab{}{}^+C\arrowPlab{}{}^*C_{i_{j+1}}$, the only possible transitions from $C_{i_j}$ are in $T_M$. Let $\pi(C_{i_j})=(\ell,v)$.
	\begin{itemize}
		\item if $C_{i_j}\arrowPlab{}{\textrm{inc}_\cpt}C$ then $C=C_{i_j}-\mset{\ell,\qinit}+\mset{\ell_\delta,q_\cpt}$ for $\delta=(\ell,\inc{\cpt},\ell')\in \Delta_b$. $\Sigma_{\cpt\in\Counters}C(q_\cpt)+C(q'_\cpt)=1$. Note that the
	message $\textrm{inc}_\cpt$ is necessarily received by some process, otherwise $C(q_\cpt)=0$ and $C$ has no successor, which is in contradiction with the fact the the execution reaches $C_f$. Moreover, the only
	possible successor configuration is $C\arrowPlab{}{\overline{\textrm{inc}}_\cpt} C_{i_{j+1}}$, with $C_{i_{j+1}}=C-\mset{q_\cpt, \ell_\delta}+\mset{1_\cpt, \ell'}$. Hence, obviously, $\pi(C_{i_j})\transNbCM\pi(C_{i_{j+1}})$.
	
		\item if $C_{i_j}\arrowPlab{}{\textrm{dec}_\cpt}C$ then $C=C_{i_j}-\mset{\ell,1_\cpt}+\mset{\ell_\delta,q'_\cpt}$ for $\delta=(\ell,\dec{\cpt},\ell')\in \Delta_b$. $\Sigma_{\cpt\in\Counters}C(q_\cpt)+C(q'_\cpt)=1$. Note that the
	message $\textrm{dec}_\cpt$ is necessarily received by some process, otherwise $C(q'_\cpt)=0$ and $C$ has no successor, which is in contradiction with the fact the the execution reaches $C_f$. Besides, $C_{i_j}(1_\cpt)>0$ hence $v(\cpt)>0$. Moreover, the only
	possible successor configuration is $C\arrowPlab{}{\overline{\textrm{dec}}_\cpt} C_{i_{j+1}}$, with $C_{i_{j+1}}=C-\mset{q'_\cpt, \ell_\delta}+\mset{\qinit, \ell'}$. Hence, obviously, $\pi(C_{i_j})\transNbCM\pi(C_{i_{j+1}})$.
	
		\item if $C_{i_j}\arrowPlab{}{{\textrm{nbdec}_\cpt}}C_{i_{j+1}}$ then $C_{i_{j+1}}=C_{i_j}-\mset{\ell,1_\cpt}+\mset{\ell',\qinit}$ for $\delta=(\ell,\nbdec{\cpt},\ell')\in \Delta_{nb}$. $\Sigma_{\cpt\in\Counters}C(q_\cpt)+C(q'_\cpt)=0$. Besides, $C_{i_j}(1_\cpt)>0$ hence $v(\cpt)>0$.  Hence, obviously, $\pi(C_{i_j})\transNbCM\pi(C_{i_{j+1}})$.
	\item if $C_{i_j}\arrowPlab{}{{\mathbf{nb}(\textrm{nbdec}_\cpt)}}C_{i_{j+1}}$ then $C_{i_{j+1}}=C_{i_j}-\mset{\ell}+\mset{\ell'}$ for $\delta=(\ell,\nbdec{\cpt},\ell')\in \Delta_{nb}$. $\Sigma_{\cpt\in\Counters}C(q_\cpt)+C(q'_\cpt)=0$. Besides, $C_{i_j}(1_\cpt)=0$ hence $v(\cpt)=0$.  Hence, obviously, $\pi(C_{i_j})~\xrsquigarrow{{\nb{\dec{\cpt}}}} \pi(C_{i_{j+1}})$.
\item if $C_{i_j}\arrowPlab{}{\tau}C_{i_{j+1}}$ then $C_{i_{j+1}}=C_{i_j}-\mset{\ell}+\mset{\ell'}$ for $\delta=(\ell,\nop,\ell')\in \Delta_{nb}$. $\Sigma_{\cpt\in\Counters}C(q_\cpt)+C(q'_\cpt)=0$. Besides, $C_{i_j}(1_\cpt)=C'_{i_{j+1}}(1_\cpt)$ for
all $\cpt\in\Counters$.  Hence, obviously, $\pi(C_{i_j})\xtransNbCM{\nop}\pi(C_{i_{j+1}})$.
	\end{itemize}
	
	\item Otherwise, let $C$ be the first configuration such that $C(q)=1$ and $C_{i_j}\arrowPlab{}{}^+C\arrowPlab{}{}^*C_{i_{j+1}}$. The transition leading to $C$ is necessarily 
	a transition where the message $L$ has been sent. Remember also that by induction hypothesis, 
	$\Sigma_{\cpt\in\Counters}C_{i_j}(q_\cpt)+C_{i_j}(q'_\cpt)=0$.
	\begin{itemize}
		\item if $C_{i_j}\arrowPlab{}{L}C$, then $C(q)=1$, and by induction hypothesis, $\Sigma_{\cpt\in\Counters}C(q_\cpt)+C(q'_\cpt)=0$. Then the only possible successor configuration is $C\arrowPlab{}{\nb{R}}C_{i_{j+1}}$, with $\Sigma_{\cpt\in\Counters}C_{i_{j+1}}(q_\cpt)+C_{i_{j+1}}(q'_\cpt)=0$, and $\pi(C_{i_{j+1}})=(\ellinit, v)$, so $\pi(C_{i_j})\xtransNbCM{\nop}\pi(C_{i_{j+1}})$, by a restore transition.

		\item if $C_{i_j}\arrowPlab{}{\textrm{inc}_\cpt}C_1\arrowPlab{}{L}C$ then $C_1=C_{i_j}-\mset{\ell,\qinit}+\mset{\ell_\delta,q_\cpt}$ for $\delta=(\ell,\inc{\cpt},\ell')\in \Delta_b$ and $\Sigma_{\cpt\in\Counters}C_1(q_\cpt)+C_1(q'_\cpt)=1$. Now, $C=C_1 - \mset{\ell_\delta, \qinit} +
	\mset{q_\bot, q}$, so $C(q)=1=C(q_\cpt)$, and $\Sigma_{\cpt\in\Counters}C(q_\cpt)+C(q'_\cpt)=1$.
		\begin{itemize}
			\item If $C\arrowPlab{}{R}C_{i_{j+1}}$, then $C_{i_{j+1}} = 
	C - \mset{q,q_\cpt}+\mset{\ellinit, \qinit}$, then $\Sigma_{\cpt\in\Counters}C_{i_{j+1}}(q_\cpt)+C_{i_{j+1}}(q'_\cpt)=0$ and $\pi(C_{i_{j+1}})=(\ellinit, v)$, hence
	$\pi(C_{i_j})\xtransNbCM{\nop}\pi(C_{i_{j+1}})$ by a restore transition. 
			
			\item Now $C(q_\cpt)=1$ so it might be that $C\arrowPlab{}{\nb{\overline{\textrm{inc}_\cpt}}} C'$, with $C'=C - \mset{q_\cpt}+\mset{1_\cpt}$. Here, $\Sigma_{\cpt\in\Counters}C'(q_\cpt)+C'(q'_\cpt)=0$. However, $\mathtt{leader}(C')=\{q\}$ so $C'$ is not $M$-compatible. The only
	possible transition from $C'$ is now $C'\arrowPlab{}{\nb{R}} C_{i_{j+1}}$ with $C_{i_{j+1}}=
	C'-\mset{q}+\mset{\ellinit}$. Hence, $C_{i_{j+1}}(1_\cpt)= C'(1_\cpt)=C_{i_j}(1_\cpt)+1=v(\cpt)+1$, and 
	$C_{i_{j+1}}(1_\cpty)=C'(1_\cpty)=C_{i_j}(1_\cpty)=v(\cpty)$ for all $\cpty\neq\cpt$. 
	So $\pi(C_{i_j})=(\ell,v)\xtransNbCM{\delta} (\ell',v+v_{\cpt})\xtransNbCM{\nop}(\ellinit, v+v_\cpt)=\pi(C_{i_{j+1}})$, the last step being a restore transition. Finally, $\Sigma_{\cpt\in\Counters}C_{i_{j+1}}(q_\cpt)+C_{i_{j+1}}(q'_\cpt)=0$.
		\end{itemize}
	
		\item if $C_{i_j}\arrowPlab{}{\textrm{dec}_\cpt}C_1\arrowPlab{}{L} C$, then $C_1=C_{i_j}-\mset{\ell,1_\cpt}+\mset{\ell_\delta,q'_\cpt}$ for $\delta=(\ell,\dec{\cpt},\ell')\in \Delta_b$ and $\Sigma_{\cpt\in\Counters}C_1(q_\cpt)+C_1(q'_\cpt)=1$. Now, $C=C_1 - \mset{\ell_\delta, \qinit} +
	\mset{q_\bot, q}$, so $C(q)=1=C(q'_\cpt)$, and $\Sigma_{\cpt\in\Counters}C(q_\cpt)+C(q'_\cpt)=1$. Again, two transitions are available:
		\begin{itemize}
			\item If $C\arrowPlab{}{R}C_{i_{j+1}}$, then $C_{i_{j+1}} = 
	C - \mset{q,q'_\cpt}+\mset{\ellinit, \qinit}$, then $\Sigma_{\cpt\in\Counters}C_{i_{j+1}}(q_\cpt)+C_{i_{j+1}}(q'_\cpt)=0$ and $\pi(C_{i_{j+1}})=(\ellinit, v)$, hence
	$\pi(C_{i_j})\xtransNbCM{\nop}\pi(C_{i_{j+1}})$ by a restore transition. 
	
			\item Now $C(q'_\cpt)=1$ so it might be that $C\arrowPlab{}{\nb{\overline{\textrm{dec}_\cpt}}} C'$, with $C'=C - \mset{q'_\cpt}+\mset{\qinit}$. Here, $\Sigma_{\cpt\in\Counters}C'(q_\cpt)+C'(q'_\cpt)=0$. However, $\mathtt{leader}(C')=\{q\}$ so $C'$ is not $M$-compatible. The only
	possible transition from $C'$ is now $C'\arrowPlab{}{\nb{R}} C_{i_{j+1}}$ with $C_{i_{j+1}}=
	C'-\mset{q}+\mset{\ellinit}$. Hence, $C_{i_{j+1}}(1_\cpt)= C'(1_\cpt)=C_{i_j}(1_\cpt)-1=v(\cpt)-1$, and 
	$C_{i_{j+1}}(1_\cpty)=C'(1_\cpty)=C_{i_j}(1_\cpty)=v(\cpty)$ for all $\cpty\neq\cpt$. 
	So $\pi(C_{i_j})=(\ell,v)\xtransNbCM{\delta} (\ell',v-v_{\cpt})\xtransNbCM{\nop}(\ellinit, v+v_\cpt)=\pi(C_{i_{j+1}})$, the last step being a restore transition. Finally, $\Sigma_{\cpt\in\Counters}C_{i_{j+1}}(q_\cpt)+C_{i_{j+1}}(q'_\cpt)=0$.
		\end{itemize}
		
		\item If $C_{i_j}\arrowPlab{}{\nb{\textrm{inc}_\cpt}} C_1$ then, it means that $C_{i_j}(\qinit)=0$. In
	that case, let $\delta=(\ell,\inc{\cpt},\ell')\in \Delta_b$, and $C_1=C_{i_j} -\mset{\ell}+\mset{\ell_\delta}$. Since, by induction hypothesis, $C_1(q_\cpt)=C_{i_j}(\cpt)=0$, 
	the only possible transition from $C_1$ would be $C_1\arrowPlab{}{L}C_{i_{j+1}}$. However,
	 $C_{i_j}(\qinit)=C_1(\qinit)=0$, so this transition is not possible, and $C_1$ is a deadlock
	 configuration, a contradiction with the hypothesis that $C_{i_j}\arrowPlab{}{}C_{i_{j+1}}$.
	 
		 \item If $C_{i_j}\arrowPlab{}{\nb{\textrm{dec}_\cpt}} C_1$ then it means that $C_{i_j}(1_\cpt)=0$. In
	that case, let $\delta=(\ell,\dec{\cpt},\ell')\in \Delta_b$, and $C_1=C_{i_j} -\mset{\ell}+\mset{\ell_\delta}$. Since, by induction hypothesis, $\Sigma_{\cpt\in\Counters}C_1(q_\cpt)+C_1(q'_\cpt) = \Sigma_{\cpt\in\Counters}C_{i_j}(q_\cpt)+C_{i_j}(q'_\cpt) = 0$, the only
	possible transition from $C_1$ is $C_1\arrowPlab{}{L}C$, with $C=C_1 - \mset{\qinit,\ell_\delta} + \mset{q, q_\bot}$. Again, $\Sigma_{\cpt\in\Counters}C(q_\cpt)+C(q'_\cpt) = 0$,
	and $C(\ell)=$ for all $\ell\in Q_M$, so the only possible transition is 
	$C\arrowPlab{}{\nb{R}} C_{i_{j+1}}$. Observe that $C_{i_{j+1}}$ is $M$-compatible,
	with $C_{i_{j+1}}(\ellinit)=1$, and $C_{i_{j+1}}(1_\cpt)=C_{i_j}(1_\cpt)$ for all $\cpt\in\Counters$. Hence $\pi(C_{i_{j+1}})=(\ellinit, v)$, and $\pi(C_{i_j})\xtransNbCM{\nop} \pi(C_{i_{j+1}})$, thanks to a restore transition of $M$.
	\end{itemize}
\end{itemize}
We then have, by $P(k)$, that $(\ellinit,\mathbf{0}_\Counters)\transNbCM^*_M \pi(C_{i_k})$, 
with $C_{i_k}$ $M$-compatible and such that $C_{i_k}\arrowPlab{}{}^* C_f$, and $C_{i_k}$ is the
last $M$-compatible configuration. 
Then, by definition of an $M$-compatible configuration, $C_{i_k}=C_f$, and $\pi(C_{i_k})=(\ell_f,v)$
for some $v\in\mathbb{N}^\Counters$. 
\end{proof}

 \section{Proof of Section \ref{sec:wo}}
 We present here omitted proofs of \cref{sec:wo}.
\subsection{Technical Lemma}
We provide here a lemma which will be useful in different parts of this section.

    \begin{lemma}\label{lem:monotonicity}
 Let $\PP$ be rendez-vous protocol and $C,C' \in \CC$ such that $C=C_0 \arrowP{} C_1 \cdots \arrowP{}
 C_\ell=C'$. Then we have the two following properties.
 \vspace{-0.5em}
 \begin{enumerate}
 \item For all $q \in Q$ verifying
  $C(q)=2.\ell+a$ for some $a \in \nat$, we have $C'(q)\geq a$.\label{it:lem-1}
 \item For all $D_0 \in \CC$ such that $D_0 \geq C_0$, there exist $D_1,\ldots,D_\ell$ such that $D_0 \arrowP{} D_1 \cdots \arrowP{}
 D_\ell$ and $D_i \geq C_i$ for all $1 \leq i \leq \ell$.\label{it:lem-2}
 \end{enumerate}
\end{lemma}

\begin{proof}
According to the semantics associated to (non-blocking) rendez-vous
protocols, each step in the execution from $C$ to $C'$ consumes at
most two processes in each control state $q$, hence the result of the
first item.

Let $C,C' \in \CC$ such that $C \arrowP{} C'$. Let $D \in \CC$ such
that $D \geq C$. We reason by a case analysis on the operation
performed to move from $C$ to $C'$ and show that there exists $D'$
such that $D \arrowP{} D'$ and $D'\geq C'$. (To obtain the final
result, we repeat $k$ times this reasoning).
\begin{itemize}
  \item Assume $C \arrowPlab{\PP}{m}
    C'$ then there exists $(q_1, !m, q_1') \in T$ and
	$(q_2, ?m, q_2')\in T$ such that $C(q_1)>0$ and $C(q_2)>0$ and $C(q_1)+C(q_2)\geq 2$ and
	  $C' = C - \mset{q_1, q_2} + \mset{q_1', q_2'}$. But since $D
      \geq C$, we have as 
      well $D(q_1)>0$ and $D(q_2)>0$ and $D(q_1)+D(q_2)\geq 2$ and as
      a matter of fact $D \arrowPlab{\PP}{m}
    D'$ for $D' = D - \mset{q_1, q_2} + \mset{q_1', q_2'}$. Since $D\geq C$, we
    have $D' \geq C'$.
  \item The case $C \arrowPlab{\PP}{\tau} C'$ can be treated in a
    similar way.
 \item Assume $C \arrowPlab{\PP}{\mathbf{nb}(m)} C'$, then  there exists $(q_1, !m, q_1') \in T$, such 
	that $C(q_1)>0$ and $(C-\mset{q_1})(q_2)=0$ for all $(q_2, ?m, q_2') \in T$ and 
	$C' = C - \mset{q_1} + \mset{q'_1}$. We have as well that
    $D(q_1)>0$. But we need to deal with two cases:
    \begin{enumerate}
    \item If  $(D-\mset{q_1})(q_2)=0$ for all $(q_2, ?m, q_2') \in
      T$. In that case we have $D \arrowPlab{\PP}{\mathbf{nb}(m)} D'$
      for 	$D' = D - \mset{q_1} + \mset{q'_1}$ and $D' \geq C'$.
     \item If  there exists $(q_2, ?m, q_2') \in
      T$ such that  $(D-\mset{q_1})(q_2)>0$. Then we have that   $D \arrowPlab{\PP}{m}
    D'$ for $D' = D - \mset{q_1, q_2} + \mset{q_1', q_2'}$. Note that
    since $(C-\mset{q_1})(q_2)=0$ and $D \geq C$, we have here again
    $D' \geq C'$.
    \end{enumerate}
\end{itemize}
\end{proof}

\subsection{Properties of Consistent Abstract Sets of Configurations}

\subsubsection{Proof of Lemma \ref{lem:interp-cover-check}}

\begin{proof}
 Let $C' \in \Interp{\gamma}$ such that $C' \geq C$. Let $q \in Q$
 such that $C(q)>0$. Then we have $C'(q)>0$. If $q \notin S$, then $q
 \in \starg{\Toks}$ and   $C'(q)=1$ and $C(q)=1$ too.  Furthermore for all $q' \in \starg{\Toks} \setminus\set{q}$ such
  $C(q')=1$, we have that $C'(q')=1$ and $q$ and $q'$ are
  conflict-free. This allows us to conclude that $C \in
  \Interp{\gamma}$. Checking whether $C$ belongs to $\Interp{\gamma}$ can be done in
 polynomial time applying the definition of $\Interp{\cdot}$.
\end{proof}

\subsubsection{Building Configurations from a Consistent Abstract Set}

\begin{lemma}\label{lem:consistent-reach}
  Let $\gamma$ be a consistent abstract set of configurations. Given a
subset of states $U \subseteq Q$, if for all $N \in \nat$ and for all
$q \in U$ there exists $C_q \in \Interp{\gamma}$ and $C'_q \in \CC$ such
  that $C_q \arrowP{}^\ast C'_q$ and $C'_q(q)\geq N$, then for all $N
  \in \nat$, there exists $C \in \Interp{\gamma}$ and $C' \in \CC $ such that $C \arrowP{}^\ast
  C'$  and $C'(q) \geq N$ for all $q \in U$.
\end{lemma}

\begin{proof}
We suppose $\gamma=(S,\Toks)$ and reason by induction on the number of
elements in $U\setminus S$. The base case is obvious.  Indeed assume $U
\setminus S=\emptyset$ and let $N\in \nat$. We define the configuration $C$
such that $C(q)=N$ for all $q \in S$ and $C(q)=0$ for all $q \in
Q\setminus S$. It is clear that $C \in \Interp{\gamma}$ and that $C(q)
\geq N$ for all $q \in U$ (since $U
\setminus S=\emptyset$, we have in fact $U \subseteq S$).

We now assume that the property holds for a set $U$ and we shall see
it holds for $U \cup \set{p}$, $p\notin S$. We assume hence that for all $N \in \nat$ and for all
$q \in U \cup \set{p}$ there exists $C_q \in \Interp{\gamma}$ and $C'_q \in \CC$ such
  that $C_q \arrowP{}^\ast C'_q$ and $C'_q(q)\geq N$. Let $N \in \nat$. By induction
hypothesis, there exists $C_U \in  \Interp{\gamma}$ and $C'_U \in \CC
$ such that $C_U \arrowP{}^\ast C'_U$  and $C_U'(q) \geq N$ for all $q
\in U$. We denote by $\ell_U$ the minimal number of steps in an execution from
$C_U$ to $C'_U$. We will see that that we can build a configuration $C
\in \Interp{\gamma}$ such that $C \arrowP{}^\ast C''_U$ with $C''_U
\geq C_U$ and $C''_U(p) \geq N+2*\ell_U$. Using Lemma \ref{lem:monotonicity}, we will
then have that $C''_U \arrowP{}^\ast C'$ with $C' \geq C'_U$ and $C'(p)
\geq N$. This will allow us to conclude.

We as well know that there exist $C_p \in \Interp{\gamma}$ and $C'_p \in \CC$ such
  that $C_p \arrowP{}^\ast C'_p$ and $C'_p(p)\geq
  N+2*\ell_U+(k*\ell)$. We denote by $\ell_p$ the minimum number of steps in an execution from
$C_p$ to $C'_p$.  We build the configuration $C$ as follows:
  we have $C(q)=C_U(q)+2*\ell_p+(k*\ell)+C_p(q)$ for all $q \in S$,
  and  we have $C(q)=C_p(q)$ for all $q \in \starg{\Toks}$. Note that
  since $C_p \in \Interp{\gamma}$, we have that $C \in
  \Interp{\gamma}$. Furthermore, we have $C \geq C_p$, hence using
  again Lemma \ref{lem:monotonicity}, we know that there exists a
  configuration $C''_p$ such that $C \arrowP{}^\ast C''_p$ and  $C''_p
  \geq C'_p$ (i.e. $C''_p(p) \geq N+2*\ell_U+(k*\ell)$ and
  $C''_p(q) \geq C_U(q)+(k*\ell) + C_p(q)$ for all $q \in S$ by~\cref{lem:monotonicity},\cref{it:lem-1})
  %\nas{Je crois que c'est $C''_p(q) \geq C_U(q)+(k*\ell)+C_p(q)$, non?}.

Having $C_U \in \Interp{\gamma}$, we name $(q_1, m_1) \dots (q_k,
m_k)$ the
tokens in $\Toks$ such that $C_U(q_j) = 1$ for all $1 \leq j \leq
k$, and for all $q \in \starg{\Toks} \setminus \{q_j\}_{1 \leq j \leq
  k}$, $C_U(q) =0$. Since $\gamma$ is consistent, for each $(q_j,
m_j)$ there exists a path 
$(q_{0,j},!m_j,q_{1,j})(q_{1,j},?m_{1,j},q_{2,j})\ldots(q_{\ell_j,j},?m_{\ell_j,j},q_j)$
in $\PP$ such that $q_{0,j}
\in S$ and  such that there
exists $(q'_{i,j},!m_{i,j},q''_{i,j}) \in T$ with $q'_{i,j} \in S$ for
all $1 \leq i \leq \ell_j$. We denote by $\ell = \max_{1 \leq j\leq
  k}(\ell_j)+1$.
  
  Assume there exists $1\leq i\leq j\leq k$ such that $(q_i,m_i),(q_j,m_j)\in\Toks$ and $C_U(q_i)=C_U(q_j)=1$, and $m_i\in\Rec{q_j}$ and $m_j\in \Rec{q_i}$. 
  Since $C_U$ respects $\Interp{\gamma}$,  $q_i$ and $q_j$ are conflict-free: there exist $(q_i,m), (q_j,m')\in\Toks$ such that $m\notin\Rec{q_j}$ and 
  $m'\notin\Rec{q_i}$. Hence, $(q_i,m_i), (q_i, m), (q_j,m_j), (q_j,m')\in\Toks$, and $m\notin\Rec{q_j}$ and $m_j\in\Rec{q_i}$. Therefore, we have 
  $(q_i,m), (q_j,m_j)\in\Toks$ and $m\notin\Rec{q_j}$ and $m_j\in\Rec{q_i}$, which is in contradiction with the fact that $\gamma$ is consistent. Hence, for all 
  $1\leq i\leq j\leq k$, for all $(q_i,m_i), (q_j,m_j)\in\Toks$, $m_i\notin\Rec{q_j}$ and $m_j\notin\Rec{q_i}$. 

%Since $C_U \in \Interp{\gamma}$ is consistent, for all $q_i, q_j$ with $i \ne j$, there exists $(q_i, m)$ and $(q_j, m')$ such that $m \nin \Rec{q_j}$ and $m' \nin \Rec{q_i}$ (there are conflict-free). As $\gamma$ is consistent, for all tokens $(q_i, m_i)$, $(q_j, m_j) \in \Toks$, either $m_i \nin \Rec{q_j}$ and $m_j \nin \Rec{q_i}$ or $m_i \in \Rec{q_j}$ and $m_j \in \Rec{q_i}$. Assume we are in the latter case, then the pair of tokens $(q_i,m), (q_j,m_j)$ is such that $m  \in \Rec{q_j}$ and $m_j \nin \Rec{q_i}$ which is not coherent with $\gamma$'s consistency, as a consequence, if $q_i$ and $q_j$ are conflict-free, for all tokens $(q_i, m_i), (q_j, m_j) \in \Toks$, it holds that $m_i\nin\Rec{q_j}$ and $m_j \nin \Rec{q_i}$. As a consequence, for all $i$, for all $j \ne i$, $m_j \nin \Rec{q_i}$ and $m_i \nin \Rec{q_j}$.
%

  We shall now explain how from $C''_p$ we reach $C''_U$ in $k*\ell$
  steps, i.e. how we put (at least) one token in
  each state $q_j$ such that $q_j \in \starg{\Toks}$ and $C_U(q_j)=1$
  in order to obtain a configuration $C''_U \geq C_U$.  We begin by
  $q_1$. Let a process on $q_{0,1}$ send the message $m_1$ (remember
  that $q_{0,1}$ belongs to $S$) and let $\ell_{1}$ other processes on
  states of $S$ send the messages needed for the process to reach
  $q_1$ following the path
  $(q_{0,1},!m_1,q_{1,1})(q_{1,1},?m_{1,1},q_{2,1})\ldots(q_{\ell_1,1},?m_{\ell_1,1},q_1)$. At this stage, we have that the number of processes in each state $q$ in $S$
  is bigger than $C_U(q)+((k-1)*\ell)$ and we have (at least) one process in
  $q_1$. We proceed similarly to put a process in $q_2$, note that the
 message $m_2$ sent at the beginning of the path cannot be received by the
 process in $q_1$ since, as explained above, $m_2 \notin \Rec{q_1}$.

 We proceed again to put a process in the
 states $q_1$ to $q_K$ and at the end we obtain the configuration
 $C''_U$ with the desired properties.
 %
% \lugtext{ici il y a un problème : $q_1$ and $q_2$ conflict-free -> il existe deux tokens tels que ...
% 
%les tokens "conflict-free" pour $q_1$ et $q_2$ et pour $q_2$ et $q_3$ ne sont pas forcément les mêmes, c'est à dire que quand on veut mettre quelqu'un sur $q_3$, si on a "choisit" le mauvais tokens, ce n'est pas (directement) vrai que le message ne sera reçu ni par un processus sur $q_1$ ni sur $q_2$.
%
%
%En fait ça l'est mais ça l'est par construction de l'algorithme. on ne peut donc pas le prouver tout de suite. ou alors il faut changer la définition de conflict-free.. sinon il faut changer l'énoncé de ce lemme et ne parler que des états dans $S$, ce lemme est utilisé qu'avec la partie qui parle des états dans $S$ et je reprouve la partie sur les états tokens plus tard grâce à l'algo }
\end{proof}

\subsection{Proof of Lemma \ref{lem:function-F}}

In this subsection, the different items of Lemma \ref{lem:function-F} have been separated in distinct lemmas.

\begin{lemma}\label{lem:F-consistent}
 	$F(\gamma)$ is consistent and can be computed in polynomial time
  for all consistent $\gamma \in \Gamma$.
\end{lemma}

\begin{proof}
  The fact that $F(\gamma)$ can be computed in polynomial time is a direct consequence of the definition of $F$ (see \cref{table:F,table2:F}).
  
% TBD (direct from the definition of the function $F$, rules 3.b and 5.b).
  Assume $\gamma = (S,\Toks) \in \Gamma$ to be consistent. Note $(S'', \Toks'')$ the intermediate sets computed during the computation of $F(\gamma)$, and note $F(\gamma) = (S', \Toks')$.
  
  To prove that $F(\gamma)$ is consistent, we need to argue that (1) for all $(q, m) \in \Toks'' \setminus \Toks$, there exists a finite sequence of transitions $(q_0, a_0, q_1) \dots (q_k, a_k, q)$ such that $q_0 \in S$, and $a_0 = !m$ and for all $1 \leq i\leq k$, we have that $a_i = ?m_i$ and that there exists $(q'_i, !m_i, q'_{i+1}) \in T$ with $q'_i \in S$, and (2) for all $(q,m), (q',m') \in \Toks'$ either $m\in\Rec{q'}$ and $m'\in\Rec{q}$ or $m\notin\Rec{q'}$ and $m'\notin\Rec{q}$.

  We start by proving property (1).
  If $(q, m)$ has been added to $\Toks''$ with rule \ref{ccover-wo-F-cond-newtok}, then by construction, there exists $p \in S$ such that $(p, !a, p') \in T$, and $(q, m) = (p', a)$. The sequence of transition is the single transition is $(p, !a, q)$. 
  %As $S\subseteq S'$, it concludes this case.
  
  If $(q, m)$ has been added to $\Toks''$ with rule \ref{ccover-wo-F-cond-tok-step}, then there exists $(q',m) \in \Toks$, and $(q', ?a, q)$ with $m \ne a$. Furthermore, $m \in \Rec{q}$ and there exists $(p, !a,p') \in T$ with $p \in S$. By hypothesis, $\gamma$ is consistent, hence there exists a finite sequence of transitions $(q_0, q_0, q_1) \dots (q_k, a_k, q')$ such that $q_0 \in S$, and $a_0 = !m$ and for all $1 \leq i\leq k$, we have that $a_i = ?m_i$ and that there exists $(q'_i, !m_i, q'_{i+1}) \in T$ with $q'_i \in S$. By completing this sequence with transition $(q', ?a, q)$ we get an appropriate finite sequence of transitions. 
  %As $S\subseteq S'$, it concludes the proof.
  
  It remains to prove property (2). Assume there exists $(q, m), (q',m') \in \Toks'$ such that $m \in \Rec{q'}$ and $m' \notin \Rec{q}$, then as $\Toks' \subseteq \Toks''$, $(q, m), (q',m') \in \Toks''$. By condition \ref{ccover-wo-F-cond-2toks-1}, $q \in S'$, therefore, as $\Toks' = \{(p, a) \in \Toks'' \mid p \notin S'\}$, we have that $(q, m) \notin \Toks'$, and we reached a contradiction.
\end{proof}

\begin{lemma}\label{lem:F-increase}
 	If $(S',\Toks')=F(S,\Toks)$ then  $S \subsetneq S'$ or
    $\Toks \subseteq \Toks'$.
\end{lemma}

\begin{proof}
  From the construction of $F$ (see \cref{table:F,table2:F}), we have $S \subseteq S'' \subseteq S'$.

  Assume now that $S=S'$. First note that $\Toks \subseteq \Toks''$ (see Table \ref{table:F}) and that $\starg{\Toks} \cap S=\emptyset$. But $\Toks'=\set{(q,m) \in \Toks'' \mid q \not\in S'}=\set{(q,m) \in \Toks'' \mid q \not\in S}$. Hence the elements that are removed from $\Toks''$ to obtain $\Toks'$ are not elements of $\Toks$. Consequently  $\Toks \subseteq \Toks'$.
\end{proof}

\begin{lemma}\label{lem:abstract-soundness}
 For all consistent $\gamma \in \Gamma$, if $C \in \Interp{\gamma}$ and $C \arrowP{}C'$ then $C' \in \Interp{F(\gamma)}$.
\end{lemma}

\begin{proof}
Let $\gamma = (S,\Toks)\in\Gamma$ be a consistent abstract set of configurations, and $C \in \CC$ such that $C \in  \Interp{\gamma}$ and $C \arrowP{} C'$. Note $F(\gamma) = (S', \Toks')$ and $\gamma' = (S'', \Toks'')$ the intermediate sets used to compute $F(\gamma)$.
We will first prove that for all state $q$ such that $C'(q) > 0$, $q \in S'$ or $q \in \mst(\Toks')$, and then we will prove that for all states $q$ such that $q \in \mst(\Toks')$ and $C'(q)>0$, $C'(q) = 1$ and for all other state $p\in \mst(\Toks')$ such that $C'(p) >0$, $p$ and $q$ are conflict-free. 

 Observe that $S \subseteq S'' \subseteq S'$, $\Toks \subseteq \Toks'' $, and $\mst(\Toks'') \subseteq \mst(\Toks') \cup S'$. 
 %We shall prove that for all $q \in Q$ such that $C'(q) > 0$, either (1) $q \in S'$ or (2) $q\in \mst(\Toks')$ and $C'(q) = 1$ and for all $q' \in \mst(\Toks') \setminus \{q\}$ such that $C'(q') = 1$, we have that $q$ and $q'$ are conflict-free. 

First, let us prove that for every state $q$ such that $C'(q)>0$, it holds that $q \in S' \cup \mst(\Toks')$. Note that for all $q$ such that $C(q) > 0$, because $C$ respects $\gamma$, $q \in \mst(\Toks) \cup S$. As $\mst(\Toks) \cup S \subseteq \mst(\Toks') \cup S'$, the property holds for $q$.
Hence, we only need to consider states $q$ such that $C(q) = 0$ and $C'(q) > 0$. If $C \arrowPlab{}{\tau} C'$ then $q$ is such that there exists $(q', \tau, q) \in T$, $q'$ is therefore an active state and so $q' \in S$, (recall that $\Toks \subseteq Q_W \times \Sigma$). Hence, $q$ should be added to $\mst(\Toks'') \cup S''$ by condition \ref{ccover-wo-F-cond-internal}. As $\mst(\Toks'') \cup S'' \subseteq \mst(\Toks') \cup S'$, it concludes this case. If $C \arrowPlab{}{\nb{a}} C'$ then $q$ is such that there exists $(q', !a, q) \in T$, with $q'$ an active state. With the same argument, $q' \in S$ and so $q$ should be added to $\mst(\Toks'') \cup S''$ by condition \ref{ccover-wo-F-cond-send-S}~or \ref{ccover-wo-F-cond-newtok}. If $C \arrowPlab{}{a} C'$, then $q$ is either a state such that $(q', !a, q) \in T$ and the argument is the same as in the previous case, or it is a state such that $(q', ?a, q) \in T$, and it should be added to $\mst(\Toks'')\cup S''$ by condition \ref{ccover-wo-F-cond-reception-S}, \ref{ccover-wo-F-cond-tok-end}, or \ref{ccover-wo-F-cond-tok-step}. Therefore, we proved that for all state $q$ such that $C'(q) >0$, it holds that $q \in \mst(\Toks') \cup S'$.

It remains to prove that if $q \in \mst(\Toks)$, then $C'(q) = 1$ and for all $q' \in \mst(\Toks') \setminus \{q\}$ such that $C'(q') = 1$, we have that $q$ and $q'$ are conflict-free. Note that if $q \in \mst(\Toks)$ and $C(q) = C'(q) = 1$, then for every state $p$ such that $p \in \mst(\Toks)$ and $C(p) = C'(p) = 1$, it holds that $q$ and $p$ are conflict-free.

Observe that if $C \arrowPlab{}{\tau} C'$, then note $q$ the state such that $(q', \tau ,q)$, it holds that $\{p \mid p \in \mst(\Toks') \textrm{ and } C'(p) > 0\} \subseteq \{p \mid p \in \mst(\Toks) \textrm{ and } C(p) = 1\}$: $q'$ is an active state, $q$ might be in $\mst(\Toks)$ but it is added to $S'' \subseteq S'$ with rule \ref{ccover-wo-F-cond-internal}, and for all other states, $C'(p) = C(p)$. If $p \in \mst(\Toks')$ and $C(p) > 0$, it implies that $C'(p)= C(p) = 1$ and $p\in \mst(\Toks)$ (otherwise $p$ is in $S \subseteq S'$). Hence, there is nothing to do as $C$ respects $\gamma$.

Take now $q \in \mst(\Toks') \setminus \mst(\Toks)$ with $C'(q) > 0$, we shall prove that $C'(q) =1$ and for all $p \in \mst(\Toks')$ and $C'(p) > 0$, $q$ and $p$ are conflict-free. If $q \in \mst(\Toks') \setminus \mst(\Toks)$, it implies that $C(q) = 0$ because $C$ respects $\gamma$. Hence: either (1) $C \arrowPlab{}{\nb{a}} C'$ with transition $(q', !a, q) \in T$, either (2) $C \arrowPlab{}{a} C'$ with transitions $(q_1, !a, q'_1) \in T$ and $(q_2, ?a, q'_2) \in T$ and $q = q'_1$ or $q=q'_2$. In the latter case, we should be careful as we need to prove that $q'_2 \ne q'_1$, otherwise, $C'(q) = 2$.

\textbf{Case (1):} Note that as only one process moves between $C$ and $C'$ and $C(q)= 0$, it is trivial that $C'(q) = 1$. In this first case, as it is a non-blocking request on $a$ between $C$ and $C'$, it holds that: for all $p \in \mst(\Toks)$ such that $C(p) = 1$, $a \notin \Rec{p}$. Take $p  \in\mst(\Toks')$, such that $p\ne q$ and $C'(p) = 1$, then $C'(p) = C(p) = 1$ and so $p \in \mst(\Toks)$, and $a \notin \Rec{p}$. Suppose $(p, m) \in \Toks'$ such that $m \in \Rec{q}$, then we found two tokens in $\Toks'$ such that $m \in \Rec{q}$ and $a \notin \Rec{p}$ which contradicts $F(\gamma)$'s consistency. Hence, $p$ and $q$ are conflict-free.
%Suppose $p$ and $q$ are not conflict-free in $F(\gamma)$, then for all $(p, m), (q, m') \in \Toks'$, $m\in \Rec{q}$ or $m'\in \Rec{p}$. As $F(\gamma)$ is consistent by Lemma \ref{lem:F-consistent}, for all $(p, m), (q, m') \in \Toks'$, $m\in \Rec{q}$ \emph{and} $m'\in \Rec{p}$. 

%As $q \in \mst(\Toks')$ and $q' \in S$, it should be from condition \ref{ccover-wo-F-cond-newtok}~that $(q, a) \in \Toks''$. Note that $\Toks' \subseteq \Toks''$ and if $(q,a) \notin \Toks'$, it should be that $q \in S'$. Hence, $(q,a)\in \Toks'$. Furthermore, for all $p \in \Toks$ such that $C(p) > 1$ and $p \notin q$, it holds that $a \notin \Rec{p}$ as the rendez-vous is not answered. By construction $C'(p) = C(p)$, hence, we found two tokens $(p,m), (q,a)$ such that $a \notin \Rec{p}$ and $m \in \Rec{q}$, which is absurd given $F(\gamma)$'s consistency. Hence for all $p  \mst(\Toks')$, such that $p\ne q$, it holds that $p$ and $q$ are conflict-free.

\textbf{Case (2):} Note that if $q'_2 \in \mst(\Toks')$, then $q_2 \in \mst(\Toks)$ (otherwise, $q'_2$ should be in $S'$ by condition \ref{ccover-wo-F-cond-reception-S}), and note $(q_2, m) \in \Toks$, with $(q'_2, m) \in \Toks'$. Note as well that if $q'_1 \in \mst(\Toks')$, then $a \in \Rec{q'_1}$ (otherwise, $q'_1$ should be in $S'$ by condition \ref{ccover-wo-F-cond-send-S}) and $(q'_1 ,a) \in \Toks'$ by condition \ref{ccover-wo-F-cond-newtok}. Furthermore, if $q'_1 \in \mst(\Toks')$, $q_2 \in \mst(\Toks)$ as well as otherwise $q'_1$ should be added to $S'$ by condition \ref{ccover-wo-F-cond-send-S}.

We first prove that either $q'_1 \in S'$, or $q'_2 \in S'$. For the sake of contradiction, assume this is not the case, then there are three tokens $(q'_1, a), (q_2, m), (q'_2, m) \in \Toks' \subseteq \Toks''$, such that $(q_2, ?a, q'_2) \in T$. From condition \ref{ccover-wo-F-cond-3toks-1}, $q'_1$ should be added to $S'$ and so $(q'_1, a) \notin \Toks'$. Note that, as a consequence $q'_1 \ne q'_2$ or $q'_1 = q'_2 \in S'$. Take $q \in \mst(\Toks') \setminus \mst(\Toks)$ such that $C'(q) >0$, if such a $q$ exists, then $q = q'_1$ or $q = q'_2$ and $q'_1 \ne q'_2$. As a consequence, $C'(q) = 1$ (note that if $q'_1 = q_2$, $C(q_2) = 1$). 

Take $p \in \mst(\Toks') \setminus \{q\}$ such that $C'(p) > 0$, it is left to prove that $q$ and $p$ are conflict-free. If $p \ne q$ and $p \in \mst(\Toks')$, then $C'(p) = C(p)$ (because $q'_1 \in S'$ or $q'_2 \in S'$). Hence, $p \in \mst(\Toks)$ and $C'(p) = 1$.

Assume $q = q'_1$ and assume $q$ and $p$ are not conflict-free. Remember that we justified that $q_2 \in \mst(\Toks)$, and therefore, $C(q_2) = 1$. Hence, either $C'(q_2) = 0$, or $q_2 = q'_2$ and in that case $q_2,q_2' \in S'$ or $q_2' = q_1'$ and then $q_2=q$. In any case, $p \ne q_2$. As $C$ respects $\gamma$, there exists $(p, m_p)$ and $(q_2, m) \in \Toks$ such that $m_p \notin \Rec{q_2}$ and $m \notin \Rec{p}$ ($q_2$ and $p$ are conflict-free). As $p \in\mst(\Toks')$, $(p,m_p) \in \Toks'$ and so $m_p\in \Rec{q}$ or $a \in \Rec{p}$ ($q$ and $p$ are not conflict-free).
% If $m_p \notin \Rec{q}$ or $a \notin \Rec{p}$, then the two tokens $(p,m_p)$ and $(q, a)$ satisifies condition \ref{ccover-wo-F-cond-2toks-1}, and so either $p$ or $q$ should be added to $S'$, which can not be as both states are in $\mst(\Toks')$. 
As $F(\gamma)$ is consistent, $m_p\in \Rec{q}$ and $a \in \Rec{p}$.
Note that $a \ne m_p$ because $a \in \Rec{q_2}$, $a \ne m$ because $m \notin \Rec{p}$, and obviously $m \ne m_p$. Note also that if $m \notin \Rec{q}$, then we found two tokens $(q,a)$ and $(q_2,m)$ in $\Toks'$ such that $a \in \Rec{q_2}$ and $m \notin \Rec{q}$, which contradicts the fact that $F(\gamma)$ is consistent (Lemma \ref{lem:F-consistent}). Hence, $m\in \Rec{q}$.
%and from conditon \ref{ccover-wo-F-cond-2toks-1}, $q$ should be added to $S'$, which is absurd as $q \in \mst(\Toks')$. 
%
Note that even if $q_2$ is added to $S''$, it still is in $\Toks''$.  As $\Toks' \subseteq \Toks''$ we found three tokens $(p, m_p), (q_2,m)$, $(q, a)$ in $\Toks''$, satisfying condition \ref{ccover-wo-F-cond-3toks-2}, and so $p$ should be added to $S'$, which is absurd as $p \in \mst(\Toks')$. We reach a contradiction and so $q$ and $p$ should be conflict-free.

Finally assume $q = q_2'$. If $q = q_2$, then, because $C$ respects $\gamma$, $q$ and $p$ are conflict-free. Otherwise, as $q_2$ is conflict-free with $p$, there exists $(q_2, m )$ and $(p, m_p)$ in $\Toks$ such that $m \notin \Rec{p}$ and $m_p \notin \Rec{q_2}$. Note that $(q,m) \in \Toks''$ from condition \ref{ccover-wo-F-cond-tok-step}~(otherwise, $q \in S''$ which is absurd). Hence, $(q, m) \in \Toks'$ and, as $p \in \mst(\Toks')$, $(p,m_p)$ is conserved from $\Toks$ to $\Toks'$. It remains to show that $m_p \notin \Rec{q}$. Assume this is not the case, then there exists $(p,m_p)$ and $(q,m) \in \Toks'$ such that $m\notin \Rec{p}$ and $m_p\in \Rec{q}$ which is absurd given $F(\gamma)$'s consistency.
As a consequence, $q$ and $p$ are conflict-free.

We managed to prove that for all $q$ such that $C'(q) >0$, $q \in S' \cup \mst(\Toks')$, and if $q \in \mst(\Toks')$, then $C'(q) = 1$ and for all others $p\in \mst(\Toks')$ such that $C'(p) = 1$, $p$ and $q$ are conflict-free.
\end{proof}

\begin{lemma}\label{lem:abstract-correctness}
 For all consistent $\gamma \in \Gamma$, if $C' \in
 \Interp{F(\gamma)}$, then there exists $C'' \in \CC$ and $C \in
 \Interp{\gamma}$ such that $C'' \geq C'$ and $C \arrowP{}^\ast C''$.
\end{lemma}

\begin{proof}
  Let $\gamma$ be a consistent abstract set of configurations and $C'\in \Interp{F(\gamma)}$. We suppose that $\gamma=(S,\Toks)$ and $F(\gamma)=\gamma'=(S',\Toks')$. We will first show that for all $N \in \nat$, for all $q \in S'$ there exists a configuration $C_q \in \Interp{\gamma}$ and a configuration $C_q' \in \CC$ such that $C_q \arrowP{}^\ast C_q'$ and $C'_q(q) \geq N$. This will allow us to rely then on Lemma \ref{lem:consistent-reach} to conclude. 

Take $N \in \nat$ and $q \in S'$, if $q \in S$, then take $C_q \in \Interp{\gamma}$ to be $\mset{N \cdot q}$. Clearly $C_q \in \Interp{F(\gamma)}$, $C_q(q) \geq N$ and $C_q \arrowP{}^\ast C_q$. Now let $q \in S' \setminus S$. Note $(\Toks'', S'')$ the intermediate sets of $F(\gamma$)'s computation.\\

\textbf{Case 1:} $q \in S''$.  As a consequence $q$ was added to $S''$ either by one of the conditions \ref{ccover-wo-F-cond-internal}, \ref{ccover-wo-F-cond-send-S}, \ref{ccover-wo-F-cond-reception-S}~or \ref{ccover-wo-F-cond-tok-end}. In cases \ref{ccover-wo-F-cond-internal}~and \ref{ccover-wo-F-cond-send-S}~when $a \notin \Rec{q}$, note $q'$ the state such that $(q', \tau, q)$ or $(q', !a, q)$, and consider the configuration $C_q = \mset{N \cdot q'}$. By doing $N$ internal transitions or non-blocking requests, we reach $C'_q= \mset{N \cdot q}$. Note that the requests on $a$ are non-blocking as $q' \in Q_A$ and $a \notin \Rec{q}$. $C'_q \in \Interp{F(\gamma)}$.

In cases \ref{ccover-wo-F-cond-send-S} with $a\in \Rec{q}$~and in case \ref{ccover-wo-F-cond-reception-S}, note $(q_1, !a, q_1')$ and $(q_2, ?a, q_2')$ the two transitions realizing the conditions. As a consequence $q_1, q_2 \in S$. Take the configuration $C_q =\mset{N \cdot q_1, N \cdot q_2}$. $C_q \in \Interp{\gamma}$ and by doing $N$ successive rendez-vous on the letter $a$, we reach configuration $C'_q = \mset{N\cdot q'_1} + \mset{N \cdot q'_2}$. $C'_q \in \Interp{F(\gamma)}$, and as $q \in \{q'_1, q'_2\}$, $C'_q(q) \geq N$.

In case \ref{ccover-wo-F-cond-tok-end}, there exists $(q', m) \in \Toks$ such that $(q', ?a, q) \in T$,  $m \notin \Rec{q}$, and there exists $p \in S$ such that $(p, !a,p') \in T$. Remember that $\gamma$ is consistent, and so there exists a finite sequence of transitions $(q_0, !m, q_1) (q_1, a_1, q_2) \dots (q_k, a_k, q')$ such that $q_0 \in S$ and for all $1 \leq i \leq k$, $a_i = ?m_i$ and there exists $(q'_i , !m_i, q''_i) \in T$ with $q'_i \in S$.
Take $C_q = \mset{(N-1) \cdot q_0} + \mset{(N-1) \cdot q'_1 } + \dots + \mset{(N-1) \cdot q'_k} + \mset{N \cdot p} + \mset{q'}$. Clearly $C_q \in \Interp{\gamma}$ as all states except $q'$ are in $S$ and $q' \in \mst(\Toks)$, $C_q(q') = 1$. We shall show how to put 2 processes on $q$ from $C_q$ and then explain how to repeat the steps in order to put $N$. Consider the following execution: $C_q \arrowPlab{}{a} C_1 \arrowPlab{}{x_m} C_2 \arrowPlab{}{m_1} \dots \arrowPlab{}{m_k} C_{k+2} \arrowPlab{}{a} C_{k+3}$. The first rendez-vous on $a$ is made with transitions $(p, !a, p')$ and $(q', ?a, q)$. Then either $m \notin \Rec{p'}$ and $x_m = \nb{m}$, otherwise, $x_m = m$, in any case, the rendez-vous or non-blocking sending is made with transition $(q_0, !m, q_1)$ and  the message is not received by the process on $q$ (because $m \notin \Rec{q}$) and so $C_2 \geq \mset{q} + \mset{q_1}$. Then, each rendez-vous on $m_i$ is made with transitions $(q'_i, !m_i,q''_i)$ and $(q_i, ?m_i, q_{i+1})$ ($q_{k+1} = q'$), . Hence $C_{k+3} \geq \mset{(N-2)\cdot q_0}+ \mset{(N-2) \cdot q'_1 } + \dots + \mset{(N-2) \cdot q'_k} + \mset{(N-2) \cdot p} + \mset{2 \cdot q}$. We can reiterate this execution (without the first rendez-vous on $a$) $N-2$ times to reach a configuration $C'_q$ such that $C'_q \geq \mset{N \cdot q}$.\\

\textbf{Case 2:} $q \notin S''$. Hence, $q$ should be added to $S'$ by one of the conditions \ref{ccover-wo-F-cond-2toks-1}, \ref{ccover-wo-F-cond-3toks-1}, and \ref{ccover-wo-F-cond-3toks-2}.
If it was added with condition \ref{ccover-wo-F-cond-2toks-1}, let $(q_1, m_1), (q_2, m_2) \in \Toks''$ such that $q =q_1$, $m_1 \ne m_2$, $m_2 \notin \Rec{q_1}$ and $m_1 \in \Rec{q_2}$.
From the proof of Lemma \ref{lem:F-consistent}, one can actually observe that all tokens in $\Toks''$ correspond to "feasible" paths regarding states in $S$, i.e there exists a finite sequence of transitions $(p_0, !m_1, p_1) (p_1, a_1, p_2) \dots (p_k, a_k, q_1)$ such that $p_0 \in S$ and for all $1 \leq i \leq k$, $a_i = ?b_i$ and there exists $(p'_i , !b_i, p''_i) \in T$ with $p'_i \in S$. The same such sequence exists for the token $(q_2, m_2)$, we note the sequence $(s_0, !m_2, s_1)\dots (s_\ell, a_\ell, q_2)$ such that $s_0 \in S$ and for all $1 \leq i \leq \ell$, $a_i = ?c_i$ and there exists $(s'_i , !c_i, s''_i) \in T$ with $s'_i \in S$. Take $C_q = \mset{N \cdot p_0} + \mset{N \cdot s_0} + \mset{N p'_1 } + \dots + \mset{N p'_k} + \mset{N \cdot s'_1 } + \dots + \mset{N \cdot s'_\ell}$. Clearly, $C_q \in \Interp{\gamma}$, as all states are in $S$. Consider the following execution: $C_q \arrowPlab{}{\nb{m_1}} C_1 \arrowPlab{}{b_1} \dots \arrowPlab{}{b_k} C_{k+1}$, the non-blocking sending of $m_1$ is made with transition $(p_0, !m_1, p_1)$ and each rendez-vous on letter $b_i$ is made with transitions $(p'_i, !b_i, p_i'')$ and $(p_i, ?b_i, p_{i+1})$ ($p_{k+1} = q_1$). Hence, $C_{k+1}$ is such that $C_{k+1} \geq \mset{q_1}$. From $C_{k+1}$, consider the following execution: $C_{k+1} \arrowPlab{}{x_{m_2}} C_{k+2} \arrowPlab{}{c_1} \dots \arrowPlab{}{c_\ell} C_{k+\ell +2} \arrowPlab{}{m_1}C_{k+\ell +3}$, where $x_{m_2} = \nb{m_2}$ if no process is on a state in $R(m_2)$, or $x_{m_2} = m_2$ otherwise. In any case, as $m_2 \notin \Rec{q_1}$, $C_{k+2} \geq \mset{q_1}$. And each rendez-vous on letter $c_i$ is made with transitions $(s'_i, !c_i, s_i'')$ and $(s_i, ?c_i, s_{i+1})$ ($s_{k+1} = q_2$), the last rendez-vous on $m_1$ is made with transitions $(p_0, !m_1, p_1)$ and $(q_2, ?m_1, q_2')$ (such a $q_2'$ exists as $m_1 \in \Rec{q_2}$). Hence, $C_{k+\ell +3} \geq \mset{p_1} + \mset{q_1}$. By repeating the two sequences of steps (without the first non-blocking sending of $m_1$) $N-1$ times (except for the last time where we don't need to repeat the second execution), we reach a configuration $C'_q$ such that $C'_q\geq \mset{N \cdot q_1}$.

If it was added with condition \ref{ccover-wo-F-cond-3toks-1}, then let $(q_1, m_1), (q_2,m_2), (q_3,m_2) \in \Toks''$ such that $m_1 \ne m_2$ and $(q_2, ?m_1, q_3) \in T$ with $q =q_1$. From the proof of Lemma \ref{lem:F-consistent}, $\Toks''$ is made of "feasible" paths regarding $S$ and so there exists a finite sequence of transitions $(p_0, !m_2, p_1) (p_1, a_1, p_2) \dots (p_k, a_k, q_2)$ such that $p_0 \in S$ and for all $1 \leq i \leq k$, $a_i = ?b_i$ and there exists $(p'_i , !b_i, p''_i) \in T$ with $p'_i \in S$. The same sequence exists for the token $(q_1, m_1)$, we note the sequence $(s_0, !m_1, s_1)\dots (s_\ell, a_\ell, q_1)$ such that $s_0 \in S$ and for all $1 \leq i \leq \ell$, $a_i = ?c_i$ and there exists $(s'_i , !c_i, s''_i) \in T$ with $s'_i \in S$. Take $C_q = \mset{N \cdot p_0} + \mset{N \cdot s_0} + \mset{N p'_1 } + \dots + \mset{N p'_k} + \mset{N \cdot s'_1 } + \dots + \mset{N \cdot s'_\ell}$. Clearly, $C_q \in \Interp{\gamma}$, as all states are in $S$. We do the same execution from $C_q$ to $C_{k+1}$ as in the previous case: $C_q \arrowPlab{}{\nb{m_2}} C_1 \arrowPlab{}{a_1} \dots \arrowPlab{}{a_k} C_{k+1}$. Here $C_{k+1}$ is then such that $C_{k+1} \geq \mset{q_2}$. Then, from $C_{k+1}$ we do the following: $C_{k+1} \arrowPlab{}{m_1} C_{k+2} \arrowPlab{}{c_1} \dots \arrowPlab{}{c_\ell} C_{k+\ell+2} \arrowPlab{}{m_2} C_{k+\ell+3}$: the rendez-vous on letter $m_1$ is made with transitons $(s_0, !m_1, s_1)$ and $(q_2, ?m_1, q_3)$. Then, each rendez-vous on letter $c_i$ is made with transitions $(s'_i, !c_i, s_i'')$ and $(s_i, ?c_i, s_{i+1})$ ($s_{k+1} = q_1$), and the last rendez-vous on letter $m_2$ is made with transitions $(p_0, !m_2, p_1)$ and $(q_3, ?m_2,q_3')$ (such a state $q_3'$ exists as $(q_3, m_2) \in \Toks''$ and so $m_2\in \Rec{q_3}$). Hence, $C_{k+\ell+3}$ is such that $C_{k+\ell +3} \geq \mset{q_1} + \mset{p_1}$. We can repeat the steps from $C_1$ $N-1$ times (except for the last time where we don't need to repeat the second execution), to reach a configuration $C'_q$ such that $C'_q\geq \mset{N \cdot q_1}$.

\nas{pas encore relu condition 8}If it was added with condition \ref{ccover-wo-F-cond-3toks-2}, then let $(q_1, m_1), (q_2, m_2), (q_3, m_3) \in \Toks''$, such that $m_1\ne m_2$, $m_2\ne m_3$, $m_1 \ne m_3$, and $m_1 \notin \Rec{q_2}$, $m_1 \in \Rec{q_3}$, and $m_2 \notin \Rec{q_1}$, $m_2 \in \Rec{q_3}$ and $m_3 \in \Rec{q_2}$ and $m_3 \in \Rec{q_1}$, and $q_1 = q$. Then there exists three finite sequences of transitions $(p_0, !m_1, p_1) (p_1, ?b_1, p_2) \dots (p_k, ?b_k, p_{k+1})$, and $(s_0, !m_2, s_1) (s_1, ?c_1, s_2)$ $ \dots (s_\ell, ?c_k, s_{\ell +1})$, and $(r_0, !m_3, r_1) (r_1, ?d_1, r_2) \dots (r_j, ?d_j, r_{j+1})$ such that $p_{k+1} = q_1$, $s_{\ell +1} = q_2$ and $r_{j+1} = q_3$, and for all messages $a \in \{ b_{i_1}, c_{i_2}, d_{i_3}\}_{1 \leq i_1 \leq k, 1 \leq i_2 \leq \ell, 1 \leq i_3 \leq j} = M$, there exists $q_{a} \in S$ such that $(q_a, !a, q'_a)$. Take $C_q = \mset{Np_0} + \mset{Ns_0} + \mset{Nr_0} + \sum_{a \in M}\mset{Nq_{a}}$. From $C_q$ there exists the following execution: $C_q \arrowPlab{}{\nb{m_1}} C_1 \arrowPlab{}{b_1} \dots \arrowPlab{}{b_k} C_{k +1} $ where the non-blocking sending is made with the transition $(p_0, !m_1, p_1)$ and each rendez-vous with letter $b_i$ is made with transitions $(q_{b_i}, !b_i, q'_{b_i})$ and $(p_i, ?b_i, p_{i+1})$. Hence, $C_{k+1} \geq \mset{q_1}$.
Then, we continue the execution in the following way:
$C_{k+1} \arrowPlab{}{x_{m_2}} C_{k+2} \arrowPlab{}{c_1} \dots \arrowPlab{}{c_\ell} C_{k+ \ell +2} $ where
$x_{m_2} = \nb{m_2}$ if there is no process on $R(m_2)$, and $x_{m_2} = m_2$ otherwise. In any case, the rendez-vous is not answered by a process on state $q_1$ because $m_2 \notin \Rec{q_1}$.
Furthermore, each rendez-vous with letter $c_i$ is made with transitions $(q_{c_i}, !c_i, q'_{c_i})$ and $(s_i, ?c_i, s_{i+1})$. Hence, $C_{k +\ell+2} \geq \mset{q_2} + \mset{q_1}$. From $C_{k+\ell +2}$ we do the following execution: $C_{k+\ell +2} \arrowPlab{}{m_3} C_{k+\ell +3} \arrowPlab{}{d_1} \dots \arrowPlab{}{d_j} C_{k +\ell + j +3}$ where the rendez-vous on letter $m_3$ is made with transitions $(r_0, !m_3, r_1)$ and $(q_2, ?m_3, q_2')$ (this transition exists as $m_3 \in \Rec{q_2}$). Each rendez-vous on $d_i$ is made with transitions $(q_{d_i}, !d_i, q'_{d_i})$ and $(r_i, ?d_i, r_{i+1})$. Hence, the configuration $C_{k+ \ell +j+3}$ is such that $C_{k+\ell +j +3} \geq \mset{q_3} + \mset{q_1}$. Then from $C_{k+\ell +j +3}$: $C_{k+\ell + j +3} \arrowPlab{}{m_1} C_{k+\ell + j +4}$ where the rendez-vous is made with transitions $(p_0, !m_1, p_1)$ and $(q_3, ?m_1, q'_3)$ (this transition exists as $m_1 \in \Rec{q_3}$). By repeating $N-1$ times the execution from configuration $C_1$, we reach a configuration $C'_q$ such that $C'_q(q_1) \geq N$.
\\

Hence, for all $N \in \mathbb{N}$, for all $q \in S'$, there exists $C_q \in \Interp{\gamma}$, such that $C_q\arrowPlab{}{}C'_q$ and $C'_q(q) \geq N$. From Lemma \ref{lem:consistent-reach}, there exists $C'_N$ and $C_N \in \Interp{\gamma}$ such that $C_N \arrowP{}^\ast C'_N$ and for all $q \in S'$, $C_N(q) \geq N$.

Take $C' \in \Interp{F(\gamma)}$, we know how to build for any $N \in \nat$, a configuration $C'_N$ such that $C'_N(q) \geq N$ for all states $q \in S'$ and there exists $C_N \in \Interp{\gamma}$, such that $C_N \arrowPlab{}{}^\ast C'_N$, in particular for $N$ bigger than the maximal value $C'(q)$ for $q \in S'$, $C'_N$ is greater than $C'_N$ on all the states in $S'$.  

To conclude the proof, we need to prove that from a configuration $C'_{N'}$ for a particular $N'$, we can reach a configuration $C''$ such that $C''(q) \geq C'(q)$ for $q \in S' \cup \mst(\Toks')$. As $C'$ respects $F(\gamma)$, remember that for all $q \in \mst(\Toks')$, $C'(q) = 1$. The execution is actually built in the manner of the end of the proof of Lemma \ref{lem:consistent-reach}. 

Note $N_{\max}$ the maximum value for any $C'(q)$. 
We enumerate states $q_1, \dots, q_m$ in $ \mst(\Toks')$ such that $C'(q_i) = 1$. As $C'$ respects $F(\gamma)$, for $i \ne j$, $q_i$ and $q_j$ are conflict free.

From Lemma \ref{lem:F-consistent}, $F(\gamma)$ is consistent, and so we note $(p^j_0, !m^j, p^j_1) $ $(p^j_1, ?m^j_1, p^j_2)$ $ \dots $ $(p^j_{k_j}, ?m^j_{k_j}, p^j_{k_j+1})$ the sequence of transitions associated to state $q_j$ such that: $p^j_{k_j+1} = q_j$, $(q_j, m^j) \in \Toks$ and for all $m^j_i$, there exists $(q_{m^j_i}, !m_i^j, q'_{m^j_i})$ with $q_{m^j_i}\in S'$.
Note that for all $i \ne j$, $q_i$ and $q_j$ are conflict-free and so there exists $(q_i, m), (q_j,m') \in \Toks'$ such that $m \notin \Rec{q_j}$ and $m' \notin \Rec{q_i}$. As $F(\gamma)$ is consistent, it should be the case for all pairs of tokens $(q_i, a), (q_j, a')$. Hence $m^j \notin \Rec{q_i}$ and $m^i \notin \Rec{q_j}$. 

%%Indeed otherwise, one of the two states should have been to $S'$ during the computation of $F(\gamma)$ by condition \ref{ccover-wo-F-cond-2toks-1}.\lug{en faire un lemme? }
%
Note $\ell_j = k_j + 1$.
For $N' =  N_{\max} + \sum_{1\leq j \leq m} \ell_j$, there exists a configuration $C'_{N'}$ such that there exists $C_{N'} \in \Interp{\gamma}$, $C_{N'}\arrowPlab{}{}^*C'_{N'}$, and $C'_{N'}(q) \geq N'$ for all $q \in S'$. In particular, for all $q \in S'$, $C'_{N'}(q) \geq C'(q)  + \sum_{1\leq j \leq m} \ell_j$.

Then, we still have to build an execution leading to a configuration $C''$ such that for all $q \in \mst(\Toks')$, $C''(q) \geq C'(q)$. We then use the defined sequences of transitions for each state $q_j$. With $\ell_1$ processes we can reach a configuration $C_1$ such that $C_1(q_1) \geq 1$: $C_1 \arrowPlab{}{x_{m^1}} C_2 \arrowPlab{}{m_1^1} \dots \arrowPlab{}{m_{k_1}^1} C_{\ell_1+ 1 }$. $x_{m^1} = \nb{m^1}$ if there is no process on $R(m^1)$, and $x_{m^1} = m^1$ otherwise.	Each rendez-vous on $m_i^1$ is made with transitions $(p_i^1, ?m_i^1, p_{i+1}^1)$ and $(q_{m_i^1}, ! m_i^1, q'{m_i^1})$. As a result, for all $q \in S'$, $C_{\ell_1+1}(q) \geq C'(q) +\sum_{2\leq j \leq m} \ell_j$ and $C_{\ell_1 +1}(q_1) \geq 1$. We then do the following execution form $C_{\ell_1 + 1}$: $C_{\ell_1 +1} \arrowPlab{}{x_{m^2}} C_{\ell_1+2} \arrowPlab{}{m_1^2} \dots  \arrowPlab{}{m_{k_2}^2} C_{\ell_1+ \ell_2+ 2 }$. $x_{m^2} = \nb{m^2}$ if there is no process on $R(m^2)$, and $x_{m^2} = m^2$ otherwise. Remember that we argued that $m^2 \notin \Rec{q_1}$, and therefore $C_{\ell_1 + 2}(q_1) \geq C_{\ell_1 +1}(q_1) \geq 1$.
Each rendez-vous on $m_i^2$ is made with transitions $(p_i^2, ?m_i^2, p_{i+1}^2)$ and $(q_{m_i^2}, ! m_i^2, q'{m_i^2})$. As a result, $C_{\ell_1+\ell_2 +2}(q) \geq C'(q) +\sum_{3\leq j \leq m} \ell_j$ for all $q \in S'$ and $C_{\ell_1+ \ell_2 + 2} \geq \mset{q_1} + \mset{q_2}$. We can then repeat the reasoning for each state $q_i$ and so reach a configuration $C''$ such that $C''(q) \geq C'(q)$ for all $q \in S'$ and, $C'' \geq \mset{q_1} + \mset{q_2} + \dots \mset{q_m}$.
We built the following execution: $C_{N'} \arrowP{}^\ast C'_{N'} \arrowP{}^\ast C''$, such that $C''\geq C'$, and $C'_{N'} \in \Interp{\gamma}$.

\end{proof}

\subsection{Proof of Lemma \ref{lem:correct-abstraction}}

  \begin{proof}
Assume that there exists $C_0 \in \Init$ and $C' \geq C$
such that $C_0 \arrowP{} C_1 \arrowP{} \ldots\arrowP{} C_\ell =C'$. Then
using the Lemma \ref{lem:abstract-soundness} iteratively, we get that $C'
\in \Interp{\gamma_\ell}$. From the definition of $F$ and $\Interp{\cdot}$,
one can furthermore easily check that $\Interp{\gamma} \subseteq
\Interp{F(\gamma)}$ for all $\gamma \in \Gamma$. Hence we have
$\Interp{\gamma_\ell} \subseteq \Interp{\gamma_f}$ and $C' \in
\Interp{\gamma_f}$.

Before proving the other direction, we first prove by induction
that for all $i \in \nat$ and for all $D \in \Interp{\gamma_i}$, there
exists $C_0 \in \Init$ and $D' \geq D$ such that $C_0 \arrowP{}^\ast
D'$.  The base case for $i=0$ is obvious. Assume the property holds
for $\gamma_i$ and let us show it is true for $\gamma_{i+1}$. Let $E \in
\Interp{\gamma_{i+1}}$. Since $\gamma_{i+1}=F(\gamma_i)$, using Lemma
\ref{lem:abstract-correctness}, we get that there exists $E' \in \CC$ and $D \in
 \Interp{\gamma_i}$ such that $E' \geq E$ and $D \arrowP{}^\ast
 E'$. By the induction hypothesis, there exist $C_0 \in \Init$ and $D'
 \geq D$ such that $C_0 \arrowP{}^\ast D'$. Using the monotonicity
 property stated in Lemma \ref{lem:monotonicity}, we deduce that there
 exists $E'' \in \CC$ such that $E'' \geq E' \geq E$ and $C_0
 \arrowP{}^\ast D' \arrowP{}^\ast E''$.

 Suppose now that there exists $C'' \in
\Interp{\gamma_f}$ such  that $C'' \geq C$. By the previous
reasoning, we get that there exist $C_0 \in \Init$ and $C' \geq C''
\geq C$
such that $C_0 \arrowP{}^\ast C'$.
\end{proof}

\end{document}